%% file: main.tex
\documentclass[12pt,twoside]{extarticle}
\renewcommand{\baselinestretch} {1.10}

\usepackage{authblk}
\usepackage{amsmath}
\usepackage{amsthm}
\usepackage{amssymb}
\usepackage{amsfonts}
\usepackage{algorithm}
\usepackage{algpseudocode}
\usepackage[top = 1 in, bottom = 1.2 in, left = 1 in, right = 1 in]{geometry}
\usepackage{color}
\usepackage{mathrsfs}
\usepackage[normalem]{ulem}
\usepackage{longtable}
\usepackage{hyperref}
\usepackage{cleveref}
\usepackage{natbib}
\usepackage{dsfont}
\usepackage{graphicx}
\usepackage{caption}
\usepackage{subcaption}
\usepackage{float}
\usepackage{mathtools}
\usepackage{etoolbox}
\usepackage{thmtools}
\usepackage{enumitem}
\usepackage[textsize=tiny]{todonotes}
\usepackage{tikz}
\usepackage{pgfplots}
\pgfplotsset{compat=1.18}

\usepackage{multirow}

\usepackage{esvect}
\usepackage[utf8]{inputenc}
\usepackage{comment}

\def\r#1{\textcolor{red}{#1}}

\newcommand{\PP}{\mathbb{P}}

\newtheorem{theorem}{Theorem}

\newtheorem{lemma}{Lemma}

\newtheorem{proposition}{Proposition}
\newtheorem{corollary}{Corollary}
\newtheorem{assumption}{Assumption}
\declaretheoremstyle[headfont=\bf,bodyfont=\normalfont]{ex}
\declaretheorem[style=ex]{example}
\declaretheoremstyle[bodyfont=\normalfont]{rm}
\declaretheorem[numbered=no,style=rm]{remark}

\DeclareMathOperator*{\argmax}{arg\,max}

\newcommand{\normmm}{{\vert\kern-0.25ex\vert\kern-0.25ex\vert}}
\newcommand{\bignormmm}{{\big\vert\kern-0.25ex\big\vert\kern-0.25ex\big\vert}}
\newcommand{\Bignormmm}{{\Big\vert\kern-0.25ex\Big\vert\kern-0.25ex\Big\vert}}

\makeatletter
\long\def\@makecaption#1#2{
\vskip 0.8ex
\setbox\@tempboxa\hbox{\small {\bf #1:} #2}
\parindent 1.5em  %% How can we use the global value of this???
\dimen0=\hsize
\advance\dimen0 by -3em
\ifdim \wd\@tempboxa >\dimen0
\hbox to \hsize{
\parindent 0em
\hfil 
\parbox{\dimen0}{\def\baselinestretch{0.96}\small
    {\bf #1.} {#2}
  } 
\hfil}
\else \hbox to \hsize{\hfil \box\@tempboxa \hfil}
\fi
}
\makeatother

\usepackage{hyperref}

\input{macros}

\numberwithin{equation}{section}

\begin{document}
% \begin{center}
%     {\bf \LARGE Optimal Semiparametric Dynamic Pricing with Feature Diversity} \\
  
% 	\vspace{1em}
    
%   \vspace{.6em}
% \end{center}
\title{Optimal Semiparametric Dynamic Pricing with Feature Diversity}

\renewcommand\Authfont{\large}
\renewcommand\Affilfont{\small}

\author[1]{Jinhang Chai}
\author[2]{Yaqi Duan}
\author[1]{Jianqing Fan\thanks{His research is supported by NSF Grant DMS-2412029 and the ONR Grant N00014-25-1-2317.}}
\author[3]{Kaizheng Wang}

\affil[1]{Princeton University}
\affil[2]{New York University}
\affil[3]{Columbia University}

\date{}
\maketitle

\vspace{-3em}
\begin{abstract}
We study contextual dynamic pricing under a semiparametric demand model in which the purchase probability is $1-F(p-m(\mathbf{x}))$, where $m(\mathbf{x})$ captures mean utility as a function of product features and buyer covariates, and $F$ is an unknown market-noise distribution. Existing methods either incur suboptimal regret or rely on restrictive structural assumptions. We propose a stagewise greedy pricing algorithm that iteratively refines the estimate of $F$ via local polynomial regression while pricing greedily with current estimates. By exploiting feature diversity, the algorithm reuses endogenous samples collected during exploitation for nonparametric estimation, avoiding costly global random exploration used in prior work.

We establish a general regret bound that applies to any estimator $\hat m$ of the utility function, and derive explicit rates for linear, nonparametric additive, and sparse linear classes of $m$. For the linear class, our regret scales as $T^{\max\{1/2,\,3/(2\beta+1)\}}$, where $\beta$ is the smoothness of $F$ and $T$ is the time horizon. This improves the best known rates for semiparametric contextual pricing and achieves the parametric $\sqrt{T}$ rate when $\beta \ge 5/2$. We further prove a matching lower bound, showing the optimality of our rate, and present numerical experiments that corroborate the theory and demonstrate the practical advantages of iterative refinement.
\end{abstract}

\newcommand{\SHOWPROOFS}{0}

\ifnum\SHOWPROOFS=1
\section{Checklist}
\begin{enumerate}
\item $\beta=1$ case?

The lower bound and upper bound using our algorithm are only valid for $\beta\ge 2$. This is because we imposed the assumption that $\phi'(u)\ge c_{\phi}$ (which is essentially $-F''(1-F)<2(F')^2$). For the upper bound, we can have $T^{3/4}$, same as \cite{fan2022policy}.
\item Stated an upper-bound theorem and its proof. 
\item Modified the initial-stage to remove additional asps, modified the post-smoothing stage.
\item Added log factors/define range coverage(bring defs upfront)/write a proof outline and top-down structure. Being a little sloppy on range estimation.
\item Revised with LLM. To add literature review: Feature diversity by e.g., Hamsa Bastani; Yining Wang's paper.
\item Question: whether to simplify Section 3  or 4; should I replace $m(\bx)$ simply by $m$?
\end{enumerate}
\fi

\section{Introduction}
\label{sec:intro}

Dynamic pricing with demand learning is a central problem in revenue management and online decision-making. A seller repeatedly interacts with heterogeneous customers, observes contextual information, and must choose prices to balance revenue maximization with demand learning. While classical pricing models assume either fully parametric demand or non-contextual settings, modern applications increasingly require flexible models that incorporate rich covariates and unknown demand distributions.

In this paper, we study a \emph{semiparametric contextual pricing} problem in which the purchase probability is given by
\[
\PP(y=1 \mid \bx, p) = 1 - F\bigl(p - m(\bx)\bigr),
\]
where $p$ is the offering price, $m(\bx)$ is the mean utility function given the product and customer covariate vector $\bx$, and $F$ is the market noise distribution; see section \S~\ref{setup} for details. This model naturally generalizes both linear contextual pricing and nonparametric demand learning, capturing heterogeneity in product features and customer valuations while remaining robust to misspecification of the market noise distribution.

Learning under this model poses two fundamental challenges. First, both $m(\cdot)$ and $F$ are tightly coupled and must be learned simultaneously, yet errors in one directly bias estimation of the other. Second, the data used to estimate $F$ are generated endogenously by the pricing policy, inducing distribution shift across time and dependence of the data over time. These challenges further complicate the optimal algorithm design and analysis.

\cite{fan2022policy} initiated the study of this model and established regret guarantees under the $\beta$-smoothness assumptions on $F$. However, their approach relies on a costly global random exploration stage, leading to regret rates that deteriorate as the smoothness parameter increases. More recently, \cite{wang2025tight} developed active-learning-based algorithms that achieve the optimal $T^{3/5}$ regret rate for the special case $\beta=2$ and linear utility. That said, extensions to higher smoothness regimes remain unclear.

Our key insight is that under a \emph{feature diversity} condition—ensuring that covariates are sufficiently well spread—samples collected during exploitation can be systematically reused to refine the estimation of the unknown distribution $F$. A related idea has appeared in \cite{bastani2021mostly} in the context of contextual bandits. Building on this observation, we propose a \emph{stagewise greedy algorithm} that alternates between greedy pricing and localized nonparametric estimation of $F$ via local polynomial regression. By organizing the learning process into geometrically growing stages, the algorithm progressively sharpens the estimate of $F$ while avoiding the cost of explicit global exploration.

\subsection{Related Work}

\paragraph{Contextual pricing with parametric demand.}
A large body of literature studies dynamic pricing with contextual information under parametric demand models. Early contributions focus on linear and generalized linear demand models with known link functions, where regret bounds of order $\sqrt{T}$ or $\log T$ can be achieved using explore--then--commit or greedy policies under appropriate conditions~\citep{broder2012dynamic,qiang2016dynamic}. In the contextual setting, \cite{ban2021personalized} was the first to establish a $\sqrt{dT}$ regret bound. More recently, \cite{chai2024localized} showed that dimension-free regret is achievable via localization and proved the resulting rates to be minimax optimal. The model studied in this paper builds on \cite{javanmard2019dynamic}, which assumes a parametric form for the CDF.

\paragraph{Nonparametric and semiparametric pricing.}
\cite{wang2014close,wang2021multimodal} studied nonparametric demand learning, but without covariates.
In the contextual setting, when the demand function is fully nonparametric, regret typically suffers from the curse of dimensionality~\citep{chen2021nonparametric}. Contextual semiparametric models offer a middle ground by combining low-dimensional structure with nonparametric flexibility. In this setting, \cite{fan2022policy} and \cite{luo2024distribution} proposed algorithms based on a nonparametric kernel estimator and a perturbed linear bandit, respectively, and derived regret bounds depending on the smoothness of $F$. However, their rates are suboptimal. Recently, \cite{wang2025tight} established minimax optimal regret for the special case of $\beta=2$ and linear utility using techniques from active learning. Readers can refer to Table 1 in \cite{wang2025tight} for a detailed comparison of prior works.

\paragraph{Feature diversity and exploration-free learning.}
Several recent works show that sufficiently rich variation in contextual features can significantly reduce the need for explicit exploration in online decision making~\citep{Bietti2018PracticalEA}. In contextual bandits, \cite{bastani2021mostly} formalize a covariate diversity condition under which greedy or mostly greedy policies achieve near-optimal regret. Related guarantees for greedy algorithms are also obtained via smoothed analysis, where small random perturbations of contexts ensure adequate excitation even in adversarial settings~\citep{kannan2018smoothed}.

%\cite{wu2020stochastic}

% Beyond purely parametric models, \cite{krishnamurthy2018semiparametric} study semiparametric contextual bandits and demonstrate that adaptively collected data can be efficiently reused through orthogonalized estimation. 

% \paragraph{Our contribution.}
% Compared to the above works, our approach differs in both algorithmic design and analysis. Rather than relying on a single exploration phase or explicit active learning procedures, we propose an iterative scheme that reuses exploitation data to refine the nonparametric component. This allows us to strictly improve upon existing regret bounds for $\beta > 2$ and recover the parametric $\sqrt{T}$ rate when the noise distribution is sufficiently smooth.

\subsection{Problem Setup} \label{setup}

We consider a dynamic pricing problem under a semiparametric demand model~\citep{fan2022policy}. At each time $t$, a customer with covariate vector $\bx_t\in\RR^d$ arrives ($\bx_t$ can also include product features and their interactions with customer covariates), and the seller posts a price $p_t\in\calP$, where $\calP=[p_{\min},p_{\max}]$ is a fixed interval.
%~\footnote{While the price set can be general, for simplicity, we later take $[p_{\min},p_{\max}]=[0,B]$.}
The customer's valuation is given by
\[
		v_t = m(\bx_t) + \varepsilon_t,
\]
and a purchase occurs whenever $v_t \ge p_t$ so that we only observe $y_t=\mathds{1}(v_t\ge p_t)$. Here, $m(\bx)$ is the mean utility (intrinsic value) function and $\{\varepsilon_t\}_{t=1}^T$ are i.i.d. zero-mean noise variables, independent of $\{\bx_t\}_{t=1}^T$.  Then, the demand function, which is the probability of purchase, is  given by
\begin{equation} \label{demand}
d(\bx,p) = P(y_t= 1 | \bx_t = x, p_t = p)  = 1 - F\!\left(p - m(\bx)\right),
\end{equation}
where $F$ denotes the CDF of the noise $\varepsilon_t$.  In this paper, $F$ is unknown, treated nonparametrically, and assumed to be $\beta$-smooth for some $\beta \ge 2$. %on a compact interval that contains the relevant range of $p - m(\bx)$.

Under the above model, the expected revenue at time $t$ with offer price $p$ is
\[
r(\bx_t,p) = p\bigl(1 - F(p - m(\bx_t))\bigr).
\]
The (clairvoyant) optimal price at time $t$ is then given by
\[
p_t^* \in \argmax_{p\in\calP} r(\bx_t,p).
\]
The mean-utility function $m$ and the noise distribution $F$ are generally unknown and need to be learned dynamically from the empirical sale data.  
%~\footnote{Our theory readily implies the optimal rate when $m$ is known.}
We assume $\{\bx_t\}_{t=1}^T$ are i.i.d.\ draws from a fixed distribution on a compact subset of $\RR^d$ and $\hat{p}_t$ is a pricing policy based on the data before time $t$.  The goal  is to maximize the cumulative revenue of the policy $\hat{p}_t$ among all possible policies, or equivalently, to minimize the regret defined by,
\[
\operatorname{Regret}(T)
=\sum_{t=1}^T \Bigl( r(\bx_t,p_t^*)-r(\bx_t,\hat p_t)\Bigr).
\]

\paragraph{Examples.}  
Typical choices for $m$ include the following three cases, which we will discuss in detail in Section~\ref{sec:theory}. 
\begin{enumerate}
\item Linear model: $m(\bx) = \bx^{\top}\btheta$. 
\item Nonparametric additive model: $m(\bx) = \sum_{j=1}^d m_j(x_j)$.
\item Sparse linear model: $m(\bx) = \bx^{\top}\btheta$ with $\|\btheta\|_0\le s$.
\end{enumerate}

\subsection{Sketch of Algorithm and Contributions}

We use the doubling time window size strategy, which has a time window of size $2^l T_0$ at stage $l$, and learn dynamically the optimal pricing policy as follows.  See Figure~\ref{fig:diagram} for an illustration.

\begin{figure}[htbp]
	\centering
	\includegraphics[width=0.6\textwidth]{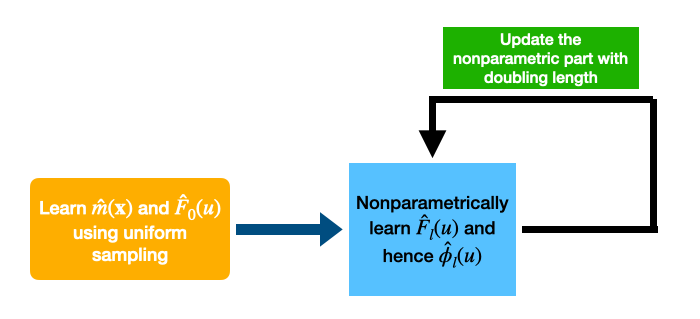}
	\caption{Algorithmic diagram, where $\phi(u) = u - \frac{1-F(u)}{F'(u)}$ is a crutical component in optimal pricing.}
	\label{fig:diagram}
\end{figure}

\paragraph{Stagewise Greedy Pricing via Iterative Local Polynomial Regression).}
\begin{enumerate}
    \item Explore in the first $T_0$ rounds to get an estimator $\hat m$ and $\hat F_0$.
    \item For $l=1,2,\ldots,\lceil \log_2(T/T_0) \rceil$: %\r{Q: How $\hat{F}_0$ is learned?}
    \begin{enumerate}
        \item Conduct greedy pricing for $2^{l-1} T_0$ periods based on $\hat m$ and the current estimate $\hat F_{l-1}$.
        \item Update the estimator of $F$ to obtain $\hat F_l$ using data collected in the preceding stage (we use local polynomial regression ~\cite{fan1996local}).
    \end{enumerate}
\end{enumerate}

\paragraph{Contributions.}
Our main contributions are summarized as follows:
\begin{itemize}[leftmargin=1.5em]
	\item We introduce a stagewise greedy pricing algorithm for semiparametric contextual demand models that leverages endogenous samples through localized nonparametric regression.
	\item For a general estimator of $m$, we establish a master regret bound. We then instantiate in cases when $m$ belongs to linear, nonparametric additive, and sparse linear models. In particular, for the linear case, we obtain an upper bound of order $T^{\max\{1/2,\,3/(2\beta+1)\}}$, improving upon existing results of \citep{fan2022policy} and achieving the parametric $\sqrt{T}$ rate when $\beta \ge 5/2$, where $\beta$ measures the smoothness of $F$. Figure ~\ref{fig:regret-beta-2-5} compares our new optimal rates of regret with those in \cite{fan2022policy}.  Note that when $\beta=2$, the upper bound becomes $T^{3/5}$, which matches the optimal rate derived in \cite{wang2025tight}.

	\item We prove a matching lower bound, thereby pinpointing the minimax regret rate with a linear utility function.
	\item Our analysis highlights how feature diversity enables efficient reuse of exploitation data for nonparametric learning, offering a new perspective on exploration--exploitation trade-offs in semiparametric pricing.
\end{itemize}

%\paragraph{In a nutshell.}  
%For parametric $m$, such as linear or generalized linear models, our method with suitably chosen hyperparameters achieves a minimax regret rate of
%\[
%T^{\max\{1/2,\, 3/(2\beta+1)\}},
%\]
%which improves upon the $T^{(2\beta+1)/(4\beta-1)}$ rate derived in \cite{fan2022policy}. Note that the doubling scheme there does not affect the rate. Our rate matches with \cite{wang2025tight} when $\beta=2$.

\begin{figure}[t]
\centering
\begin{tikzpicture}
\begin{axis}[
    width=0.82\linewidth,
    height=0.45\linewidth,
    xlabel={Smoothness $\beta$},
    ylabel={Regret exponent $e(\beta)$ in $T^{e(\beta)}$},
    xmin=2, xmax=5,
    ymin=0.45, ymax=0.9,
    xtick={2,2.5,3,3.5,4,4.5,5},
    grid=both,
    legend style={at={(0.98,0.98)}, anchor=north east},
    domain=2:5,
    samples=400,
]

% Our rate: max{1/2, 3/(2beta+1)}
\addplot[very thick] {max(0.5, 3/(2*x+1))};
\addlegendentry{Ours: $\max\{\tfrac12,\tfrac{3}{2\beta+1}\}$}

% Fan (2022) rate
\addplot[very thick, dashed] {(2*x+1)/(4*x-1)};
\addlegendentry{\cite{fan2022policy}: $\tfrac{2\beta+1}{4\beta-1}$}

% Crossover point beta = 5/2
\addplot[only marks, mark=*] coordinates {(2.5,0.5)};
\node[anchor=south] at (axis cs:2.5,0.52) {$\beta=\tfrac52$};

\end{axis}
\end{tikzpicture}
\caption{Regret exponent as a function of smoothness $\beta\in[2,5]$.
Lower is better. Our rate saturates at the parametric regime $T^{1/2}$ for
$\beta\ge 2.5$.
\cite{wang2025tight} gives the same rate $T^{3/5}$ only at $\beta=2$.}
\label{fig:regret-beta-2-5}
\end{figure}
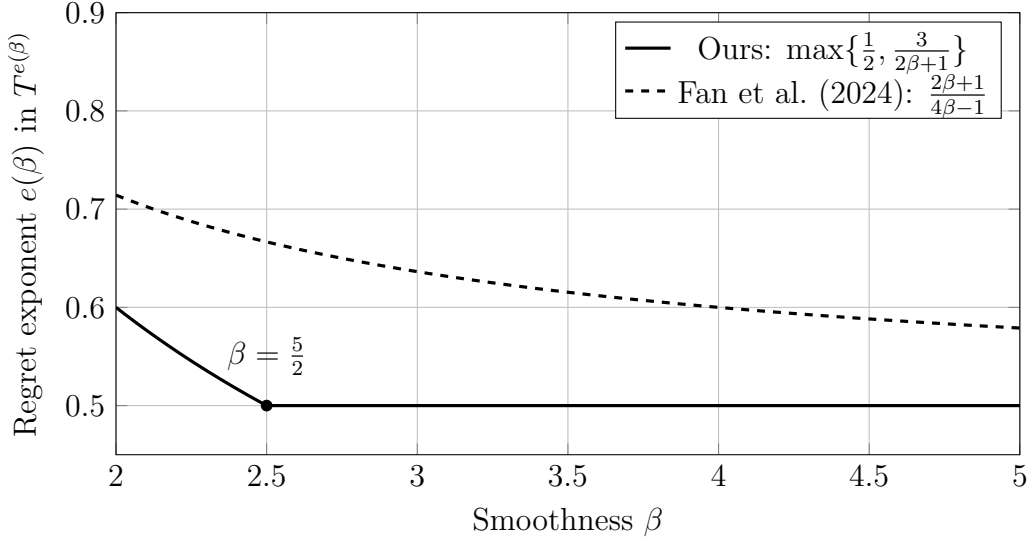

\subsection{Notations}

Constants such as $c$, $C_0$, and $C'$ may change from line to line.  
We write $[m]=\{1,2,\ldots,m\}$.  

For two nonnegative sequences $f(n)$ and $g(n)$, we use $f(n)\lesssim g(n)$ or $f(n)=\mathcal{O}(g(n))$ to denote $f(n)\le C g(n)$ for some universal constant $C$; similarly, $f(n)\gtrsim g(n)$ or $f(n)=\Omega(g(n))$ means $f(n)\ge C g(n)$. We write $f(n)\asymp g(n)$ when both relations hold.

We use $a\wedge b=\min\{a,b\}$ and $a\vee b=\max\{a,b\}$. For an event $\mathcal{E}$, let $\indicator{\mathcal{E}}$ denote its indicator function. Let $\mathrm{Unif}[a,b]$ denote the uniform distribution on $[a,b]$. 
Define the relative distance of $x$ in interval $[a,b]$ as $\alpha_{[a,b]}(x)=\min\{\frac{x-a}{b-a},\frac{b-x}{b-a}\}$.

For an interval $[\underline z,\overline z]$ and $v\in(0,1)$, define its $v$-interior by
\[
I_v[\underline z,\overline z] 
:= [\,\underline z + v(\overline z-\underline z),\; \underline z + (1-v)(\overline z-\underline z)\,].
\]
Also define the left and right $v$-shrink points:
\[
l_v[\underline z,\overline z] := \underline z + v(\overline z-\underline z), 
\qquad
r_v[\underline z,\overline z] := \underline z + (1-v)(\overline z-\underline z).
\]
To avoid overburdening notation, we do not mark parameters with $\star$ for ground truth.

\input{known-utility}

\input{methodology}

\input{theory}
\input{simulation}

\section{Discussion}
%\g{Add a one-paragraph discussion.}

This paper studies dynamic pricing under a general semiparametric model. Rather than relying on costly global random exploration, the proposed ILPR policy leverages endogenous data collected during exploitation to iteratively refine the nonparametric estimate of the market noise distribution. This leads to sharp regret guarantees and, in the linear utility setting, a minimax-optimal rate that attains the parametric $\sqrt{T}$ regime when $\beta \ge 5/2$. Our analysis underscores the importance of integrating adaptive pricing with carefully designed nonparametric estimation, including local polynomial regression, post-smoothing, and boundary correction. Several directions remain open. First, while matching lower bounds are established for linear utility, extending them to additive and high-dimensional sparse models would further elucidate the fundamental limits of semiparametric pricing. Second, it would be of interest to relax the feature diversity condition or develop robust variants under weaker support assumptions, especially since effective learning primarily relies on regions of the covariate space that are frequently visited. Finally, extending the framework to richer settings—such as more general demand models, multi-product pricing, inventory constraints, and nonstationary environments—may further reveal the benefits of endogenous data reuse.

\bibliographystyle{apalike2}
\bibliography{ref-main}

\newpage
\appendix
\input{app-Proof_ub}

\input{app-Proof_lb}

\input{app-utility}

\input{app-technical}

\end{document}

%% file: macros.tex
\usepackage{stackrel}

\newcommand{\one}{\ensuremath{\mathds{1}}}

\newcommand{\RKHS}{\ensuremath{\mathds{H}}}

\usepackage{upgreek} \newcommand{\distr}{\ensuremath{\upmu}}

\newcommand{\distrdata}{\ensuremath{\bar{\distr}}}

\newcommand{\bx}{{\ensuremath{\boldsymbol{x}}}}

\DeclarePairedDelimiterX{\anglep}[1]{(}{)}{#1}
\makeatletter
\newcommand{\@spanstar}[1]{{\rm span}\anglep*{#1}}
\newcommand{\@spannostar}[2][]{{\rm span}\anglep[#1]{#2}}
\newcommand{\Span}{\@ifstar\@spanstar\@spannostar}
\makeatother

\DeclarePairedDelimiterX{\dfun}[2]{(}{)}{#1 \;\delimsize\|\; #2}
\makeatletter
\newcommand{\@trunstar}[2]{$\chi$^2\dfun*{#1}{#2}}
\newcommand{\@trunnostar}[3][]{$\chi$^2\dfun[#1]{#2}{#3}}
\newcommand{\chisq}{\@ifstar\@trunstar\@trunnostar}
\makeatother

\DeclarePairedDelimiterX{\inprod}[2]{\langle}{\rangle}{#1, \, #2}

\DeclarePairedDelimiterX{\kulldiv}[2]{(}{)}{#1\;\delimsize\|\;#2}
\makeatletter
\newcommand{\@kullstar}[2]{D_{\text{KL}}\kulldiv*{#1}{#2}}
\newcommand{\@kullnostar}[3][]{D_{\text{KL}}\kulldiv[#1]{#2}{#3}}
\newcommand{\kull}{\@ifstar\@kullstar\@kullnostar}
\makeatother

\makeatletter
\newcommand{\@hilinstar}[2]{\inprod*{#1}{#2}_{\RKHS}}
\newcommand{\@hilinnostar}[3][]{\inprod[#1]{#2}{#3}_{\RKHS}}
\newcommand{\hilin}{\@ifstar\@hilinstar\@hilinnostar}
\makeatother

\makeatletter
\newcommand{\@mudatainstar}[2]{\inprod*{#1}{#2}_{\distrdata}}
\newcommand{\@mudatainnostar}[3][]{\inprod[#1]{#2}{#3}_{\distrdata}}
\newcommand{\mudatain}{\@ifstar\@mudatainstar\@mudatainnostar}
\makeatother

\makeatletter
\newcommand{\@mudatahinstar}[3]{\inprod*{#2}{#3}_{\distrdata}}
\newcommand{\@mudatahinnostar}[4][]{\inprod[#1]{#3}{#4}_{\distrdata}}
\newcommand{\mudatahin}{\@ifstar\@mudatahinstar\@mudatahinnostar}
\makeatother

\DeclarePairedDelimiterX{\defabs}[1]{|}{|}{#1}
\makeatletter
\newcommand{\@absstar}[1]{\defabs*{#1}}
\newcommand{\@absnostar}[2][]{\defabs[#1]{#2}}
\newcommand{\abs}{\@ifstar\@absstar\@absnostar}
\makeatother

\DeclarePairedDelimiterX{\norm}[1]{\|}{\|}{#1}
\makeatletter
\newcommand{\@normstar}[1]{\norm*{#1}_{\RKHS}}
\newcommand{\@normnostar}[2][]{\norm[#1]{#2}_{\RKHS}}
\newcommand{\hilnorm}{\@ifstar\@normstar\@normnostar}
\makeatother

\DeclareFontFamily{U}{matha}{\hyphenchar\font45}
\DeclareFontShape{U}{matha}{m}{n}{
	<-6> matha5 <6-7> matha6 <7-8> matha7
	<8-9> matha8 <9-10> matha9
	<10-12> matha10 <12-> matha12
}{}
\DeclareSymbolFont{matha}{U}{matha}{m}{n}

\DeclareFontFamily{U}{mathx}{\hyphenchar\font45}
\DeclareFontShape{U}{mathx}{m}{n}{
	<-6> mathx5 <6-7> mathx6 <7-8> mathx7
	<8-9> mathx8 <9-10> mathx9
	<10-12> mathx10 <12-> mathx12
}{}
\DeclareSymbolFont{mathx}{U}{mathx}{m}{n}

\DeclareMathDelimiter{\vvvert} {0}{matha}{"7E}{mathx}{"17}%

\DeclarePairedDelimiterX{\opnorm}[1]{\vvvert}{\vvvert}{#1}

\makeatletter
\newcommand{\@hilopnormstar}[1]{\opnorm*{#1}_{\RKHS}}
\newcommand{\@hilopnormnostar}[2][]{\opnorm[#1]{#2}_{\RKHS}}
\newcommand{\hilopnorm}{\@ifstar\@hilopnormstar\@hilopnormnostar}
\makeatother

\makeatletter
\newcommand{\@muopnormstar}[1]{\opnorm*{#1}_{\distr}}
\newcommand{\@muopnormnostar}[2][]{\opnorm[#1]{#2}_{\distr}}
\newcommand{\muopnorm}{\@ifstar\@muopnormstar\@muopnormnostar}
\makeatother

\makeatletter
\newcommand{\@mudataopnormstar}[1]{\opnorm*{#1}_{\distrdata}}
\newcommand{\@mudataopnormnostar}[2][]{\opnorm[#1]{#2}_{\distrdata}}
\newcommand{\mudataopnorm}{\@ifstar\@mudataopnormstar\@mudataopnormnostar}
\makeatother

\makeatletter
\newcommand{\@supnormstar}[1]{\norm*{#1}_{\infty}}
\newcommand{\@supnormnostar}[2][]{\norm[#1]{#2}_{\infty}}
\newcommand{\supnorm}{\@ifstar\@supnormstar\@supnormnostar}
\makeatother

\makeatletter
\newcommand{\@munormstar}[1]{\norm*{#1}_{\distr}}
\newcommand{\@munormnostar}[2][]{\norm[#1]{#2}_{\distr}}
\newcommand{\munorm}{\@ifstar\@munormstar\@munormnostar}
\makeatother

\makeatletter
\newcommand{\@mudatanormstar}[1]{\norm*{#1}_{\distrdata}}
\newcommand{\@mudatanormnostar}[2][]{\norm[#1]{#2}_{\distrdata}}
\newcommand{\mudatanorm}{\@ifstar\@mudatanormstar\@mudatanormnostar}
\makeatother

\makeatletter
\newcommand{\@xidatanormstar}[1]{\norm*{#1}_{\sdistrdata}}
\newcommand{\@xidatanormnostar}[2][]{\norm[#1]{#2}_{\sdistrdata}}
\newcommand{\xidatanorm}{\@ifstar\@xidatanormstar\@xidatanormnostar}
\makeatother

\makeatletter
\newcommand{\@distrnormstar}[2]{\norm*{#1}_{#2}}
\newcommand{\@distrnormnostar}[3][]{\norm[#1]{#2}_{#3}}
\newcommand{\distrnorm}{\@ifstar\@distrnormstar\@distrnormnostar}
\makeatother

\makeatletter
\newcommand{\@psinormstar}[2]{\norm*{#2}_{\psi_{#1}}}
\newcommand{\@psinormnostar}[3][]{\norm[#1]{#3}_{\psi_{#2}}}
\newcommand{\psinorm}{\@ifstar\@psinormstar\@psinormnostar}
\makeatother

\newcommand{\featureh}[1]{\ensuremath{\phi}}

\newcommand{\Lip}{\ensuremath{L}}

\newcommand{\indicator}[1]{\ensuremath{\one_{#1}}}

\newcommand{\Lipf}[1]{\ensuremath{\Lip_{f}}}

\newenvironment{carlist}
{\begin{list}{$\bullet$}
		{\setlength{\topsep}{0.1in} \setlength{\partopsep}{0in}
			\setlength{\parsep}{0.1in} \setlength{\itemsep}{\parskip}
			\setlength{\leftmargin}{0.15in} \setlength{\rightmargin}{0.08in}
			\setlength{\listparindent}{0in} \setlength{\labelwidth}{0.08in}
			\setlength{\labelsep}{0.1in} \setlength{\itemindent}{0in}}}
	{\end{list}}

\newcommand{\bcar}{\begin{carlist}}
	\newcommand{\ecar}{\end{carlist}}

\newcommand{\fullhline}{\vspace*{.1in}\noindent\makebox[\linewidth]{\rule{\linewidth}{0.4pt}}\vspace*{.1in}}

     \def\EE{\mathbb{E}}

     \def\PP{\mathbb{P}}
     
     \def\RR{\mathbb{R}}

 \def\cD{{\cal  D}}

\def\calG{{\cal  G}} 
 \def\cH{{\cal  H}}

\def\calM{{\cal  M}} \def\cM{{\cal  M}}
 
\def\calO{{\cal  O}} 
\def\calP{{\cal  P}} \def\cP{{\cal  P}}

\def\calT{{\cal  T}} \def\cT{{\cal  T}}
\def\calU{{\cal  U}}

\def\calX{{\cal  X}} \def\cX{{\cal  X}}

\newcommand{\bfsym}[1]{\ensuremath{\boldsymbol{#1}}}

 \def\btheta{\bfsym {\theta}}

 \def\htheta{\hat {\theta}}

\DeclareMathOperator{\E}{E}

%% file: known-utility.tex
\section{Stagewise Dynamic Pricing}

This section describes our stagewise dynamic pricing in detail.  Throughout this paper, we make the following assumption on the smoothness of $F$.
\begin{assumption}[Smoothness]
	\label{asp:smoothness}
	The CDF $F$ is $\beta$-smooth for some $\beta\ge 2$, in the sense that letting $q=\lfloor\beta\rfloor$, it holds that
	\[
	|F^{(q)}(x) - F^{(q)}(y)| \le C|x-y|^{\beta-q}, \quad\forall x,y,
	\]
	where $F^{(q)}$ denote the $q$-th derivative of $F$.
\end{assumption}

\subsection{Known Utility Function}
\label{sec:known-utility}

The improvements in the rate of convergence over the method in \cite{fan2022policy} come from the fact that we use the endogenous pricing data from the previous stages. To illustrate the basic idea, let us begin with the simplified setting where the mean utility function $m(\bx)$ is known (e.g., obtained from an external market%\b{or products and customers are homogenous so that there are no covariates}
), whereas the noise distribution $F$ is unknown. The task, therefore, reduces to estimating the noise CDF. The case of an unknown $m$ is treated in the next section. As alluded to before, we will assume diversity on the distribution of $\bx$, see Assumption~\ref{asp:feature-diversity-known-utility}. If the utility $m(\bx)$ concentrates near a single point, the problem effectively reduces to nonparametric dynamic pricing without covariates, for which the minimax regret rate $T^{(\beta+1)/(2\beta+1)}$ is established in \cite{wang2021multimodal}.

\paragraph*{Clairvoyant Pricing Policy}
The optimal (clairvoyant) price is
\[
p_t^\star = \arg\max_{p>0} p\bigl(1 - F(p - m(\bx_t))\bigr).
\]
The first-order condition yields the solution
\[
p_t^\star = m(\bx_t) + \phi^{-1}\bigl(-m(\bx_t)\bigr),
\]
where $\phi(u) = u - \frac{1-F(u)}{F'(u)}$. Define $g(u) = u + \phi^{-1}(-u)$. We have $p_t^\star = g(m(\bx_t))$.

% \begin{figure}[htbp]
%     \centering
%     \includegraphics[width=0.4\textwidth]{pic/function.png}
%     \caption{A visual illustration of $\phi(u)$ for a distribution $F$ supported on $[-1,1]$.}
%     \label{fig:function}
% \end{figure}

%Figure~\ref{fig:function} depicts a visualization example where $F$ is supported on $[-1,1]$.  \r{Q:  How helpful is this?}
The central nonparametric estimand in the pricing rule is $\phi^{-1}(u)$.  
A crucial advantage of the stagewise algorithm is that the pricing policy in stage $l$ (denoted by $\calT_l$) depends only on the information collected in stages $1$ through $l-1$. Consequently, conditional on the history up to stage $l-1$ and the independent contexts $\{\bx_t\}_{t \in \calT_l}$ in stage $l$, the resulting demands
$d(\bx_t, \hat p_{l-1} (\bx_t))$ (see \eqref{demand}) within stage $l$ are independent. This conditional independence allows $\phi^{-1}$ to be estimated nonparametrically using all collected observations.

% Assume the noise $\varepsilon$ is supported on $I_\epsilon=[-b_\epsilon,b_\epsilon]$.\footnote{The support need not be symmetric, and the algorithm does not require knowledge of it.} The function $m(\bx)$ takes values in an interval $I_m\subseteq[b_m,B-b_m]$ with $b_m\ge b_\epsilon$.  
To guarantee the existence and uniqueness of the optimal price, we assume that $\phi$ is invertible. This is ensured by imposing a monotonicity condition on $\phi$~\footnote{See also \cite{fan2022policy}, Assumption~2.1.}. Since $\phi$ depends on the first-order derivative of $F$ and the monotonicity condition effectively enforces a near–first-order smoothness requirement, this condition implicitly encodes our assumption on the smoothness parameter, $\beta \ge 2$. Most of the existing literature considers this range of smoothness; see, e.g., \cite{fan2022policy, luo2024distribution}. In particular, \cite{wang2025tight} focuses exclusively on the case $\beta = 2$.

\begin{assumption}[Monotonicity]
\label{asp:transform-strict-increasing}
There exists $c_\phi>0$ such that $\phi'(u)\ge c_\phi$ for all $u$.
\end{assumption}

\paragraph*{High-Level Algorithm}
Since $m$ is known, we need only to estimate $\phi$. We briefly describe the main steps of each stage.
Consider stage $l$ and $t\in\mathcal{T}_l$.   
Set the price $p_t = \hat g_{l-1}(\bx_t)$ using the current pricing rule with initial $\hat g_0$ being a random pricing policy  and observe $y_t$; for simplicity of notation, we drop the dependence of this pricing policy on the previous stages of the data.    Define the variable $u_t = p_t -m(\bx_t)$. 
A nonparametric estimator is then applied to $\{(u_t,y_t)\}$ in the current stage of the data to refine the estimate of $\phi$.  
This step constitutes the core of the method and is composed of three key
components, each of which is described below.

\paragraph*{Local Polynomial Regression}
We adopt local polynomial regression as our main nonparametric regression tool.  This allows us to estimate simultaneously the regression and its derivative.
At time $t$, since the price is set to $p_t=\hat g_{l-1}(\bx_t)$, % = m(\bx_t) + \hat\phi_{l-1}^{-1}( m(\bx_t))$, 
we have 
\[
\EE[y_t \mid \bx_t,p_t] = 1 - F(u_t), 
\]
which depends on  $u_t = p_t -m(\bx_t)$ and is a  canonical nonparametric regression problem.

Given $\{(u_t,y_t)\}_{t\in\mathcal{T}_l}$, a local polynomial estimator of order $q$ (where $q = \lfloor\beta\rfloor$) with kernel $K$ and bandwidth $h_l$ is defined as 
\begin{align}
\label{equ:local-polynomial}
\hat \mu_l(u)
=\arg\min_{\mu\in\RR^{k+1}}
\sum_{t\in\mathcal{T}_l}
\Bigl[y_t-\mu^\top U\!\left(\frac{u_t-u}{h_l}\right)\Bigr]^2
K\!\left(\frac{u_t-u}{h_l}\right),
\end{align}
where $U(u) = [1,u,\ldots,u^q/q!]^\top$.  
Then, following the local polynomial technique \citep{fan1996local}, we have the estimators for $F$ and its derivative 
\begin{align}
\label{equ:F-F'-estimator}
\hat F_l(u) = 1 - e_1^\top \hat\mu_l(u),\qquad
\hat F_l^{(1)}(u) = - \frac{e_2^\top \hat\mu_l(u)}{h_l},
\end{align}
where $e_1$ and $e_2$ denote the first two standard basis vectors, corresponding to the intercept and the slope.
%\r{Q:   Can notation slightly simpler if I take $1 - y$?}

\paragraph*{Post-Smoothing}

Define the initial estimate with LPR (local polynomial regression) at stage $l$ as $\hat \phi^I_l$: 
%\r{Q: $I$ stands for "iterated"? Why not "stagewise"? or "improved"?}
\begin{align*}
\hat\phi^I_l(u)=u-\frac{1-\hat F_l(u)}{\hat F^{(1)}_l(u)}.
\end{align*}
LPR alone does not ensure sufficient smoothness of $\hat\phi^I_l$.  
While $\|\hat\phi^I_l-\phi\|_\infty$ is well controlled, the derivative error $\|(\hat\phi^I_l)'-\phi'\|_\infty$ may not;  worse still, $\hat\phi^I_l$ can fail to be monotone. That prevents the inversion of $\hat\phi^I_l$ in computing the price.  
%This obstructs estimation in subsequent stages.

We therefore apply a post-processing step to $\hat\phi^I_l$ by a kernel smoothing.  
Let $K$ be a kernel function supported on $[-1,1]$, and let $\delta_u$ be a location-dependent bandwidth specified in Algorithm~\ref{alg}.  
Define the kernel smoothing as the convolution
\[
\hat\phi^S_l(u) = \int_{-\infty}^{\infty} \hat\phi^I_l(u-t) K_{\delta_u}(t)\,dt,
\qquad
K_\delta(t)=\frac{1}{\delta}K\!\left(\frac{t}{\delta}\right).
\]

\paragraph*{Boundary Perturbation}
So far, if we use $\hat\phi^S_l$ in pricing, the random variable  $(\hat\phi^S_l)^{-1}(-m(\bx_t))$ is supported on 

\[
\widehat\cD=
\{(\hat\phi^S_l)^{-1}(-m(\bx)):\bx\in\mathcal{X}\}.
\]
However, constructing the pricing rule requires estimating $\phi$, and hence $F$ and $F'$, on the domain
\[
\cD =
\{\phi^{-1}(-m(\bx)):\bx\in\mathcal{X}\}.
\]
These domains may not align, creating boundary issues well-known in nonparametric estimation \citep{fan1996local}.  
Errors early in the horizon may persist and propagate.

The convolution step introduces further boundary distortion.  
To mitigate these effects, we introduce a perturbation $\eta$ to $\hat\phi^S_l$. %so that the range of the resulting function covers that of $\phi$, 
See Algorithm~\ref{alg} for details.  
For $u$ in the interior $\operatorname{Int}_v([a,b])=[a+v(b-a),\, b-v(b-a)]$, define
\[
\hat\phi_l(u)=\hat\phi^S_l(u)+\eta(u)
\]
with linear extrapolation near the boundary. Here $\eta(u)$ is explicitly given by
\begin{subequations}
            \begin{align}
            \eta(x)=\begin{cases}
            \frac{c_1}{2}(x-v_1),\quad&\text{if}\ x\in (-\infty,v_1],\\
            0, &\text{if}\ x\in [v_1,v_2],\\
            \frac{c_1}{2}(x-v_2),\quad&\text{if}\ x\in [v_2,\infty),
            \end{cases}
            \end{align}
            \end{subequations}
            with $v_1,v_2,c_1$ to be chosen in Algorithm~\ref{alg}. %\r{Q: Is $\eta(u)$ known?  Explain more}
This perturbation guarantees the \emph{range-coverage} condition 
\begin{equation}
\label{equ:range-coverage-known-utility}
\hat\phi_l(0)\ge \phi(0),\qquad
\hat\phi_l(1)\le \phi(1),
\end{equation}
without degrading the regret rate. $\hat\phi_l$ is our final estimate for pricing in the next stage. 
Finally, we impose the following feature diversity assumption on the distribution of $m(\bx)$. This ensures sufficient probability mass of $(\hat\phi_l)^{-1}(-m(\bx_t))$ on the domain $\cD=\{\phi^{-1}(-m(\bx)):\bx\in\mathcal{X}\}$, enabling automatic exploration by our algorithm.

\begin{assumption}[Feature Diversity]
\label{asp:feature-diversity-known-utility}
Let $J_{m}=[\underline{u},\bar{u}]$ be the range of $m(\bx)$. Define
\[
\delta_{[\underline{u},\bar{u}]}(z):=\min\{(z-\underline{u})^{\kappa},(\bar u-z)^{\kappa}\},
\]
for a given $\kappa > 0$
and write $\delta(z)$ when the interval is clear from context. Assume that the density of the utility $m(\bx)$, denoted by $p_{ m}(z)$, exists and that there exist constants $0<c_d\le C_d<\infty$ such that
\[
c_d\,\delta(z)\le p_{m}(z)\le C_d\,\delta(z),
\qquad \text{for all } z\in J_{m}.
\]
\end{assumption}

\begin{theorem}[Known Utility]
Under Assumptions~\ref{asp:smoothness}--\ref{asp:feature-diversity-known-utility}, for the iterative local polynomial regression pricing policy $\pi^{\mathrm{ILPR}}$, it holds that
\[
\operatorname{Regret}(\pi^{\mathrm{ILPR}},T)
\;\lesssim\;
T^{\frac{3}{2\beta+1}}
\]
with probability at least $1-1/T$.
On the other hand, there exists a class $\mathfrak{F}$ satisfying Assumptions~\ref{asp:smoothness},\ref{asp:transform-strict-increasing} and \ref{asp:feature-diversity-known-utility} such that
\[
\inf_\pi\sup_{\,F\in \mathfrak{F}}
\mathbb{E}\bigl[\operatorname{Regret}(\pi, T)\bigr]
\;\ge\;
\tilde c\, T^{\frac{3}{2\beta+1}}.
\]
\end{theorem}

Since the proof is a special case of the general Theorem~\ref{thm:upper_bound} and the nonparametric part of Theorem~\ref{thm:lower_bound}, we omit it.

\begin{remark}
We give some high-level intuition of this result.
It is similar to the idea of how feature diversity helps improve the regret rates in bandit literature~\citep{bastani2021mostly}.  With our feature diversity Assumption~\ref{asp:feature-diversity-known-utility}, we can continue exploring (estimation) while at the same time exploiting (regret minimization). 
To illustrate the ideas, suppose for simplicity that each stage only consists of one time step.
At time $t+1$, the $t$ previous samples can all be used in estimation, incurring the $t^{- \frac{\beta-1}{2\beta+1}}$ estimation error of $\phi$, as it involves the derivative of $F$. Hence, the estimation error of the optimal pricing function $m+\phi^{-1}(-m)$ also has the same rate. Moreover, the per-time regret nearly the optimal pricing is a quadratic function of the pricing gap. Therefore, the total regret can be roughly bounded by 
\begin{align*}
\operatorname{Regret}(T)\lesssim\sum_{t=1}^T t^{- \frac{2 (\beta-1)}{2\beta+1}}\lesssim T^{\frac{3}{2\beta+1}}.
\end{align*}  In particular, when $\beta$ goes to infinity, the exponent of our bound effectively goes to zero.
This indicates that the regret grows slower than any polynomial rate, which is consistent with the logarithmic bound in \cite{bastani2021mostly}.
\end{remark}

%% file: methodology.tex
\subsection{Unknown Utility Function}
\label{sec:methodology}

We start with our general methodology when $m$ is unknown, in which case we require an additional exploration stage, resulting an estimator $\hat m$. We will detail how $\hat{m}$ can be obtained via regression in Section~\ref{sec:upperbound} and instantiate it for three commonly used examples.

\paragraph*{High-Level Algorithm}
We briefly describe the main steps of each stage; the complete procedure is given in Algorithm~\ref{alg}.  The stage schedule $\cT_l$ is usually taking a doubling strategy:  $\cT_l = \{T_02^{l-1}, \cdots, T_0^{l}\}$, for $l = 1,2 \cdots$. 

\begin{itemize}
\item 
\emph{Exploration.}  
Collect $T_0$ samples $\{(\bx_t,p_t,y_t)\}_{t=1}^{T_0}$ by the random pricing in the price range  $\cP = (0, B)$, say, with $p_t \sim \mathsf{Unif}(0,B)$   and regress $B y_t$ on $\bx_t$ using the data $\{(\bx_t,By_t)\}_{t=1}^{T_0}$ to obtain $\hat m$ using the model structure ~\citep{fan2022policy,luo2024distribution,bracale2025dynamic}; see Section~\ref{sec:upperbound} for further details.  The regression relationship follows from the fact 
$$\EE( By | \bx_t) = B\EE I(p_t \leq m(\bx_t) + \varepsilon_t  | \bx_t ) = \EE [m(\bx_t) + \varepsilon_t | \bx_t] = m(\bx_t).$$
%.~\footnote{A typical choice is to use uniform sampling~ .}.  
% Fit a regression of $(\bx_t, B y_t)$ to obtain an estimate $\hat m$ of $m$, where $B$ is a sufficiently large constant.  
% This yields an $L_\infty$ guarantee on $\|\hat m - m\|_\infty$.

\item 
\emph{Exploitation.}  
Consider stage $l$ and $t\in\mathcal{T}_l$.  
Set the price $p_t = \hat g_{l-1}(\bx_t)$ using the current pricing rule, and observe demand $y_t$.  
Define the observed variable $u_t = p_t - \hat m(\bx_t)$. 
A nonparametric estimator is then applied to $\{(u_t,y_t)\}_{t \in \cT_l}$ to refine the estimate of $\phi$.  
This step constitutes the core of the method and is composed of three key
components, each of which is described below.
\end{itemize}

\paragraph*{Local Quadratic Regression}
As in the known utility case, we adopt local polynomial regression as our main nonparametric tool. As we show later, estimating the unknown utility necessarily incurs at least $\sqrt{T}$ regret standard in online literature~\citep{lattimore2020bandit}. In the unknown-utility setting, when the smoothness parameter $\beta > 5/2$, the regret contribution from estimating the link function, of order $T^{\frac{3}{2\beta+1}}$, is always dominated by the $\sqrt{T}$ term arising from utility estimation. Consequently, it suffices to restrict attention to smoothness levels up to $5/2$ and to employ local quadratic regression under the assumption $\beta \geq 2$.
At time $t$, since the price is set to $p_t=\hat g_{l-1}(\bx_t)=\hat m(\bx_t)+\hat\phi_{l-1}^{-1}(-\hat m(\bx_t))$, 
we have,
\[
\EE[y_t \mid \bx_t,p_t] = 1 - F(\tilde u_t).
\]
where the true variable driving demand is given by
\[
\tilde u_t = \hat\phi_{l-1}^{-1}(-\hat m(\bx_t)) + \hat m(\bx_t) - m(\bx_t).
\]
However, the function $m$ is unknown and can only be estimated by $\hat m$.
Consequently, we use the observed variable
\[
u_t := p_t - \hat m(\bx_t)
= \hat\phi_{l-1}^{-1}\!\bigl(-\hat m(\bx_t)\bigr)
\]
as a proxy for the unobserved variable $\tilde u_t$.
Note that $u_t$ differs from $\tilde u_t$ only through the estimation error
$\hat m(\bx_t)-m(\bx_t)$.
Given $\{(u_t,y_t)\}_{t\in\mathcal{T}_l}$, the local quadratic estimate with kernel $K$ and bandwidth $h_l$ is defined via \eqref{equ:local-polynomial} and \eqref{equ:F-F'-estimator} with $q=2$.

% \[
% \hat \mu_l(u)
% =\arg\min_{\mu\in\RR^{k+1}}
% \sum_{t\in\mathcal{T}_l}
% \Bigl[y_t-\mu^\top U\!\left(\frac{u_t-u}{h_l}\right)\Bigr]^2
% K\!\left(\frac{u_t-u}{h_l}\right),
% \]
% where $U(u) = [1,u,u^2/2!]^\top$.  
% Then
% \[
% \hat F_l(u) = 1 - e_1^\top \hat\mu_l(u),\qquad
% \hat F_l^{(1)}(u) = - \frac{e_2^\top \hat\mu_l(u)}{h_l},
% \]
% where $e_1$ and $e_2$ denote the first two standard basis vectors.

\paragraph*{Post-Smoothing}
This step is identical to the known utility case and is thus omitted.

% Define the preliminary estimate with LPR at stage $l$ as $\hat \phi^I_l$,
% \begin{align*}
% \hat\phi_{l}^I(u)=u-\frac{1-\hat F_l(u)}{\hat F^{(1)}_l(u)}.
% \end{align*}
% LPR alone does not ensure sufficient smoothness of $\hat\phi^I$.  
% While $\|\hat\phi^I-\phi\|_\infty$ is well controlled, the derivative error $\|(\hat\phi^I)'-\phi'\|_\infty$ may not;  worse still, $\hat\phi^I$ can fail to be monotone, preventing $\hat g(\hat m(\bx))$ from having a proper density.  \r{Q: density?}
% This obstructs estimation in subsequent stages.

% We therefore apply a post-smoothing step to $\hat\phi^I$.  
% Let $K$ be a kernel supported on $[-1,1]$, and let $\delta_u$ be a location-dependent bandwidth specified in Algorithm~\ref{alg}.  
% Define the convolution
% \[
% \hat\phi^S(u) = \int_{-\infty}^{\infty} \hat\phi^I(u-t) K_{\delta_u}(t)\,dt,
% \qquad
% K_\delta(t)=\frac{1}{\delta}K\!\left(\frac{t}{\delta}\right).
% \]

\paragraph*{Boundary Perturbation}
This step follows nearly identically with the known utility case, except that the nonparametric design  $(\hat\phi^S_l)^{-1}(-\hat m(\bx_t))$ involves $\hat m$ instead of $m$, and we have to take into account the perturbation $\eta$.

% So far, if we use $\hat \phi=\hat\phi^S$ in pricing, the sampling variable would follow $(\hat\phi^S)^{-1}(-\hat m(\bx_t))$, supported on 
% \[
% \mathcal{D}=\{(\hat\phi^S)^{-1}(-\hat m(\bx)):\bx\in\mathcal{X}\}.
% \]
% However, constructing the pricing rule requires estimating $\phi$, and hence $F$ and $F'$, on the domain
% \[
% \{\phi^{-1}(-m(\bx)):\bx\in\mathcal{X}\}.
% \]
% These domains may not align, creating boundary issues well-known in nonparametric estimation \citep{fan1996local}.  
% Errors early in the horizon may persist and propagate.

% The convolution step introduces further boundary distortion.  
% To mitigate these effects, we introduce a perturbation $\eta$ so that $\hat\phi$ covers the true range of $\phi$, see Algorithm~\ref{alg} for details.  
% For $u$ in the interior $\operatorname{Int}_v([a,b])=[a+v(b-a),\, b-v(b-a)]$, define
% \[
% \hat\phi(u)=\hat\phi^S(u)+\eta(u)
% \]
% with linear extrapolation near the boundary.  \r{Q: Is $\eta(u)$ known, to be specified?}
% The perturbation is chosen to satisfy the \emph{range-coverage} condition
% \begin{equation}
% \label{equ:range-coverage}
% \hat\phi(0)\ge \phi(0),\qquad
% \hat\phi(1)\le \phi(1),
% \end{equation}
% without degrading the regret rate. $\hat\phi$ is our final estimate for pricing in the next stage.  

\newcommand{\HRule}{\vspace{-2em} \\ \rule{\linewidth}{0.1mm}}

\begin{algorithm}[htbp]
    \caption{Iterative Local Polynomial Regression (ILPR)}
    \label{alg}

\noindent\begin{minipage}{\linewidth}
\textbf{Input:}
initial valuation estimator $\hat m$; initial exploration length $T_0$; bandwidths $h_l$; number of stages $L=\lceil\log_2(T/T_0)\rceil$; other hyperparameters $a_1\neq 0$, $b_1$, $v$ as in \eqref{def:v-explicit} and $\delta_x$ as in \eqref{def:delta_x}.\\
\textbf{Output:} sequence of prices $\{p_t\}_{t=1}^T$.
\end{minipage}

    \fullhline
    \begin{algorithmic}[0]
        \State \emph{Initial stage:}
        % \For{$t=1,\ldots,T_0$}
        %     \State Draw $p_t \sim \mathsf{Unif}(0,B)$ and observe demand $y_t$.
        % \EndFor
        % \State Regress $y_t$ on $\bx_t$ using $\{(\bx_t,y_t)\}_{t=1}^{T_0}$ to obtain $\hat m$.
        \For{$t=1,\ldots,T_0$}
            \State Set $p_t = \hat m(\bx_t) + a_1 \hat m(\bx_t) + b_1$ and observe demand $y_t$.
        \EndFor
        \State Apply LPR to $\{(a_1 \hat m(\bx_t)+b_1, y_t)\}_{t=1}^{T_0}$ to obtain $\hat F_0$ and $\hat F_0^{(1)}$.
        \State Compute the initial transform estimator
        $
        \hat\phi_{0}^I(u)=u-\frac{1-\hat F_0(u)}{\hat F^{(1)}_0(u)}.
        $
        \State Compute the smoothed version
        $
        \hat\phi_0^S(u)=\int \hat\phi_0^I(u-t)K_{\delta}(t)\,dt,
        $
        and construct $\hat\phi_0$ as $\hat\phi_0(x)=\hat\phi_0^S(x)+b_0 (1-2x)$.
        Set $\hat g_0(u)=u+\hat\phi_0^{-1}(-u)$.
			
        \vspace{.5em}
        \State \emph{Refinement stage:}
        \State Set $t=T_0$.
        \For{$l=1,2,\ldots,L$}
            \State Set $t_l=0$.
            \While{$t_l < 2^l T_0$ and $t < T$}
                \State $t_l \gets t_l+1$, \quad $t \gets t+1$.
                \State Set $p_t=\hat g_{l-1}(\hat m(\bx_t))$ and observe $y_t$. 
                \algorithmiccomment{Greedy pricing given $\hat m$ and $\hat F_{l-1}$.}
            \EndWhile
            \State Let $\mathcal{T}_l$ be the set of time indices used in stage $l$.
            \State Apply local quadratic regression to $\{(p_t-\hat m(\bx_t),y_t)\}_{t\in \calT_l}$ to obtain $\hat F_l$ and $\hat F^{(1)}_l$.
            \State Compute
            $
            \hat\phi_{l}^I(u)=u-\frac{1-\hat F_l(u)}{\hat F^{(1)}_l(u)}.
            $
            \algorithmiccomment{Initial transform estimator.}
            \State Compute
            $
            \hat\phi_l^S(u)=\int \hat\phi_l^I(u-t)K_{\delta_u}(t)\,dt
            $
            \algorithmiccomment{Post-smoothing.}
            \State Construct $\hat\phi_l$ as follows,
            \begin{subequations}
            \label{def:perturbation}
            \begin{align}
            \hat\phi_l(x)=\begin{cases}
            \hat\phi_l^S(v_1)+\frac{c_1}{2}(x-v_1),\quad&\text{if}\ x\in (-\infty,v_1],\\
            \hat\phi_l^S(x), &\text{if}\ x\in [v_1,v_2],\\
            \hat\phi_l^S(v_2)+\frac{c_1}{2}(x-v_2),\quad&\text{if}\ x\in [v_2,\infty),
            \end{cases}
            \end{align}
            \end{subequations}
where $v_1=l_v[\underline{\hat z_{l-1}},\overline{\hat z_{l-1}}]$, $v_2=r_v [\underline{\hat z_{l-1}},\overline{\hat z_{l-1}}]$. And set $\hat g_l(u)=u+\hat\phi_l^{-1}(-u)$.
            \algorithmiccomment{Perturbation as in to ensure range coverage  $\hat g_l$ covers the range of $g$.} 
        \EndFor
        \vspace{.5em}
    \end{algorithmic}  
\end{algorithm}

%% file: theory.tex
\section{Theory}
\label{sec:theory}

In this section, we present our main theory, as well as the limitations of prior work.
We first list a few assumptions to deal with the feature density and mismatch issue. Note that, up to our main theorem, we treat $\hat m$ as a generic given estimator. All probabilistic results are conditional on $\hat{m}$. Later in Corollary~\ref{cor:upper-bound}, we present an end-to-end result where $\hat m$ is estimated via uniform sampling \citep{fan2022policy}.

\subsection{General requirements for the initial estimate}

In \cite{fan2022policy}, data from earlier stages are not fully exploited. To improve statistical efficiency, we impose Assumption~\ref{asp:feature-diversity}, which allows us to leverage all collected data and attain the minimax rate. 
It is a modified version of Assumption \ref{asp:feature-diversity-known-utility}, with the utility function $m$ replaced by the initial estimate $\hat{m}$.

\begin{assumption}[Feature Diversity]
\label{asp:feature-diversity}
Let $J_{\hat m}=[\underline{u},\bar{u}]$ be the range of $\hat m(\bx)$. Define
\[
\delta_{[\underline{u},\bar{u}]}(z):=\min\{(z-\underline{u})^{\kappa},(\bar u-z)^{\kappa}\},
\]
and write $\delta(z)$ when the interval is clear from context. Assume that the density of the utility estimator $\hat m(\bx)$, denoted by $p_{\hat m}(z)$, exists and that there exist constants $0<c_d\le C_d<\infty$ such that
\[
c_d\,\delta(z)\le p_{\hat m}(z)\le C_d\,\delta(z),
\qquad \text{for all } z\in J_{\hat m}.
\]
\end{assumption}

While boundary decay complicates algorithm design, it is partially mitigated in the regret analysis: regions where data are rarely observed correspond to states that contribute less to the regret. 

We also require mild regularity of the covariate distribution, in the same spirit of Assumption~4.3 (Lipschitzness) in \cite{fan2022policy}.  
Our condition is imposed on the conditional mean.  
Let $\hat\Delta:=\hat m-m$ and define
\[
\mathfrak{m}(t)
:= \mathbb{E}\bigl[\hat\Delta(\bx)/\|\hat\Delta\|_\infty \,\big|\, \hat m(\bx)=t\bigr],
\]
where the expectation is taken with respect to the randomness of $\bx$ and regarding $\hat m$ as a fixed function.

\begin{assumption}[Conditional Mean Regularity]
\label{asp:conditional-mean-Lip}
The function $\mathfrak{m}(t)$ is $\bar L$-Lipschitz on its domain for some constant $\bar L$.
\end{assumption}

% \begin{assumption}[Utility Estimation Accuracy]
% \label{asp:conditional-mean-Lip}
% %There exists a function $\epsilon_m(T_0)$ and a probability bound $\delta_m(T_0)$ such that, with probability at least $1-\delta_m(T_0)$, 
% The initial utility estimator $\hat m$ satisfies
% \[
% \|\hat m-m\|_{\infty}\le c\,\epsilon_m
% \]
% for some universal constant $c>0$.  \r{Q:  This is a deterministic bound and is unappealing.  Probably, replace it in the result in Theorem 2}
% \end{assumption}

After stating the main result in Section~\ref{sec:upperbound}, we will give three examples for modeling the mean utility function, from a low-dimensional linear model and nonparametric additive model to a high-dimensional sparse model, to illustrate that Assumptions~\ref{asp:feature-diversity} and \ref{asp:conditional-mean-Lip} hold.  See Corollary~\ref{cor:upper-bound} for the summary of the results.

\subsection{Upper and Lower Bounds on the Regret}\label{sec:upperbound}

In this section, we first state an upper bound on the regret and outline the main proof ideas. Then, we justify the sharpness of our results via a lower bound.
%, where a major part follows the spirit of Section~1.6 in \cite{Tsybakov2009}. 

\begin{theorem}
\label{thm:upper_bound}
Suppose that Assumptions~\ref{asp:smoothness}, \ref{asp:transform-strict-increasing}, \ref{asp:feature-diversity}, and \ref{asp:conditional-mean-Lip} hold. Then, up to logarithmic factors, the regret of the aforementioned 
{\em ILPR} policy satisfies
\[
\operatorname{Regret}(\pi^{\mathrm{ILPR}},T)
\;\lesssim\;
\max\Bigl\{T_0,\; T\,\epsilon_m^2,\; T^{\frac{3}{2\beta+1}}\Bigr\}
\]
with probability at least $1-1/T$, where $\epsilon_m=\|\hat m-m\|_{\infty}$ is the initial utility estimator error.

\end{theorem}

The proof is deferred to Appendix~\ref{sec:proof_ub}.  Theorem \ref{thm:upper_bound} holds conditionally on any given initial estimator $\hat{m}$ that satisfies the assumptions.
One common technique for estimating the mean utility is the random pricing during the exploration stage with a period of
$T_0$ and use a regression technique to obtain $\hat m$.  We denote the estimator error $\epsilon_m$ also as $\epsilon_m(T_0)$.  We now verify the conditions using the three useful examples.

\paragraph*{Example 1 (Linear Model)}
Consider the linear utility model 
\begin{align*}m(\bx)=\btheta_0^\top\bx\quad \text{with}\ \mathcal{X} = B^{d}(0,1).
\end{align*}  
%where $\|\btheta_0\|_2$ is a constant.
The shape of the domain is not essential, provided mild smoothness and regularity conditions. For example, a similar result holds with $\mathcal{X}=[0,1]^d$.  We say that a random vector $X$ is \emph{quasi-uniform} on $\mathcal{X}$ if its density $p_X$ is bounded from above and below, namely there exist some constants $0<c\le C<\infty$, such that  $c \le p_X(x)\le C$ for all $x\in\mathcal{X}$.

Since the model is linear, it is natural to estimate it within the class of linear functions using the least-square based on the exploiration data.  Accordingly, we take $\hat m \in \{f:f(\bx)=\btheta^\top \bx\}$. We verify Assumption~\ref{asp:feature-diversity} and \ref{asp:conditional-mean-Lip} for every linear function $f(\bx)=\btheta^\top \bx$.
If $X$ is quasi-uniform on the $d$-dimensional unit ball and $U = \btheta^\top X \in[-\|\btheta\|_2,\|\btheta\|_2]$, then the density of $U$ at $u \in [-\|\btheta\|_2,\|\btheta\|_2]$ is both upper and lower bounded by
\[
\Bigl(1-(u/\|\btheta\|_2)^2\Bigr)^{(d-1)/2}
\]
up to constant factors. Hence, the density $g_Z$ of
\[
Z = \frac{U/\|\btheta\|_2 + 1}{2} \in [0,1],
\]
satisfies,  for some constants $0<c\le C<\infty$,
\[
c\,\bigl(\min\{z,1-z\}\bigr)^{(d-1)/2}
\ \le\ 
g_Z(z)
\ \le\ 
C\,\bigl(\min\{z,1-z\}\bigr)^{(d-1)/2},
\qquad z\in[0,1].
\]
Thus the boundary profile in Assumption~\ref{asp:feature-diversity} holds with $\kappa=(d-1)/2$. 

%Assumption~\ref{asp:conditional-mean-Lip} is satisfied whenever the conditional expectation $\mathbb{E}(\bx \mid \btheta^\top \bx = u)$ is Lipschitz in $u$ for all $\|\btheta\|_2 = 1$.
% \g{
% Is it correct to rephrase the previous statement as the following:
% ``Assumption~\ref{asp:conditional-mean-Lip} is satisfied if for any $\|\btheta\|_2 = 1$, the function $u \mapsto \mathbb{E}(\bx \mid \btheta^\top \bx = u)$ is Lipschitz?''
% Do we need $L$-Lipschitz (as in Assumption 5) or simply Lipschitz? Also, we should provide a proof (put it here if short; otherwise, defer to the appendix).
% }
Denote $\hat m(\bx)=\btheta_1^\top \bx$. Without loss of generality, assume that $\|\btheta_1\|$ is bounded away from zero and infinity, i.e., $1/\tilde C<\|\btheta_1\|<\tilde C$ for some constant $C$. Indeed, one may rescale the norm without increasing the error $\epsilon_m(T_0)$ if $1/\tilde C<\|\btheta_0\|<\tilde C$. Denote $\hat\Delta(\bx_i)/\|\hat\Delta\|_\infty =\btheta_2^\top \bx$. 
We claim that Assumption~\ref{asp:conditional-mean-Lip} is satisfied if 
\begin{equation}
\label{equ:asp:linear-case-lip}\text{for any } \|\btheta\|_2 = 1\text{, the function }  u \mapsto h_{\btheta}(u):=\mathbb{E}(\bx \mid \btheta^\top \bx = u) \text{ is } \bar L/\tilde C\text{-Lipschitz}.\end{equation}
To see this, note that
\begin{align*}
\mathfrak{m}(t)
=\mathbb{E}\Bigl[\btheta_2^\top \bx\,\Big|\, \frac{\btheta_1^\top\bx}{\|\btheta_1\|}=\frac{t}{\|\btheta_1\|}\Bigr]=h\left(\frac{t}{\|\btheta_1\|}\right)
\end{align*}
is $\tilde C\cdot \bar L/\tilde C=\bar L$-Lipschitz.

Condition \eqref{equ:asp:linear-case-lip} can be further reduced to a Lipschitz condition on the density function; see Appendix~\ref{sec:linear-model} for details.

% \g{Assumptions 4 and 5 are imposed on $\hat{m}$, which does not appear above. It is implicitly assumed that $\hat{m}$ is linear.}\b{Revised}

Moreover, with probability at least $1-T_0^{-2}$, the standard least-squares estimator achieves
\[ 
\epsilon_m(T_0) \lesssim \sqrt{d/T_0}.
\]

\paragraph*{Example 2 (Additive Model)}
More generally, $m$ may be nonparametric, for instance, an additive model
\[
m(\bx) = \sum_{j=1}^d m_j(x_j),
\]
with $\mathcal{X}=[0,1]^d$.  
Consider a nonparametric least-squares estimator over a sieve function class $\mathcal{M}:=\left\{
m(x)=\sum_{j=1}^d m_j(x_j): m_j(x_j)=\sum_{k=1}^K \theta_{jk}\phi_k(x_j),
\theta_{jk}\in\RR,\ \gamma_1\le m_j'(x_j)\le\Gamma_1
\right\}$ for orthogonal basis $\{\phi_k\}_{k\ge 1}$,
\[
\hat m = \arg\min_{m \in \mathcal{M}} \frac{1}{T_0}\sum_{i=1}^{T_0} \bigl(Y_i - m(x_i)\bigr)^2,
\]
where $(x_i,Y_i)$ are i.i.d.\ samples.  

Under suitable regularity conditions on the covariate distribution, Assumptions~\ref{asp:feature-diversity}--\ref{asp:conditional-mean-Lip} are satisfied. When \{$m_j\}$ are $\gamma$-smooth, with probability at least $1-T_0^{-2}$, this estimator achieve
\[ 
\epsilon_m(T_0) \lesssim \sqrt{d}T_0^{-\frac{\gamma}{2\gamma+1}}.
\]
Further conditions and details are given in Appendix~\ref{sec:additive-model}.

\paragraph*{Example 3 (Sparse Linear Model)}
Consider the sparse linear model
\[
m(\bx)=\btheta^\top \bx,\qquad \|\btheta\|_0\le s,
\]
with covariates $\bx\sim\mathrm{Unif}([0,1]^d)$.
Let
\[
\hat\btheta^{\mathrm{lasso}}
=\arg\min_{\beta\in\RR^d}
\Big\{
\frac{1}{2T_0}\sum_{i=1}^{T_0} (By_i-\beta^\top \bx_i)^2
+\lambda\|\beta\|_1
\Big\},
\qquad
\lambda=C \sqrt{\frac{\log (dT_0)}{T_0}}.
\]

Since the covariate is sub-Gaussian, it satisfies the restricted eigenvalue condition~\citep{bickel2009simultaneous,negahban2012unified} on the true
support $S=\mathrm{supp}(\btheta)$. 
%Theorem 7.13 in \cite{wainwright2019high} gives
%~\footnote{One may obtain better dependence on $s$ by using irrepresentable conditions, but we omit it.}
%\[
%\|\hat\btheta^{\mathrm{lasso}}-\btheta\|_\infty
%\le \|\hat\btheta^{\mathrm{lasso}}-\btheta\|_2\le C_\infty\lambda\sqrt{s}
%\]
%with probability at least $1-T_0^{-2}$ for some constant $C_\infty>0$.

Theorem 5.5 of \cite{fan2020statistical} reveal that $\|\hat\btheta-\btheta\|_1\lesssim
s \sqrt{\frac{\log d}{T_0}}$ with probability tending to one exponentially fast.
Consequently, the initial utility error is bounded by using
\[
\|\hat m-m\|_\infty
=
\sup_{\bx\in[0,1]^d}
|(\hat\btheta-\btheta)^\top\bx|
\le
\|\hat\btheta-\btheta\|_1
\lesssim
s \sqrt{\frac{\log d}{T_0}}.
\]

Define the hard-thresholded estimator
\[
\hat\btheta_j
:=\hat\btheta^{\mathrm{lasso}}_j\,
\mathbf 1\{|\hat\btheta^{\mathrm{lasso}}_j|\ge C_\infty\lambda\sqrt{s}\},
\qquad j\in[d].
\]
If the beta-min condition $\min_{j\in S}|\theta_j|>3C_\infty\lambda$ holds, then
$\mathrm{supp}(\hat\btheta)\subseteq S$ and
\[
\|\hat\btheta-\btheta\|_\infty\lesssim\lambda\sqrt{s}.
\]

Finally, conditional on the recovered support, the model reduces to a
low-dimensional linear regression of dimension at most $s$. In particular,
Assumptions~\ref{asp:feature-diversity} and~\ref{asp:conditional-mean-Lip} are
satisfied with constants depending only on $s$ and not on $d$.

\vspace{2em}
To summarize, we obtain the following corollary of Theorem~\ref{thm:upper_bound}.
\begin{corollary}
\label{cor:upper-bound}
Using the aforementioned uniform pricing and the associated statistical learning, we have
\begin{enumerate}
	\item For linear model utility in Example 1, we have $\epsilon_m\lesssim (d/T_0)^{1/2}$ with probability at least $1-T_0^{-2}$. Choosing $T_0=\lceil\sqrt{dT}\rceil$ yields, with probability at least $1-2/T$,
\[
\operatorname{Regret}(\pi^{\mathrm{ILPR}},T)
\;\lesssim\;
\max\{\sqrt{dT},T^{\frac{3}{2\beta+1}}\}.
\]

\item For additive model utility in Example 2, we have $\epsilon_m\lesssim T_0^{-\frac{\gamma}{2\gamma+1}}$ with probability at least $1-T_0^{-2}$. Choosing $T_0=\lceil (Td)^{\frac{2\gamma+1}{4\gamma+1}}\rceil$ yields, with probability at least $1-2/T$,
\[
\operatorname{Regret}(\pi^{\mathrm{ILPR}},T)
\;\lesssim\;
\max\{(Td)^\frac{2\gamma+1}{4\gamma+1},T^{\frac{3}{2\beta+1}}\}.
\]

\item For sparse linear model utility in Example 3, we have $\epsilon_m\lesssim T_0^{-1/2}$ with probability at least $1-T_0^{-2}$. Choosing $T_0=\lceil\sqrt{s^3T\log d}\rceil$ yields, with probability at least $1-2/T$,
\[
\operatorname{Regret}(\pi^{\mathrm{ILPR}},T)
\;\lesssim\;
\max\{\sqrt{s^3 T\log d},T^{\frac{3}{2\beta+1}}\}.
\]
\end{enumerate}
\end{corollary}

\begin{remark}
In the above bounds, the first term captures the error from estimating the mean
utility, while the second term reflects to the error from estimating the
nonparametric link function. When $\beta \ge 5/2$, the nonparametric estimation
error saturates and is always dominated by the utility estimation error.
\end{remark}

%\subsection{Lower Bound}

For the linear utility class, we obtain the first minimax regret bound for every smoothness parameter $\beta\ge 2$. The proof is deferred to Appendix~\ref{sec:proof_lb}. To give an overview, we consider two scenarios, both with dimension $d=1$.
 
In the \textit{first} scenario, define the function classes $\mathfrak{F}_1$ and $\mathfrak{M}_1$ as
\begin{align*}
\mathfrak{F}_1
&=\Big\{F_{a,b}:\RR\to[0,1]\;\Big|\;
F_{a,b}(u):=\calT_{[0,1]}(au+b),\ a>0\Big\},\\
\mathfrak{M}_1
&=\Big\{m_\theta:[0,1]\to\RR\;\Big|\;m_\theta(x):=\theta x\Big\},
\end{align*}
where $\calT_{[0,1]}(x)$ clips $x \in \RR$ to $[0,1]$.

In the \textit{second} scenario, fix $m(x)=x$ so that $\mathfrak{M}_2=\{\operatorname{Id}\}$. Define the baseline CDF on $[-1/4,1/4]$ as 
$
F_0(u)=u+\frac{1}{2}, \quad \forall u\in [-1/4,1/4]$
{and extend it smoothly to a CDF on $\RR$ so that $F_0$ is $\beta$-smooth on $\RR$ and coincides with $u\mapsto u+1/2$ on $[-1/4,1/4]$.}
We then construct a family of CDFs by adding small localized bumps onto the baseline. Specifically, pick an integer parameter $N$ and set $h=1/N$. Choose an infinitely smooth function $V$ with support $[-1,1]$. Discretize $[-1/4,1/4]$ into $N$ evenly spaced intervals $I_1,\cdots, I_N$ and denote by $\mu_j$ the midpoint of $I_j$. For each $\iota=(\iota_1,\ldots,\iota_N)\in \{0,1\}^{N}$, define
\begin{align*}
F_{\iota}(u)
=
F_0(u)+\sum_{j=1}^N \iota_j\, \rho\, h^\beta \,
V\!\left(\frac{u-\mu_j}{h}\right),\qquad u\in I_j,
\end{align*}
Let $\mathfrak{F}_2$ denote the collection of all such functions.

\begin{theorem}
\label{thm:lower_bound}
Let $\mathfrak{F}_{l}$ and $\mathfrak{M}_{l}$ ($l=1,2$) be function classes for $F$ and $m$ (linear model) defined above. They satisfy Assumption~\ref{asp:smoothness} and \ref{asp:transform-strict-increasing}; in addition, Assumptions~\ref{asp:feature-diversity} and \ref{asp:conditional-mean-Lip} holds for every function in $\mathfrak{M}_{l}$ ($l=1,2$). 
%\g{Assumption 3 also needs to be mentioned somewhere, as it was used for developing the upper bound.}
%\g{Move their definitions from the appendix to here. Also, Assumptions 4 and 5 are imposed on $\hat{m}$ rather than $m$.}
There exists a constant $\tilde c>0$ such that
\[
\inf_\pi\sup_{(m,F)\in (\mathfrak{M}_1,\mathfrak{F}_1)\cup (\mathfrak{M}_2,\mathfrak{F}_2)}
\mathbb{E}\bigl[\operatorname{Regret}(\pi, T)\bigr]
\;\ge\;
\tilde c\, T^{\max\{\frac{3}{2\beta+1},\frac{1}{2}\}}.
\]
\end{theorem}

\subsection{Comparisons with prior work}
The algorithm in \cite{fan2022policy} essentially employs an initial uniform pricing phase to estimate $m$ and $F$, followed by the static greedy pricing using these estimates for the remaining periods. Besides the algorithm design, the nonparametric estimator used therein also has several limitations, as detailed below.

\paragraph*{Limitations of the Nadaraya--Watson Estimator}
\cite{fan2022policy} employ the Nadaraya--Watson (NW) estimator. However, it exhibits two structural limitations, as detailed below.  
They start by estimating the nonparametric link (noise CDF)
\[
\hat F_l(u) = 1 - \frac{h_l(u)}{f_l(u)},
\qquad
h_l(u) = \frac{1}{n_l b_l}\!\sum_{t\in\mathcal{T}_l} K\!\left(\frac{w_t-u}{b_l}\right)y_t,\quad
f_l(u) = \frac{1}{n_l b_l}\!\sum_{t\in\mathcal{T}_l} K\!\left(\frac{w_t-u}{b_l}\right).
\]
The derivative estimator is
\[
\hat F_l^{(1)}(u)
= -\frac{h_l^{(1)}(u)f_l(u) - f_l^{(1)}(u)h_l(u)}{f_l^2(u)}.
\]
Then they obtain the plug-in estimator
\[
\hat\phi_l(u)=u-\frac{1-\hat F_l(u)}{\hat F_l^{(1)}(u)},\qquad
\hat g_{l-1}(\bx_t)=\hat m(\bx_t)+\hat\phi_{l-1}^{-1}\!\bigl(-\hat m(\bx_t)\bigr),
\]

The NW estimator exhibits two severe issues as detailed below.
\begin{itemize}
	\item\emph{Error-in-variables.}  
As discussed before, the observed covariates are given by
\[
u_t=p_t-\hat m(\bx_t) = \hat\phi_{l-1}^{-1}\!\bigl(-\hat m(\bx_t)\bigr),
\]
while the true regression target corresponds to  
$\tilde u_t=p_t - m(\bx_t)$, leading to a discrepancy of $\hat m(\bx_t)-m(\bx_t)$.  
Classical methods for the error-in-variables problem—e.g., deconvolution \citep{fan1993nonparametric} or instrumental variables \citep{adusumilli2018nonparametric}—do not directly apply because the errors here depend on the covariates. %In particular, expanding and analyzing the error-in-variable problem is much more complex than the local polynomial estimator.

\item  \emph{Smoothness of the covariate distribution.}  
The NW estimator requires the covariate density to be as smooth as the regression
function. 
In our setting, $\phi$ depends on $F'$ and may not inherit this smoothness.  
Attempting to relate the observed covariate distribution to that of the oracle covariate $\phi^{-1}(-m(\bx_t))$ via Taylor expansion produces leading error terms that cannot be controlled at the optimal rate.  
Hence, NW estimators are fundamentally limited in this setting.  
Moreover, \cite{fan2022policy} assumes a smooth covariate density, which cannot hold for the covariates induced by our algorithm.
\end{itemize}

%% file: simulation.tex
\section{Numerical experiments}

We conduct extensive numerical experiments to back up our theory, including simulations and semi-real data~\footnote{The code can be found in \href{https://github.com/jinhangc/ILPR}{https://github.com/jinhangc/ILPR}.}.
%\g{Prepare a GitHub repository and provide the link here. Also, make sure to describe all the experimental details in the main body or the appendix.}

\subsection{Simulations}
\label{sec:simulation}

We first evaluate our proposed pricing algorithm in a synthetic environment, considering both cases of known and unknown utility.

\paragraph*{Data-generating process}
At each round $t=1,\dots,T$, a context $x_t \sim \mathrm{Unif}([x_{\min},x_{\max}])$ is observed and the learner posts a price $p_t \in [P_{\min},P_{\max}]$. The buyer purchases with probability
\[
\Pr(D_t=1 \mid x_t,p_t) \;=\; 1 - F\!\big(p_t - m(x_t)\big),
\]
where $F$ is an unknown CDF supported on $[-0.3,0.3]$ and $m(\cdot)$ is a mean-utility function. We choose $P_{\min}=0, P_{\max}=1, x_{\min}=0.35, x_{\max}=0.65$.
We generate $F$ from a smooth baseline CDF on $[-0.25,0.25]$ and add alternating-sign smooth ``bumps'' so that $F$ remains $C^4$-smooth while exhibiting local non-monotone features in its density. Figure~\ref{fig:truth} plots the ground truth functions, with further details given in Appendix~\ref{sec:experiment-detail}. We consider linear utility $m(x)=\theta x$ with $\theta=1$ throughout.

\begin{figure}
    \centering
    \includegraphics[width=\linewidth]{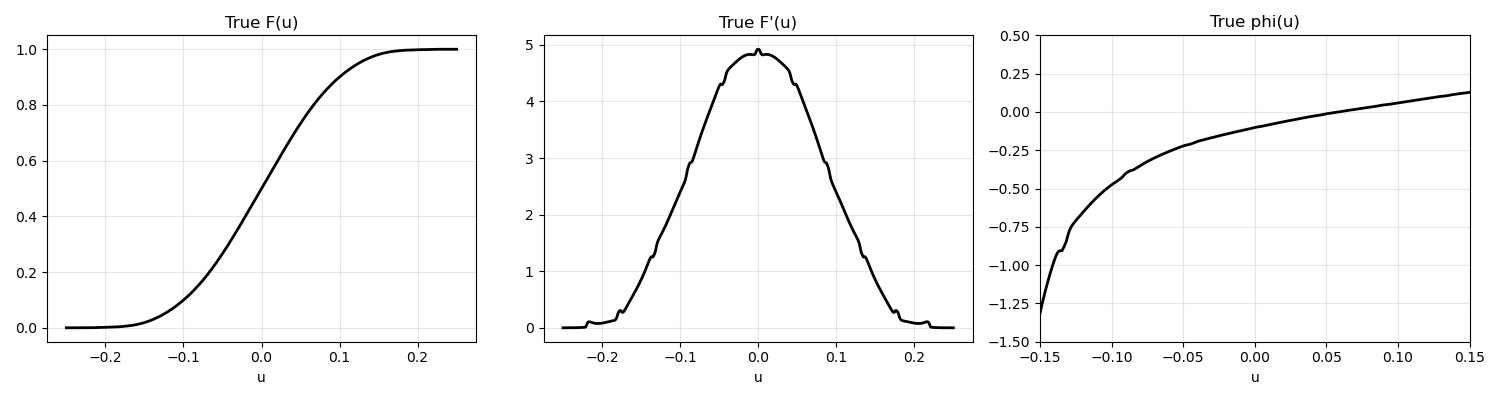}
    \caption{Ground truth functions of the noise CDF, with smoothness $\beta=2$, and $10$ knots}
    \label{fig:truth}
\end{figure}

\paragraph*{Evaluation}
For each horizon $T$ and smoothness parameter $\beta$, we report the average cumulative regret over $N_{\mathrm{trials}}=200$ independent runs.
% We additionally visualize the estimated $(\hat F,\hat F',\hat\phi)$ along with the induced pricing rule $x\mapsto \hat p(x)$.
We plot the curves on a log–log scale and estimate the slope using linear regression. In addition, we apply the cluster bootstrap method~\citep{cameron2008bootstrap}, performing 2,000 bootstrap refits to construct confidence intervals for the slope estimator.

\paragraph*{Known utility}
We consider first the case with known utility (Section~\ref{sec:known-utility}).
The learner proceeds in stages:
(i) an exploration phase of length $T_0=100$ using uniformly random prices to obtain an initial estimate of $F$ and $F'$ mainly via local polynomial regression with Epanechnikov kernel on $\E[ Y \mid u]$ with $u=p-\hat m(x)$;
(ii) exploitation phase, where the price is chosen by the estimated optimal pricing policy,
\[
p_t \;=\; \Pi_{[P_{\min},P_{\max}]}\Big(\hat\phi^{-1}\!\big(-\hat m(x_t)\big) + \hat m(x_t)\Big),
\]
with stagewise refits of $\hat\phi$ at doubling stage time windows.  When using the local polynomial regression, we calculate both the function value and derivative on a grid of $300$ with bandwidth $h=0.5\cdot n^{-1/(2\beta+1)}$.  To improve stability, we then conduct padding on the estimated CDF: $\hat F(u)=0$ for $u\le -0.3$ and $\hat F(u)=1$ for $u\ge 0.3$, and we floor the density estimate inside the support with $10^{-3}$. Then we apply post-smoothing with the specified variable bandwidth, with details in Appendix~\ref{sec:experiment-detail}. After that, we add additional monotonic regression on the estimated $\hat\phi$ to further enforce stability. Figure \ref{fig:known} depicts the growth of total regret over the time horizon, and Table~\ref{tab:expo_known} reports the least-square estimate of the regret exponent and corresponding confidence intervals. As shown in the figure and table, our proposed method achieves an error exponent that closely matches the theoretical benchmark.

\begin{figure}[htbp]
    \centering
    \includegraphics[width=0.7\textwidth]{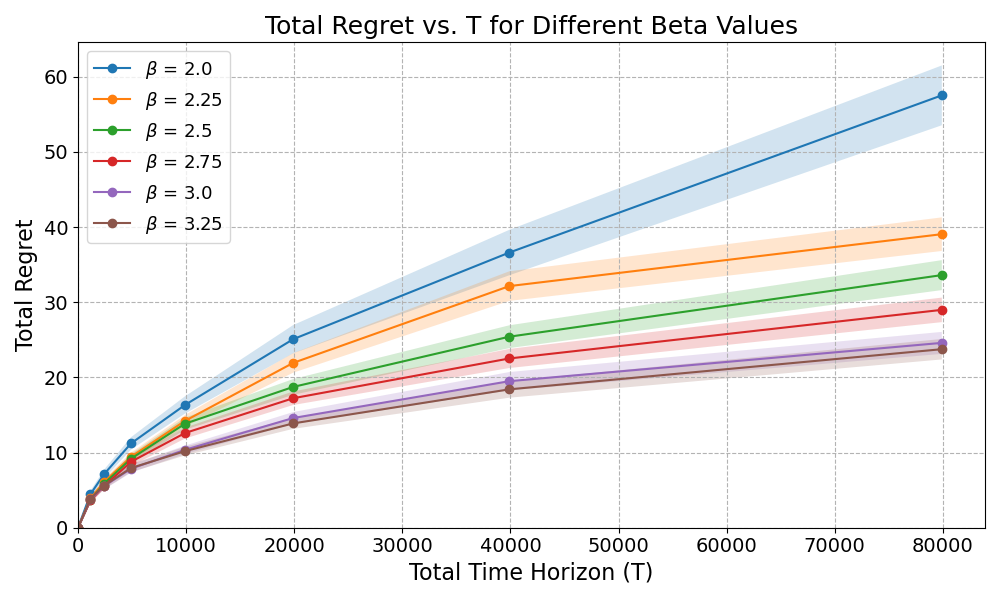}
    \caption{\textbf{Known utility} 
    The x-axis stands for total time horizon, and the y-axis stands for total regret. 
    The solid lines represent average regrets, and the shaded areas reflect the standard deviations across trials. Here, blue, orange, blue, red, purple, and brown components correspond to the cases with $\beta=2,2.25,2.5,2.75,3,3.25$ respectively. }
    \label{fig:known}
\end{figure}

\begin{table}[]
    \centering
    \begin{tabular}{c|ccc}
        $\beta$ & $\widehat{\text{slope}}$ & Theory & CI \\
        \hline
        2.00 & 0.595 & 0.600 & [0.540, 0.649] \\
        2.25 & 0.552 & 0.545 & [0.510, 0.595] \\
        2.50 & 0.499 & 0.500 & [0.454, 0.544] \\
        2.75 & 0.470 & 0.462 & [0.425, 0.515] \\
        3.00 & 0.434 & 0.429 & [0.382, 0.485] \\
        3.25 & 0.415 & 0.400 & [0.363, 0.466] \\
    \end{tabular}
    \caption{Least-squares estimates of the regret exponent (so that regret scales as $T^{\widehat{\text{slope}}}$) for each $\beta$ with known utility. Brackets indicate confidence intervals obtained via cluster bootstrap. Our proposed method does not differ significantly from the theoretical one in terms of the regret exponent.}
    \label{tab:expo_known}
\end{table}

\paragraph*{Unknown utility}
We then consider unknown utility. The only difference is that we insert an initial utility-estimation phase of length $T_{0m}=\lceil\sqrt{4T}\rceil$ using uniformly random prices to collect the data and to fit $\hat m$ by OLS from binary outcomes. The results are shown in Figure \ref{fig:unknown}, while Table~\ref{tab:expo_unknown} reports the least-square estimate of the regret exponent and corresponding confidence intervals. As shown in the figure and table, our proposed method does not differ significantly from the theoretical one in terms of the error exponent.

\begin{figure}[htbp]
    \centering
    \includegraphics[width=0.7\textwidth]{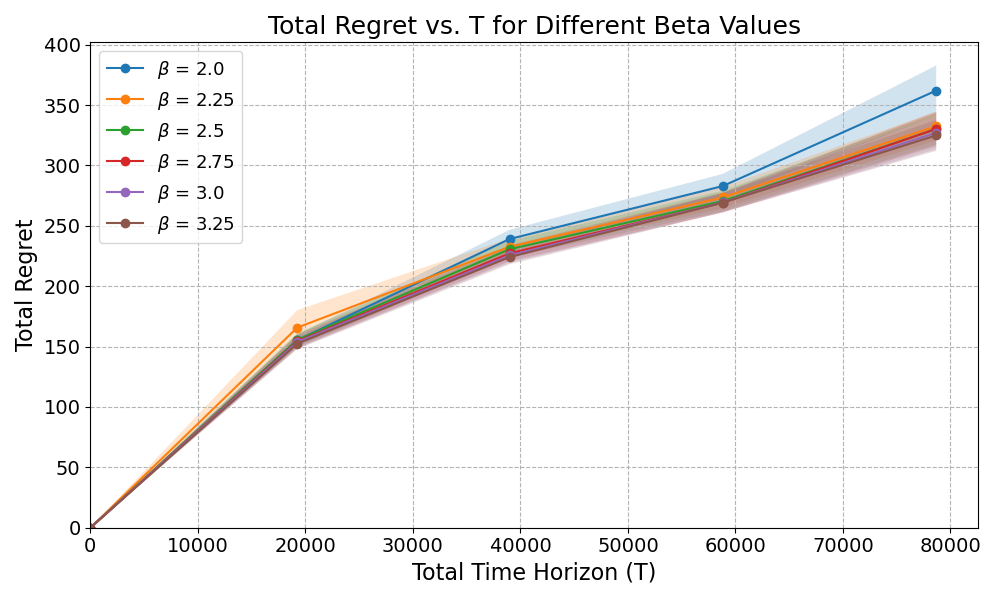}
    \caption{\textbf{Unknown utility.} The solid lines represent average regrets, and the shaded areas reflect the standard deviations across trials. }
    \label{fig:unknown}
\end{figure}

\begin{table}[htbp]
    \centering
    \begin{tabular}{c|ccc}
        $\beta$ & $\widehat{\text{slope}}$ & Theory & CI \\
        \hline
        2.00 & 0.582 & 0.600 & [0.555, 0.609] \\
        2.25 & 0.533 & 0.545 & [0.499, 0.568] \\
        2.50 & 0.512 & 0.500 & [0.488, 0.536] \\
        2.75 & 0.518 & 0.500 & [0.494, 0.542] \\
        3.00 & 0.516 & 0.500 & [0.492, 0.540] \\
        3.25 & 0.509 & 0.500 & [0.486, 0.532] \\
\end{tabular}
    \caption{Least-squares estimates of the regret exponent (so that regret scales as $T^{\widehat{\text{slope}}}$) for each $\beta$ with unknown utility. Brackets indicate confidence intervals obtained via cluster bootstrap. Our proposed method does not differ significantly from the theoretical one in terms of the regret exponent.}
    \label{tab:expo_unknown}
\end{table}

\paragraph*{Comparison with other methods}
We conduct additional experiments to compare our method with the kernel-based policy proposed in \cite{fan2022policy} and the DIP policy proposed by \cite{luo2024distribution}. We use the same ground-truth demand function as before with $\beta=2$.  Figure~\ref{fig:regret-comparison} shows the regret growth averaged over $50$ trials. The shaded area represents the standard error. We moderately tune the hyperparameters (bandwidth, exploration length, etc.) for all three policies. More implementation details are left in the Appendix~\ref{sec:experiment-detail}. 
%\g{Need to provide all details, at least in the appendix.}
The result shows that our method significantly outperforms both baselines.

\begin{figure}[htbp]
    \centering
    \includegraphics[width=0.7\textwidth]{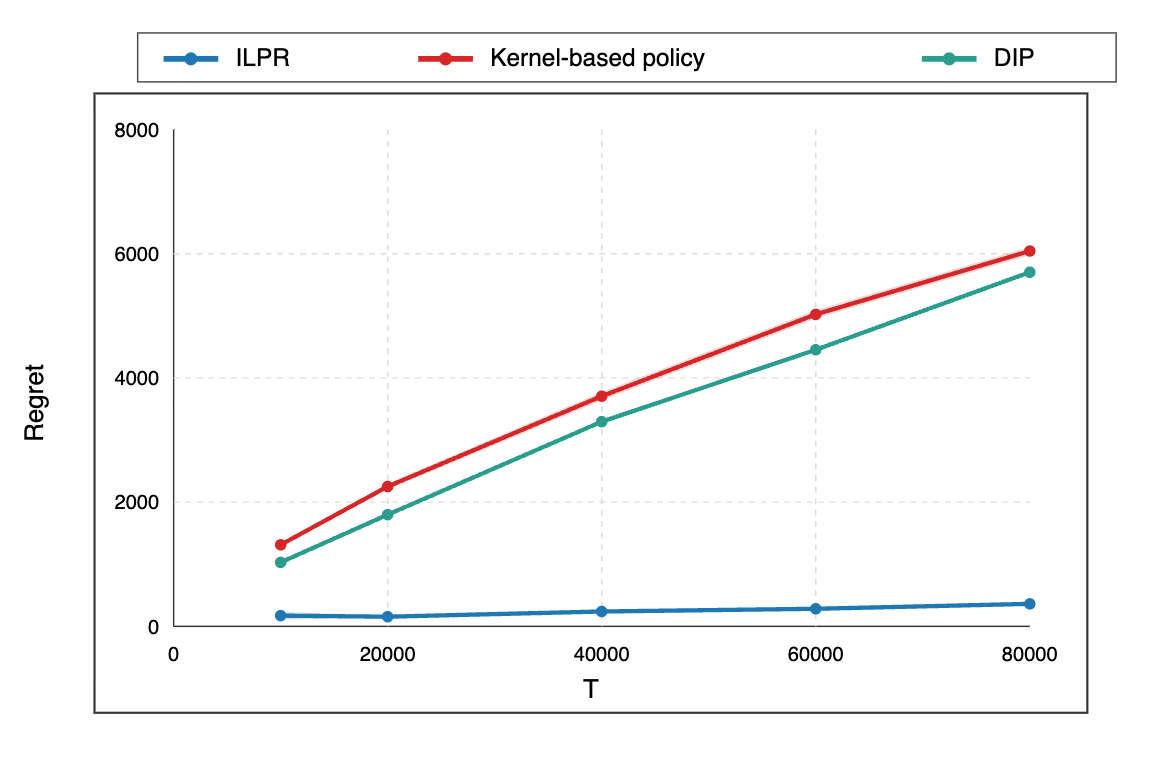}
    \caption{\textbf{Unknown utility.} Regret comparison of ILPR and kernel-based policy~\citep{fan2022policy}in simulation. Our ILPR achieves tremendous benefit.}
    \label{fig:regret-comparison}
\end{figure}

\subsection{Semi-real data}

%\g{Questions below}

We next construct a semi-real online pricing environment using the real data available at the \href{https://sites.google.com/view/dmdaworkshop2023/data-challenge}{INFORMS
2023 BSS Data Challenge Competition}. We define a binary demand indicator that is equal to one if the daily units ordered are non-zero. To ensure reliable estimation, we retain only products with sufficient observations and meaningful variation in purchase behavior.
After an initial screening step, we obtain 161 usable products. For each product $i$, we require that it contains more than 300 observations ($n_i>300$), and has nondegenerate purchase probability:
\[0.05 < \frac{1}{n_i}\sum_{t=1}^{n_i}\mathds{1}(D_t=1) < 0.95.\]
%\g{Would it make more sense to perform the screening before fitting the single-index demand model? Can we re-run the experiment using the new protocol?} 
We then fit a semiparametric single-index demand model of the form $\PP(D=1|\bx,p)=1-F(p-m(\bx))$, where $\bx$ includes six weekday dummy variables, competitor maximum and minimum prices, and stock level. %\g{Why is it a single-index model? Is $m$ linear?} 
The CDF $F$ is constrained to be monotone decreasing. From the data, we first estimate $m(\bx)$ in a linear function class using ridge logistic regression, following the practice of \cite{luo2024distribution} and \cite{bartlett2006convexity}.
%with ridge parameter $1.0$. 
% That is, 
% \[\log
% \bigg(
% \frac{\PP(D=1|\bx,p)}{\PP(D=0|\bx,p)}
% \bigg)=p-m(\bx).\]
% \g{Is this only true when $F$ is logistic?}
After that, we estimate the CDF function $F$ nonparametrically using isotonic regression followed by a Gaussian smoothing step to make $F \in C^2$ on the implied utility residuals $u=p-m(\bx)$.

To evaluate different pricing policies, we create a calibrated online environment for each usable product by resampling observed covariate vectors and feasible price bounds from that product’s empirical history.  We again compare the proposed ILPR policy and the kernel-based baseline and DIP baseline. Performance is measured by average regret relative to the oracle policy induced by the fitted semi-real environment. We report in Figure~\ref{fig:regret-comparison-real} the cumulative regrets averaged over all products at time $T=50,100,150,...,700$. %\g{Are they averaged over all products?}
We run for $50$ trials, and the standard errors are again shown in the shaded regions. ILPR already stands out significantly over this short horizon (less than 2 years). %\g{The results look nice and convincing! You may comment that ILPR already stands out significantly in such a short horizon (less than 2 years).}

\begin{figure}[htbp]
    \centering
    \includegraphics[width=0.7\textwidth]{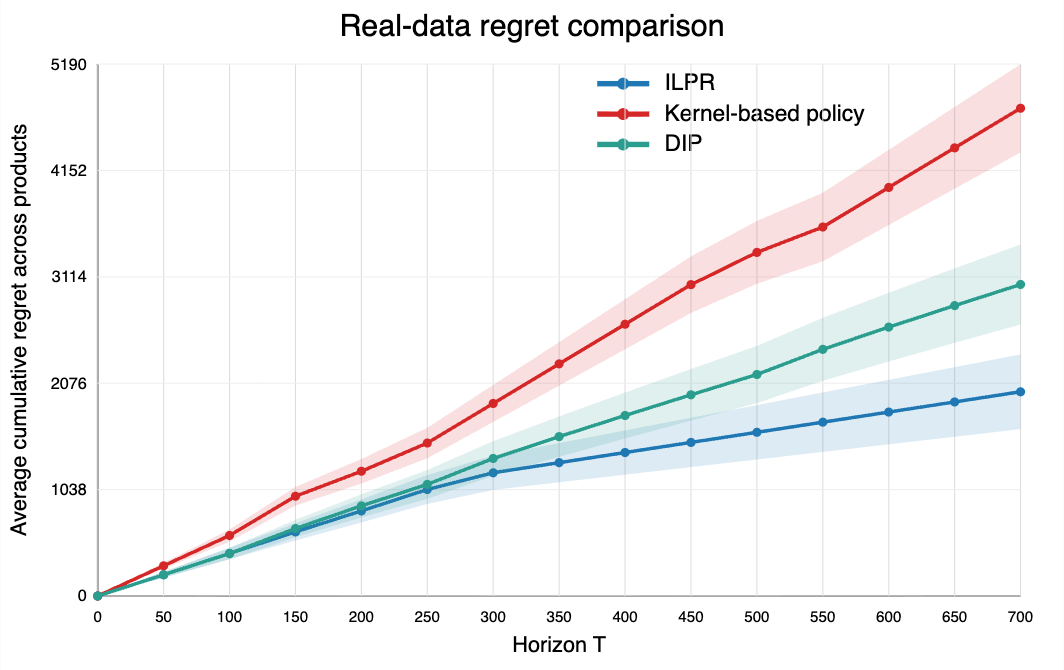}
    \caption{\textbf{Unknown utility.} Regret comparison of ILPR and kernel-based policy~\citep{fan2022policy} in a semi-real environment with real data. Our ILPR outperforms the kernel-based policy and DIP policy.}
    \label{fig:regret-comparison-real}
\end{figure}

For product-level performance comparison, Figure~\ref{fig:histogram} presents the histogram of regret improvement ratios,  $
1 - \operatorname{Regret}(\text{ILPR}) / \operatorname{Regret}(\text{Baseline})$,
%\g{define this quantity}
of our policy against the two baselines. The red bars represent the DIP baseline, and the blue bars represent the kernel-based policy. The histogram clearly shows that our policy achieves substantial improvement over both baselines across all products. 
The mean and median improvements are $66.6\%$ and $74.7\%$ for the kernel-based baseline, and $43.6\%$ and $47.6\%$ for the DIP baseline.
%\g{Provide summary statistics here, rather than in the figure or caption (that would be hard to read).}
\begin{figure}[htbp]
    \centering
    \includegraphics[width=0.7\textwidth]{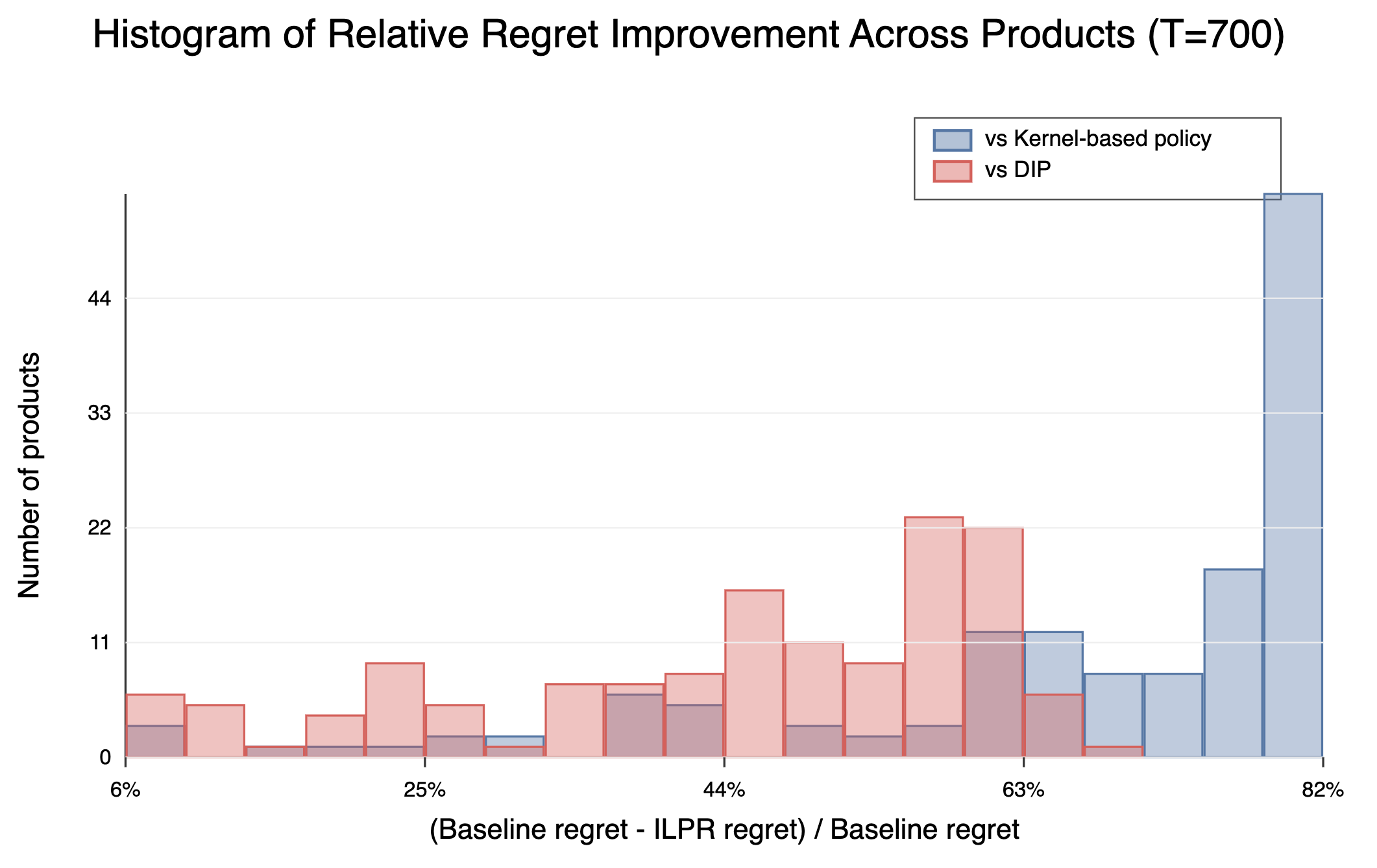}
    \caption{Histogram of regret improvement at $T=700$. The improvement is huge across products. %\g{Change SKU in the figure to product.}
    }
    \label{fig:histogram}
\end{figure}

%\paragraph*{Other implementation details.}
%We choose $x_{\min}=0.35$, $x_{\max}=0.65$, $P_{\min}=0$, $P_{\max}=1$. The LPR bandwidth scales as $h=n^{-1/(2\beta+1)}/2$. For both known and unknown utility cases, the refitting schedule follows a doubling rule with initial block size $T_0=100$. For unknown utility case, we choose $T_{0m}=\lceil\sqrt{4T}\rceil$ for horizon $T$.

%% file: app-Proof_ub.tex
\section{Proof of Upper bound: Theorem~\ref{thm:upper_bound}}
\label{sec:proof_ub}

\subsection{Proof Sketch}
We first give a proof sketch.
\begin{enumerate}
\item Section~\ref{sec:upper_bound_initial_stage}: Initial stage (no boundary decay, so analysis is simpler).
\item Section~\ref{sec:upper_bound_induction_1}: Realized design distribution behaves uniformly locally (we apply concentration, and there exists a loss of a logarithmic factor). 

We also establish that an essential quantity for LPR---the second moment matrix---has a minimum eigenvalue lower bounded by a positive constant.
\item Section~\ref{sec:upper_bound_induction_2}: LPR analysis. We decompose the estimation error into bias, variance, and error from design mismatch.
\item Section~\ref{sec:upper_bound_induction_3}: 
We analyze post-smoothing and perturbation.
\item Section~\ref{sec:upper_bound_regret_analysis}: Final regret analysis by discretizing the design domain.
\end{enumerate}

Before digging into each subsection, we first set up the important notations. Let $[\underline u,\bar u]$ denote the range set of $-\hat m$. Also denote $[\phi^{-1}(\underline u), \phi^{-1}(\bar u)] = [\underline z,\bar z]$, which is the true design domain. Similarly, we define $[\hat\phi_l^{-1}(\underline u), \hat\phi_l^{-1}(\bar u)] = [\underline z_l,\bar z_l]$ as the design domain in stage $l$. However, this is not known to the algorithm. We use an estimated interval $[\underline{\hat z_l},\overline{\hat z_l}
]$, where $\underline{\hat z_l}=\min_{t\in \calT_l}\hat\phi_l^{-1}(-\hat m(\bx_t))$ and $\overline{\hat z_l}=\max_{t\in \calT_l}\hat\phi_l^{-1}(-\hat m(\bx_t))$.

Using that $\bx_t$ are independent and that the density of $-\hat m(\bx)$ satisfies the polynomial boundary decay with exponent $\kappa$ (cf.\ Assumption~\ref{asp:feature-diversity}), one checks that there exists a constant $c>0$ such that, for any $h>0$ and each stage $l$ with $n=|\calT_l|$,
\[
\PP\big(\underline{\hat z_l}-\underline{z_l}\ge h\big)
\;\le\; (1-c\,h^{\kappa+1})^n.
\]
{Choosing $h$ such that $c\,n\,h^{\kappa+1}\ge 3\log T$, we obtain}
\[
{\PP\big(\underline{\hat z_l}-\underline{z_l}\ge h\big)\le \exp(-c n h^{\kappa+1})\le T^{-3}.}
\]
{Equivalently, with probability at least $1-T^{-3}$,}
\begin{align}
{\underline{\hat z_l}-\underline{z_l}\;\le\; C\Big(\frac{\log T}{n}\Big)^{\frac{1}{\kappa+1}},}
\end{align}
{for some constant $C>0$ depending only on the boundary decay parameters.}
A similar inequality holds for $\overline{z_l}-\overline{\hat z_l}$. Hence the boundary error is negligible; more precisely, when $T$ is greater than a large constant~\footnote{When $T$ is smaller than a constant the regret bound is trivial.}, we have
\begin{align}
\label{equ:interval-boundary-error}
{\underline{\hat z_l}-\underline{z_l}\le \min(\delta_v,v)/2,\qquad
\overline{z_l}-\overline{\hat z_l}\le \min(\delta_v,v)/2,}
\end{align}
{and in all subsequent arguments we tacitly replace the theoretical interval $[\underline z_l,\bar z_l]$ by its empirical counterpart $[\underline{\hat z_l},\overline{\hat z_l}]$, at the cost of at most changing constants in the bounds.}
{We summarize this by introducing $\hat v$ and $\hat\delta_x$ as empirical analogues of $v$ and $\delta_x$, respectively, with $\hat v$ and $\hat\delta_x$ linear in $v$ and $\delta_x$.}

In some sections, we rescale the domain as $[0,1]$, as the extension to a general interval is straightforward. {This can always be achieved by the affine map $T(z)=(z-a)/(b-a)$ when the support is $[a,b]$; the boundary decay condition and the quantities $\mu_h(x)$ transform in a way that preserves all inequalities up to multiplicative constants, so we do not distinguish these cases in the sequel.} For notational simplicity, we denote $n=|\calT_l|$ when there is no confusion. Also, we omit the stage index $l$ when there is no confusion.

\subsection{Initial Stage}
\label{sec:upper_bound_initial_stage}
%Link: Algorithm~\ref{alg}.

We start by focusing on the estimation of $\hat F_0$.
The constants $a_1,b_1$ are chosen so that $a_1[\bar u,\underline u]+b_1$ covers the range of $[\underline z-c_1,\bar z+c_1]$. Therefore, 

There is no density decay at the boundary in the initial stage. We set the initial estimator as
\begin{align*}
\hat \phi_0^I(u)=u-\frac{1-\hat F_0(u)}{\hat F_0^{(1)}(u)}.
\end{align*}
And the post-smoothed version as
\begin{align*}
\hat \phi_0^S(u)=\int_{-\infty}^\infty \hat\phi_0^I(x-t)K_\delta(t)dt
\end{align*}
for some fixed smoothing bandwidth $\delta=Cn^{-\frac{\beta-1}{2\beta+1}}$.

Following similar but simpler lines as in Sections~\ref{sec:upper_bound_induction_2} and~\ref{sec:upper_bound_induction_3}, we can establish that with the optimal bandwidth, on the entire range $[\underline z,\bar z]$, 
\begin{align*}
|\hat F_0^{(1)}(u)-F'(u)|&\lesssim |T_0|^{-\frac{\beta-1}{2\beta+1}},\\
|\hat F_0(u)-F(u)|&\lesssim  |T_0|^{-\frac{\beta}{2\beta+1}}.
\end{align*}

In particular, the same error decomposition holds. Analysis of the bias and variance terms is the same as in the general stage. For the design mismatch term, a similar proof goes through since $a_1\hat m(\bx)+b_1$ is Lipschitz with respect to $\hat m(\bx)$.

Hence for $x\in [\underline u,\bar u]$, we have
\begin{align*}
|\hat \phi_0^I(x)-\phi(x)|&\le C|T_0|^{-\frac{\beta-1}{2\beta+1}}:=b_0,
\end{align*}
and
\begin{align*}
|\hat \phi_0^S(x)-\phi(x)|&\le C|T_0|^{-\frac{\beta-1}{2\beta+1}},\\
|(\hat \phi_0^S)^{(1)}(x)-\phi'(x)|&\le c.
\end{align*}

To guarantee the range coverage in the next stage, we again perturb at the boundary and set
\begin{align}
\label{def:perturbation_initial_stage}
\hat\phi_0(x)=\hat\phi_0^S(x)+b_0 (1-2x).
\end{align}
It follows that 
\begin{align*}
\hat\phi_0(0)\ge \phi(0),\qquad \hat\phi_0(1)\le \phi(1),
\end{align*}
facilitating the next stage.

In the following induction steps, all statements are conditioned on $\hat\phi_{l-1}$, which is assumed to satisfy the inductive hypotheses.
The distribution in stage $l$ is supported on $[\underline z_{l-1},\bar z_{l-1}]$ as defined.
Consider stage \(l\). For theoretical generality, we use a uniform notation where the index \(i\in [n]\) represents the time point \(t\in \mathcal{T}_{l}\), with \(n=T_{l}=2^{l}T_{0}\).

\subsection{Induction: Step 1}
\label{sec:upper_bound_induction_1}
%Link: Algorithm~\ref{alg}.

This section develops two prerequisite lemmas for analyzing local polynomial regression and is self-contained.
The main message is that the realized designs are essentially uniformly scattered. 

For simplicity of exposition, we consider the distribution to be supported on $[0,1]$. {As noted above, the case of a general interval $[a,b]$ can be reduced to $[0,1]$ by an affine rescaling of the design; all the bounds below remain valid up to multiplicative constants depending only on the fixed parameters.}

\textbf{Design.} Let $X_1,\ldots,X_n$ be i.i.d. on $[0,1]$ with density $g$.
There exists $c_b>0$ such that for all $t\in[0,1]$,
\[
c_b\,\delta_\kappa(t)\le g(t) \le M,
\]
where $\delta_\kappa(t):=\min\{t,1-t\}^{\kappa}$ and $\kappa\ge 0$ is fixed.
(This single condition implies an interior lower bound:
if $t\in[h,1-h]$ then $\delta_\kappa(t)\ge h^{\kappa}$, hence $g(t)\ge c_b h^{\kappa}$;
no separate “interior” assumption is required.)

\medskip
\textbf{Kernel.} Let $K:\RR\to[0,\infty)$ be bounded, Lipschitz on $[-1,1]$, supported on $[-1,1]$, with
\[
\int_{-1}^{1} K(u)\,du =1.
\]

\medskip
\textbf{Bandwidth.} Fix $h\asymp n^{-\frac{1}{2\beta+1}}$ (independent of $x$).

\medskip
\textbf{Sliding-window count and mass:}
\[
N_h(x):=\sum_{i=1}^n \mathbf{1}\{|X_i-x|\le h\}, \qquad
m_h(x):=\int_{x-h}^{x+h} g(t)\,dt.
\]

\medskip
\textbf{Local minimum density:}
\[
\mu_h(x):=\inf_{u\in [-1,1]} g(x+hu).
\]
We make the regularity condition to prevent the density from changing abruptly, that is, for $x\in [2h,1-2h]$,
\begin{align}
\label{equ:density-regularity}
{\sup_{u\in [-1,1]} g(x+hu)\le C\,\mu_h(x).}
\end{align}
Hence we have $2h\mu_h(x)\le m_h(x)\le 2Ch\mu_h(x)$. {In our application, this regularity follows from the polynomial boundary decay and the smoothness of the pushforward density of $-\hat m(\bx)$; specifically, the density is sandwiched between two multiples of $\delta_\kappa(\cdot)$, which yields~\eqref{equ:density-regularity} for some $C<\infty$.}

\medskip
\textbf{Local polynomial matrices.} For $U(u)=(1,u,u^2/2!,\ldots,u^q/q!)^\top$, define
\[
S_0 := \int_{-1}^{1} U(u)U(u)^\top K(u)\,du,
\]
\[
B_n(x):=\frac{1}{nh}\sum_{i=1}^n K\!\Big(\tfrac{X_i-x}{h}\Big)\,
U\!\Big(\tfrac{X_i-x}{h}\Big)U\!\Big(\tfrac{X_i-x}{h}\Big)^{\!\top}, 
\quad
B(x):=\EE[B_n(x)]=\int_{-1}^{1} U(u)U(u)^\top K(u)\,g(x+hu)\,du.
\]

\medskip

\begin{lemma}[Uniform scattering in sliding windows]
\label{lemma:uniform_scattering}
Under the conditions above, there exist constants $C$ such that if for any $\delta\in(0,1)$ and $v\ge 2h$ one has
\begin{align}
\label{equ:condition_uniform_scattering}
{n\,m_{h/2}(v-h/2)\ \ge\ C\big(\log(1/h)+\log(1/\delta)\big),}
\end{align}
then with probability at least $1-\delta$, it holds that
\[
N_h(x)\ \ge\ c\,n\,m_h(x)
\qquad\text{for all }x\in[v,1-v].
\]

Note that \eqref{equ:condition_uniform_scattering} can be satisfied when we have
\begin{align*}
n\delta_\kappa(v)\ge C\big(\log(1/h)+\log(1/\delta)\big).
\end{align*}
\end{lemma}

\begin{proof}
Fix $x\in[v,1-v]$ and set $Z_i(x):=\mathbf{1}\{|X_i-x|\le h/2\}$. Then $N_{h/2}(x)=\sum_{i=1}^n Z_i(x)$ and $\EE[Z_i(x)]=m_{h/2}(x)$. By Bernstein's inequality,
\[
\PP\!\Big(\big|N_{h/2}(x)-n m_{h/2}(x)\big|\ge t\Big)
\;\le\; 2\exp\!\left(-\frac{t^2}{2n\,\mathrm{Var}(Z_1(x))+\tfrac{2}{3}t}\right)
\;\le\; 2\exp\!\left(-\frac{t^2}{2n\,m_{h/2}(x)+\tfrac{2}{3}t}\right).
\]

Hence applying a union bound and by~\eqref{equ:condition_uniform_scattering}, we have with probability at least $1-\delta$, for all $v-h/2,v+h/2,\cdots,1-v+h/2$ (we assume $v$ can be divided by $h$ as this does not make much difference to the proof),
\[
\Big|\,N_{h/2}(x) - n\,m_{h/2}(x)\,\Big|
\le
C_1\sqrt{n\,m_{h/2}(x)\,(\log(1/h\delta))}+C_2(\log(1/h\delta))\le \frac{1}{2}nm_{h/2}(x).
\]

Now consider a general $x\in [v,1-v]$. Define the integer $k$ to be such that $v+kh\le x\le v+(k+1)h$. It holds that
\begin{align*}
N_{h/2}(v+kh+h/2)\le N_h(x)\le N_{h/2}(v+kh-h/2)+N_{h/2}(v+kh+h/2)+N_{h/2}(v+kh+3h/2).
\end{align*}
Hence we have, by~\eqref{equ:density-regularity},
\begin{align*}
N_h(x)\ge \frac{1}{2}nm_{h/2}(v+kh+h/2)\ge cn m_h(x).
\end{align*}
\end{proof}

\begin{lemma}[Uniform lower bound on $\lambda_{\min}$ of the local-polynomial Gram matrix]
\label{lemma:matrix_min_eigen}
Under the conditions above, there exist constants $c_*,C_1,C_2>0$ (depending only on $p$, $c_b$, and $\kappa$) such that for any $\delta\in(0,1)$ and $v\ge 2h$, with probability at least $1-\delta$, for all $x\in [v,1-v]$,
\[
\ \lambda_{\min}\big(B_n(x)\big)
\ \ge\
c_*\,\mu_h(x)-
C_1\,\sqrt{\frac{\big(\log(1/h)+\log(1/\delta)\big)\mu_h(x)}{nh}}-
C_2\,\frac{\log(1/h)+\log(1/\delta)}{nh}.
\]
In particular, if
\[
nh\mu_h(x)\ge C\,\big(\log(1/h)+\log(1/\delta)\big)\qquad\text{for all }x\in [v,1-v],
\]
then, with probability at least $1-\delta$,
\[
\lambda_{\min}\!\big(B_n(x)\big)\ \ge\ \tfrac{1}{2}\,c_*\,\mu_h(x).
\]
\end{lemma}

\begin{proof}
We first show that $S_0$ is positive definite. For any nonzero vector $v\in\mathbb{R}^d$,
\[
v^\top S_0 v
= \int_{-1}^1 \bigl(U(u)^\top v\bigr)^2 K(u)\,du.
\]
Since $U(u)^\top v$ is a nonzero polynomial in $u$, it has only finitely many zeros.
Moreover, $K(u)\ge 0$ and is strictly positive on a set of positive Lebesgue measure.
Therefore, $(U(u)^\top v)^2 K(u) > 0$ for almost every $u\in[-1,1]$, which implies
\[
v^\top S_0 v > 0.
\]
Hence, $S_0$ is positive definite.

Also by $K\ge 0$,
\[
B(x)=\int U(u)U(u)^\top K(u)\,g(x+hu)\,du
\ \succeq\
\big(\inf_{u\in[-1,1]} g(x+hu)\big)\,S_0.
\]
Hence there exists $c_*>0$ such that
\[
\lambda_{\min}\!\big(B(x)\big)\ \ge\ c_*\,\mu_h(x)\qquad\text{for all }x\in[0,1].
\]

\emph{Step 1.}
We can apply Matrix Bernstein.
Let
\[
W_i(x):=\frac{1}{h}\,K\!\Big(\frac{X_i-x}{h}\Big)\,U\!\Big(\frac{X_i-x}{h}\Big)U\!\Big(\frac{X_i-x}{h}\Big)^{\!\top}.
\]
Then $B_n(x)-B(x)=\frac{1}{n}\sum_{i=1}^n\big(W_i(x)-\EE W_i(x)\big)$, and
\[
\|W_i(x)\|_{2}\ \le\ \frac{\|K\|_\infty}{h}\,\sup_{|u|\le 1}\|U(u)\|_2^2:= \frac{L_U}{h}.
\]
Using $W_i(x)^2\preceq \|W_i(x)\|_{2}\,W_i(x)$,
\[
\Big\|\EE\big[(W_i(x)-\EE W_i(x))^2\big]\Big\|_{2}
\ \le\ \Big\|\EE W_i(x)^2\Big\|_{2}
\ \le\ \frac{L_U}{h}\,\Big\|\EE W_i(x)\Big\|_{2}
\ \le\ \frac{C}{h}\,\mu_h(x),
\]
for a constant $C$ depending only on $p$ and $K$. The matrix Bernstein inequality (e.g. \cite{tropp2015introduction}) yields, for any $t>0$,
\[
\PP\!\left(\big\|B_n(x)-B(x)\big\|_{2}\ge t\right)
\ \le\ (p+1)\,\exp\!\left(-\frac{n\,t^2}{C\,h^{-1}\mu_h(x) + (L_U/h)\,t}\right).
\]
Choosing appropriate $t$ gives, with probability at least $1-\delta$,
\[
\big\|B_n(x)-B(x)\big\|_{2}
\ \le\
C_1\,\sqrt{\frac{\log(1/\delta)\,\mu_h(x)}{nh}}
\;+\;
C_2\,\frac{\log(1/\delta)}{nh}.
\]

\emph{Step 2.}
To obtain a uniform-type result,
cover $[v,1-v]$ by a grid $\mathcal{G}$ of spacing $h/(4n)$. Apply the pointwise bound on $\mathcal{G}$ and take a union bound, we get that for any $x\in \cal{G}$,
\[
\big\|B_n(x)-B(x)\big\|_{2}
\ \le\
C_1\,\sqrt{\frac{\log\big(1/(h\delta)\big)\,\mu_h(x)}{nh}}
\;+\;
C_2\,\frac{\log\big(1/(h\delta)\big)}{nh}.
\]
For an arbitrary $x$, choose $x_0\in\mathcal{G}$ with $|x-x_0|\le h/(8n)$. Since $K$ is Lipschitz on $[-1,1]$ with constant $L_K$,
\[
\|W_i(x)-W_i(x_0)\|_{2}
\ \le\ \frac{C'}{h^2}\,|x-x_0|\,\mathbf{1}\{|X_i-x_0|\le \tfrac{3h}{2}\},
\]
for a constant $C'$ depending on $p$ and $K$. Summing and dividing by $n$, the off-grid increment is of order at most
\[
\big\|B_n(x)-B_n(x_0)\big\|_{2}
\ \le\ \frac{C'}{n h^2}\,|x-x_0|\,N_{3h/2}(x_0)\lesssim \frac{1}{nh},
\]
and is absorbed by enlarging $C_1,C_2$ in the Bernstein bound.

\emph{Step 3.}
We can now apply Weyl's inequality.
For all $x\in [v,1-v]$,
\begin{align*}
\lambda_{\min}\!\big(B_n(x)\big)
\ &\ge\ \lambda_{\min}\!\big(B(x_0)\big) - \big\|B_n(x_0)-B(x_0)\big\|_{2}
- \big\|B_n(x)-B(x_0)\big\|_{2}\\
&\ge\ c_*\,\mu_h(x) - \Big\{ C_1\sqrt{\tfrac{(\log(1/h)+\log(1/\delta))\,\mu_h(x)}{n h}} + C_2\tfrac{\log(1/h)+\log(1/\delta)}{n h}\Big\}.
\end{align*}
If in addition, $n h\,\mu_h(x)\ge C(\log(1/h)+\log(1/\delta))$ for all $x\in [v,1-v]$, the bracketed term is at most $\tfrac{1}{2}c_*\,\mu_h(x)$, proving the claim. 
\end{proof}

\paragraph{Choice of $v$}
Take $v$ as in \eqref{def:v-explicit} below. It satisfies for some large constant $C$,
\begin{align}
\label{def:v}
nh^3\delta_\kappa(v)\ge C\log T.
\end{align}
This definition of $v$ satisfies the precondition of both lemmas when we take $\delta=1/T^2$ in both lemmas.

\subsection{Induction: Step 2}
\label{sec:upper_bound_induction_2}
%Link: Algorithm~\ref{alg}.

As usual, consider a fixed stage $l$. We aim to deploy classical results in \cite{Tsybakov2009} to analyze $\hat\phi^I_l$. Recall the observed design and true design $u_i=\hat \phi_{l-1}^{-1}(-\hat m(\bx_i))$ and $\tilde u_i=p_t-m(\bx_t)=\hat \phi_{l-1}^{-1}(-\hat m(\bx_i))+\hat m(\bx_i)-m(\bx_i)$.

Recall the density of $u_i$ is defined as $\mathrm{d}_{l-1}$. Since the derivative of $\hat\phi_{l-1}$ is bounded, and by Assumption~\ref{asp:feature-diversity}, we have $\mu(x)\asymp \delta_{[\underline{ z_{l-1}},\overline{ z_{l-1}}]}(z)$.

By \eqref{equ:interval-boundary-error}, the same density decay rate holds with the estimated interval $[\underline{\hat z_{l-1}},\overline{\hat z_{l-1}}]$, for $x\in I_v[\underline{\hat z_{l-1}},\overline{\hat z_{l-1}}]$. Here recall that $I_v[\underline z,\overline z]:=[\underline z+v(\overline z-\underline z),\underline z+(1-v)(\overline z-\underline z)]$, and $\delta_{[\underline{u},\bar{u}]}(z):=\min\{(z-\underline{u})^{\kappa},(\bar u-z)^{\kappa}\}$. That is to say,  $\mathrm{d}_{l-1}(x)\asymp \delta_{[\underline{\hat z_{l-1}},\overline{\hat z_{l-1}}]}(z)$. In other words, the boundary error can be neglected.

Denote the linear smoothing weights for estimating the first-order derivative as $w^1(x,u_i)$. To be specific,
\begin{align*}
B_n(x)&=\frac{1}{nh}\sum_{i=1}^n K\left(\frac{u_i-x}{h}\right)U\left(\frac{u_i-x}{h}\right) U\left(\frac{u_i-x}{h}\right)^{\top},\\
q_n(x)&=\frac{1}{nh}\sum_{i=1}^n D_i K\left(\frac{u_i-x}{h}\right)U\left(\frac{u_i-x}{h}\right).
\end{align*}

The complete local solution of LPR is given by
\begin{align*}
\hat\theta_n(x)&=\arg\min \sum_{i=1}^n \left[D_i-\theta^{\top} U\left(\frac{u_i-x}{h}\right)\right]^2 K\left(\frac{u_i-x}{h}\right)\\&=B_n(x)^{-1}q_n(x),
\end{align*}
as long as $B_n(x)\succ 0$.

Then we have 
\begin{align*}
\hat F_l(x)=1-e_1^{\top} \hat\theta_n(x),\qquad 
\hat F_l^{(1)}(x)=-e_2^{\top} \hat\theta_n(x)/h.
\end{align*}

Define the linear smoothing weights for the first-order derivative as follows:
\begin{align*}
w^1(x,u_i)=\frac{1}{nh^2}e_2^{\top} B_n^{-1}(x)K\left(\frac{u_i-x}{h}\right)U\left(\frac{u_i-x}{h}\right).
\end{align*}
Then
\begin{align*}
\hat F_l^{(1)}(x)=-\sum_{i=1}^n D_i w^1(x,u_i).
\end{align*}

We only consider the estimation error of the derivative for now, which has a slower rate.
We have the following key error decomposition:
\begin{align*}
\hat F_l^{(1)}(x)-F'(x)
&=\overbrace{\sum_{i=1}^n F(u_i)w^1(x,u_i)-F'(x)}^{T^\textrm{b}(x)}
-\overbrace{\sum_{i=1}^n \tilde e_i w^1(x,u_i)}^{T^\textrm{v}(x)}
+\overbrace{\sum_{i=1}^n \left(F(\tilde u_i)-F(u_i)\right)w^1(x,u_i)}^{T^\textrm{m}(x)}.
\end{align*}
Here we define the error $\tilde e_i=D_i-\big(1-F(\tilde u_i)\big)$. {Conditional on the designs $\{u_i\}_{i=1}^n$, the variables $\tilde e_i$ are independent, mean-zero, and uniformly bounded; all applications of Bernstein's inequality below are to such conditional sums.}

We proceed to present a lemma on the weights. It heavily hinges on Lemmas~\ref{lemma:uniform_scattering} and~\ref{lemma:matrix_min_eigen}. Built upon it, we can analyze the bias term and the variance term. The proof is in Appendix~\ref{sec:proof_lemma_weights}. For simplicity of exposition, we consider $[\underline{\hat z_{l-1}},\overline{\hat z_{l-1}}]=[0,1]$, and $I_v[\underline{\hat z_{l-1}},\overline{\hat z_{l-1}}]=[v,1-v]$.

\begin{lemma}
\label{lemma:weights}
Under the conditions of Lemmas~\ref{lemma:uniform_scattering} and~\ref{lemma:matrix_min_eigen}, the weights satisfy the following for all $x\in [v,1-v]$:
\begin{enumerate}
\item[(i)] $\sup_{i}|w^1(x,u_i)|\le \frac{C_*}{nh^2\mu_h(x)}$, $\sup_{i}|w^0(x,u_i)|\le \frac{C_*}{nh\mu_h(x)}$.
\item[(ii)] $\sum_{i=1}^n |w^1(x,u_i)|\le \frac{C_*}{h}$, $\sum_{i=1}^n |w^0(x,u_i)|\le C_*$.
\item[(iii)] $w^1(x,u_i)=w^0(x,u_i)=0\ \text{if}\ |u_i-x|\ge h$.
\item[(iv)] $w^1(x,u_i)$ is Lipschitz in $x$ with constant at most $C n$.~\footnote{More precisely, one obtains a bound of order $C/(h^2)$ from differentiating the kernel and $B_n^{-1}(x)$, and under our bandwidth choice $h\asymp n^{-1/(2\beta+1)}$ this implies $1/h^2\lesssim n$, so stating an $O(n)$ bound keeps the notation simpler without affecting subsequent rates.}
\end{enumerate}
\end{lemma}

\paragraph{Bias term}
We analyze the bias term. By the $\beta$-smooth condition,
\begin{align*}
|T^\textrm{b}(x)|
&=\left|\sum_{i=1}^n F(u_i)w^1(x,u_i)-F'(x)\right|\\
&=\left|\sum_{i=1}^n \left(F(u_i)-F(x)\right)w^1(x,u_i)\right|\\
&\le\sum_{i=1}^n \frac{|F^{(q)}(x+\tau_i(u_i-x))-F^{(q)}(x)|}{q!}(u_i-x)^q w^1(x,u_i)\quad \text{where} \ \tau_i\in [0,1]\\
&\le \sum_{i=1}^n L\frac{|u_i-x|^\alpha}{q!}|u_i-x|^q |w^1(x,u_i)|\\
&\le L\frac{h^\beta}{q!}\sum_{i=1}^n |w^1(x,u_i)|\\
&\lesssim h^{\beta-1}.
\end{align*}

\paragraph{Variance term}
We analyze the variance term. We aim for a uniform-type result. 
By Lemma~\ref{lemma:weights}, it holds that
\begin{align*}
\frac{1}{n}\sum_{i=1}^n\EE\left[(\tilde e_i w^1(x,u_i))^2\mid\{u_j\}_{j=1}^n\right]
\le \frac{1}{n}\sup_{i}|w^1(x,u_i)|\sum_{i=1}^n|w^1(x,u_i)|
\lesssim \frac{1}{n^2h^3\mu_h(x)}.
\end{align*}

Using Bernstein's inequality (conditionally on $\{u_i\}_{i=1}^n$), for fixed $x\in I_v[\underline{\hat z_{l-1}},\overline{\hat z_{l-1}}]$, we have with probability at least $1-\delta$,
\begin{align*}
|T^\textrm{v}(x)|=\left|\sum_{i=1}^n \tilde e_i w^1(x,u_i)\right|
\lesssim \sqrt{n\cdot \frac{\log(1/\delta)}{n^2h^3\mu_h(x)}}+\frac{\log(1/\delta)}{nh^2\mu_h(x)}.
\end{align*}

To get a uniform-type result, consider discretizing $I_v[\underline{\hat z_{l-1}},\overline{\hat z_{l-1}}]$ into $T^3$ equidistant grid points $\calG$, and apply a union bound together with the Lipschitzness of $w^1(x,u_i)$.
Specifically, for any $x\in I_v[\underline{\hat z_{l-1}},\overline{\hat z_{l-1}}]$, consider the nearest $x_0\in \calG$ to $x$. As $\sqrt{T}\le n\le T$ and by Lemma~\ref{lemma:weights}(iv),
\begin{align*}
\left|T^\textrm{v}(x_0)-T^\textrm{v}(x)\right| &\le \sum_{i=1}^n\left|w^1(x_0,u_i)-w^1(x,u_i)\right|\\
&\le C n |x-x_0|\le \frac{1}{T}
\end{align*}
for an appropriate choice of the grid size.

Therefore, with probability at least $1-\delta$, it holds that for any $x\in I_v[\underline{\hat z_{l-1}},\overline{\hat z_{l-1}}]$,
\begin{align*}
|T^\textrm{v}(x)|
&\lesssim \sqrt{n\cdot \frac{\log(T^3/\delta)}{n^2h^3\mu_h(x_0)}}+\frac{\log(T^3/\delta)}{nh^2\mu_h(x_0)}+\frac{1}{T}\\
&\lesssim \sqrt{n\cdot \frac{\log(T^3/\delta)}{n^2h^3\mu_h(x)}}+\frac{\log(T^3/\delta)}{nh^2\mu_h(x)}+\frac{1}{T}.
\end{align*}
Taking $\delta=T^{-3}$, the right-hand side becomes $\sqrt{\frac{\log T}{nh^3\mu_h(x)}}$ up to constants.

\paragraph{Mismatch term}
We analyze the design-misalignment error at a fixed $x$ first.
\[
|T^\textrm{m}(x)|=\sum_{i=1}^n w^1(x,u_i)\left[F(\tilde u_i)-F(u_i)\right],
\]
where we recall $w^1(x,u_i)$ are the weights of linear smoothing for first-derivative estimation. Although we can only bound $\sum_{i=1}^n |w^1(x,u_i)|\lesssim \frac{1}{h}$, tighter control is available using Taylor expansion when mild smoothness is satisfied, since we have $\sum_{i=1}^n w^1(x,u_i)(x-u_i)^k=\mathbf{1}(k=1)$ for $k\le p$.

Specifically, recall $\hat\Delta:=\hat m-m$, and note that $\beta\ge 2$ and weights are zero outside the local band $[x-h,x+h]$. We can rewrite $T^\textrm{m}(x)$ as 
\begin{align*}
T^\textrm{m}(x)
&=\sum_{i=1}^n w^1(x,\hat \phi_{l-1}^{-1}\circ (-\hat m)(\bx_i))
\Big[F\big(\hat \phi_{l-1}^{-1}\circ (-\hat m)(\bx_i)\big)-F\big(\hat \phi_{l-1}^{-1}\circ (-\hat m)(\bx_i)+\hat m(\bx_i)-m(\bx_i)\big)\Big]\\
&=\sum_{i=1}^n w^1(x,\hat \phi_{l-1}^{-1}\circ (-\hat m)(\bx_i))
F'\big(\hat \phi_{l-1}^{-1}\circ (-\hat m)(\bx_i)\big)\Big(\hat m(\bx_i)-m(\bx_i)\Big)+R_2\\
&=\sum_{i=1}^n w^1(x,\hat \phi_{l-1}^{-1}\circ (-\hat m)(\bx_i))
F'\big(\hat \phi_{l-1}^{-1}\circ (-\hat m)(\bx_i)\big)\hat \Delta(\bx_i)\,\mathds{1}\big(x-h\le \hat \phi_{l-1}^{-1}\circ (-\hat m)(\bx_i)\le x+h\big)+R_2.
\end{align*}
Note that we have defined $\mathfrak{m}(t)$ as the conditional mean of $\hat \Delta(\bx_i)/\|\hat\Delta\|_\infty$ given $\hat m(\bx_i)=t$. The misalignment error can be further written as 
\begin{align*}
T^\textrm{m}(x)
&=\overbrace{\|\hat\Delta\|_\infty\sum_{i=1}^n w^1(x,u_i)F'(u_i)\mathfrak{m}(\hat m(\bx_i))\mathds{1}\big(x-h\le u_i\le x+h\big)}^{T_1^\textrm{m}(x)}
+\overbrace{\|\hat\Delta\|_\infty\sum_{i=1}^n \delta_i(x)}^{T_2^\textrm{m}(x)}+R_2,
\end{align*}
where $\delta_i(x):=w^1(x,u_i)F'(u_i)\left(\hat \Delta(\bx_i)/\|\hat\Delta\|_\infty-\mathfrak{m}(\hat m(\bx_i))\right)$. The second-order remainder term can be bounded by
\begin{align*}
|R_2|\lesssim \frac{\|\hat m-m\|^2_\infty}{h}\le \|\hat\Delta\|_\infty.
\end{align*} 
We proceed to bound $T_1^\textrm{m}(x)$ and $T_2^\textrm{m}(x)$ separately.

Denote $q=F'\cdot(\mathfrak{m}\circ \hat\phi_{l-1})$. By Assumption~\ref{asp:conditional-mean-Lip}, $q$ is $\tilde L$-Lipschitz for some constant $\tilde L$. It follows that for any $x\in I_v[\underline{\hat z_{l-1}},\overline{\hat z_{l-1}}]$,
\begin{align*}
|T_1^\textrm{m}(x)|
&=\|\hat\Delta\|_\infty\Big|\sum_{i=1}^n w^1(x,u_i)q(u_i)\mathds{1}\big(x-h\le u_i\le x+h\big)\Big|\\
&=\|\hat\Delta\|_\infty\Big|\sum_{i=1}^n w^1(x,u_i)\big(q(u_i)-q(x)\big)\mathds{1}\big(x-h\le u_i\le x+h\big)\Big|
\\
&\le \|\hat\Delta\|_\infty\tilde L\sum_{i=1}^n \big|w^1(x,u_i)\big|\cdot \big|u_i-x\big|\mathds{1}\big(x-h\le u_i\le x+h\big)\\
&\lesssim  \|\hat\Delta\|_\infty.
\end{align*}

Note that for fixed $x$ and considering only the randomness of $\bx_i$, the $\delta_i(x)$ are independent and have zero mean. We can show
\begin{align*}
\sum_{i=1}^n\EE\big[\delta_i(x)^2\mid\{u_j\}_{j=1}^n\big]
&\lesssim \sup_i|w^1(x,u_i)|\cdot\sum_{i=1}^n |w^1(x,u_i)|
\lesssim \frac{1}{nh^3\mu_h(x)},\\
|\delta_i(x)|&\le \sup_i |w^1(x,u_i)|\lesssim \frac{1}{nh^2\mu_h(x)}.
\end{align*}

By Bernstein's inequality, it holds with probability at least $1-\delta$,
\begin{align*}
\sum_{i=1}^n\delta_i(x)
&\lesssim \sqrt{\log (1/\delta)\cdot \sum_{i=1}^n\EE\delta_i(x)^2}+\frac{\log(1/\delta)}{nh^2\mu_h(x)}\\
&\lesssim \sqrt{\frac{\log (1/\delta)}{nh^3\mu_h(x)}}+\frac{\log(1/\delta)}{nh^2\mu_h(x)}.
\end{align*}

To get a uniform-type result, we take the same discretization approach as analyzing the variance term. Consider discretizing $I_v[\underline{\hat z_{l-1}},\overline{\hat z_{l-1}}]$ into $T^3$ equidistant grid points $\calG$.
For any $x\in I_v[\underline{\hat z_{l-1}},\overline{\hat z_{l-1}}]$, let $x_0\in \calG$ be the closest grid point to $x$.
It holds that
\begin{align*}
|T_2^\textrm{m}(x)-T_2^\textrm{m}(x_0)|
&\lesssim \|\hat\Delta\|_\infty\sum_{i=1}^n \big|w^1(x,u_i)-w^1(x_0,u_i)\big|\\
&\le C n \|\hat\Delta\|_\infty |x-x_0|
\le \|\hat\Delta\|_\infty,
\end{align*}
for a suitable grid size and constant $C$.

Hence, taking a union bound, we have with probability at least $1-\delta$,
\begin{align*}
|T^{\textrm{m}}(x)|
\lesssim\|\hat\Delta\|_\infty\left(1+\sqrt{\frac{\log (T^3/\delta)}{nh^3\mu_h(x)}}+\frac{\log(T^3/\delta)}{nh^2\mu_h(x)}\right).
\end{align*}

Taking $\delta=T^{-3}$, from~\eqref{def:v}, we obtain with probability at least $1-T^{-3}$, for any $x\in I_v[\underline{\hat z_{l-1}},\overline{\hat z_{l-1}}]$,
\begin{align*}
|T^\textrm{m}(x)|\lesssim \|\hat\Delta\|_\infty.
\end{align*}

Combining the three sources of error, we obtain with probability at least $1-2T^{-3}$, for all $x\in I_v[\underline{\hat z_{l-1}},\overline{\hat z_{l-1}}]$,
\begin{align*}
|\hat F_l^{(1)}(x)-F'(x)|\lesssim \epsilon_m+\sqrt{\frac{\log T}{nh^3\mu_h(x)}}+h^{\beta-1}.
\end{align*}

\subsection{Induction: Step 3}
\label{sec:upper_bound_induction_3}
%Link: Algorithm~\ref{alg}.

Note that 
\begin{align*}
|\hat\phi_{l}^I(x)-\phi(x)|&=
\left|\frac{1-\hat F_l(x)}{\hat F^{(1)}_l(x)}-\frac{1-F(x)}{F'(x)}\right|\\
&\lesssim |\hat F_l(x)-F(x)|+ |\hat F_l^{(1)}(x)- F'(x)|.
\end{align*}
We have controlled the estimation error of $\hat F_l$ and $\hat F^{(1)}_l$, hence of $\hat\phi_{l}^I$. Now we are ready to analyze the post-smoothing convolution and perturbation to obtain~\eqref{equ:induction_regularity} at stage $l$. Meanwhile, we prove the same estimation error rate is preserved for $\hat\phi_l$.

Denote $e_l(x):=\hat \phi^I_l(x)-\phi(x)$. Recall the relative distance of $x$ in interval $[a,b]$ as $\alpha_{[a,b]}(x)=\min\{\frac{x-a}{b-a},\frac{b-x}{b-a}\}$. With no ambiguity, we also abbreviate $\alpha(x):=\alpha_{[\underline{\hat z_{l-1}},\overline{\hat z_{l-1}}]}(x)$ here.
For any $x\in I_v[\underline{\hat z_{l-1}},\overline{\hat z_{l-1}}]$, we have
\begin{align*}
|e_l(x)|\le Cn^{-\frac{\beta-1}{2\beta+1}}\sqrt{\frac{\log T}{\alpha(x)^\kappa}}+\epsilon_m:=b(x).
\end{align*}

In this subsection, we aim to prove that, after post-smoothing, $\hat \phi_{l}$ satisfies 
\begin{subequations}
\label{equ:induction_regularity}
\begin{align}
\hat \phi_{l}'\in [c,C] \quad\text{for some constants } \ C>c>0,\\
\hat\phi^{-1}_{l}(\underline u)\le \phi^{-1}(\underline u),\\
\hat\phi^{-1}_{l}(\bar u)\ge \phi^{-1}(\bar u),
\end{align}
\end{subequations}
where we recall that $[\underline u,\bar u]$ is the range of $-\hat m(\bx)$ and $[\phi^{-1}(\underline u),\phi^{-1}(\bar u)]=[\underline z,\bar z]$.

\paragraph{Post-smoothing}
The post-smoothing step is
\begin{align*}
\hat\phi_l^S(x)=\int_{-\infty}^\infty \hat \phi_l^I(x-t)K_{\delta_x}(t)dt.
\end{align*}
Note that here the bandwidth $\delta_x$ as a function of $x$ is location-dependent. We will decide later how to choose it adaptively from data.

We can decompose the estimation error as follows:
\begin{align*}
\hat\phi_l^S(x)-\phi(x)
&=\overbrace{\int_{-\infty}^\infty \left(\hat \phi_l^I(x-t)-\phi(x-t)\right)K_{\delta_x}(t)dt}^{T_1} 
+ \overbrace{\int_{-\infty}^\infty \Big(\phi(x-t)- \phi(x)\Big)K_{\delta_x}(t)dt}^{T_2}.
\end{align*}

We bound the two terms separately as
\begin{align*}
|T_1|\le \int_{-\infty}^\infty |e_l(x-t)| K_{\delta_x}(t)dt
\le \sup_{z\in [x-\delta_x,x+\delta_x]}|e_l(z)|
\end{align*}
and
\begin{align*}
|T_2|\lesssim \int_{-\infty}^\infty |t|K_{\delta_x}(t)dt\le\delta_x.
\end{align*}

As for the derivative of $\hat\phi_l^S(x)$, it holds that $(\hat\phi_l^S)'(x)=T_3+T_4+T_5$, where
\begin{align*}
T_3&=\int_{-\infty}^\infty \left((\hat\phi_l^I)'(x-t)-\phi'(x-t)\right)K_{\delta_x}(t)dt, \\
T_4&=\int_{-\infty}^\infty \phi'(x-t)K_{\delta_x}(t)dt,\\
T_5&=\int_{-\infty}^{\infty} \left(\hat \phi_l^I(x-t)-\phi(x-t)+\phi(x-t)-\phi(x)\right)\frac{dK_{\delta_x}(t)}{dx}dt.
\end{align*}

It holds that
\begin{align*}
|T_3|
&=\left|\int_{-\infty}^{\infty} e_l(x-t)K_{\delta_x}'(t)dt \right|\\
&\le \left|\int_{-\infty}^{\infty} e_l(x-t)\frac{1}{\delta_x^2}K'\left(\frac{t}{\delta_x}\right)dt \right|\\
&\le \sup_{z\in [x-\delta_x,x+\delta_x]}|e_l(z)|\cdot \left|\int_{-\infty}^{\infty} \frac{1}{\delta_x}\Big|K'\left(\frac{t}{\delta_x}\right)\Big|d\left(\frac{t}{\delta_x}\right) \right|\\
&\lesssim \frac{\sup_{z\in [x-\delta_x,x+\delta_x]}|e_l(z)|}{\delta_x}.
\end{align*}
And
\begin{align*}
c\le\inf_{z\in [x-\delta_x,x+\delta_x]}\phi'(z)\le T_4\le \sup_{z\in [x-\delta_x,x+\delta_x]}\phi'(z)\le C.
\end{align*}

Note that we have
\begin{align*}
\frac{dK_{\delta_x}(t)}{dx}
&=\frac{d\delta_x}{dx}\cdot\left[-\frac{1}{\delta_x^2}K\left(\frac{t}{\delta_x}\right)-\frac{t}{\delta_x^3}K'\left(\frac{t}{\delta_x}\right)\right].
\end{align*}
It follows that as long as, for some small constant $c$, it holds that $|\frac{d\delta_x}{dx}|\le c$ for every $x\in I_v[\underline{\hat z_{l-1}},\overline{\hat z_{l-1}}]$, we have
\begin{align*}
|T_5|
&\le\int_{-\infty}^{\infty} \left(|e_l(x-t)|+|t|\right)\left|\frac{dK_{\delta_x}(t)}{dx}\right|dt\\
&\le c\left(\sup_{z\in [x-\delta_x,x+\delta_x]}|e_l(z)|+\delta_x\right)\cdot \int_{-\infty}^{\infty}\left|\frac{1}{\delta_x^2}K\left(\frac{t}{\delta_x}\right)+\frac{t}{\delta_x^3}K'\left(\frac{t}{\delta_x}\right)\right|dt\\
&\le
c\left(\frac{\sup_{z\in [x-\delta_x,x+\delta_x]}|e_l(z)|}{\delta_x}+1\right).
\end{align*}

For every $x\in I_v[\underline{\hat z_{l-1}},\overline{\hat z_{l-1}}]$, take
\begin{align}
\label{def:delta_x}
\delta_x:=\bar\delta_{\alpha(x)}
=Cn^{-\frac{\beta-1}{2\beta+1}}\sqrt{\frac{\log T}{\alpha(x)^\kappa}}+\epsilon_m
\end{align}
for a constant $C$ large enough, enabling the first term in the above formula to be bounded by a constant. Here, $v$ is defined by
\begin{align}
\label{def:v-explicit}
\quad v=(C^2n^{-\frac{\beta-1}{2\beta+1}}\log^{1/2}T)^{\frac{2}{\kappa+2}}+C\epsilon_m,
\end{align}
and it holds that $v\ge  C\bar\delta_{v}$. Combining the analysis of the various terms above, there exist some constants $c_1,C_1$ such that 
\[
(\hat\phi_l^S)'(x), \phi'(x)\in [c_1,C_1]\quad \text{for all} \quad x\in I_v[\underline{\hat z_{l-1}},\overline{\hat z_{l-1}}].
\]

\paragraph{Perturbation}
In the perturbation step, for $x\in I_v[\underline{\hat z_{l-1}},\overline{\hat z_{l-1}}]$, we set
\begin{align*}
\hat\phi_l(x)=\hat\phi_l^S(x),
\end{align*} 
while for $x\in (-\infty,l_v[\underline{\hat z_{l-1}},\overline{\hat z_{l-1}}]]\cup [r_v[\underline{\hat z_{l-1}},\overline{\hat z_{l-1}}],\infty)$, we use linear extrapolation as the perturbation with slope $c_1/2$ (where we recall that $c_1$ is the lower bound of $(\hat\phi_l^S)'(x)$). 
In a nutshell, without perturbation, the points near the boundary may never be sampled, and cumulatively this incurs large regret.
The main purpose of this step is to ensure the last two inequalities in~\eqref{equ:induction_regularity}.

Specifically, define
\begin{subequations}
%\label{def:perturbation}
\begin{align}
\hat\phi_l(x)=\begin{cases}
\hat\phi_l^S(v_1)+\frac{c_1}{2}(x-v_1),\quad&\text{if}\ x\in (-\infty,v_1],\\
\hat\phi_l^S(x), &\text{if}\ x\in [v_1,v_2],\\
\hat\phi_l^S(v_2)+\frac{c_1}{2}(x-v_2),\quad&\text{if}\ x\in [v_2,\infty),
\end{cases}
\end{align}
\end{subequations}
where we abbreviate $v_1=l_v[\underline{\hat z_{l-1}},\overline{\hat z_{l-1}}]$, $v_2=r_v [\underline{\hat z_{l-1}},\overline{\hat z_{l-1}}]$.

We claim that the perturbation satisfies the following properties:

\paragraph{Claim 1.} It is an increasing function with derivative bounded in $[c_1/2,C_1]$.

\paragraph{Claim 2.} The range coverage property is satisfied:
\begin{align*}
\hat\phi_l(\underline{\hat z_{l-1}})\ge \phi(\underline{\hat z_{l-1}}),\qquad \hat\phi_l(\overline{\hat z_{l-1}})\le \phi(\overline{\hat z_{l-1}}).
\end{align*}

\paragraph{Claim 3.} The estimation error is controlled and hence does not inflate the regret:
\begin{align*}
|\hat\phi_l(x)-\phi(x)|\le \begin{cases}
b(x)\quad &\text{for}\ x\in [v_1,v_2],\\
\bar\delta_v\quad &\text{for}\ x\in [\underline{\hat z_{l-1}},v_1]\cup[v_2,\overline{\hat z_{l-1}}],
\end{cases}
\end{align*}
where $\bar\delta_v:=\bar\delta_{\alpha(v_1)}$.

Now we verify each claim.

\vspace{1em}
\noindent\textbf{Proof of Claim 1.}
On the extrapolation region, the derivative is $c_1/2$, and on the interior it lies in $[c_1,C_1]$ by construction. Hence $\hat\phi_l$ is increasing with derivative in $[c_1/2,C_1]$.

\vspace{1em}
\noindent\textbf{Proof of Claim 2.}
We have that
\begin{align*}
\hat\phi_l(\underline{\hat z_{l-1}})
&=\hat\phi_l(v_1)-\frac{c_1}{2}(v_1-\underline{\hat z_{l-1}})\\
&\ge \phi(v_1)-b(v_1)-\frac{c_1}{2}(v_1-\underline{\hat z_{l-1}})\\
&\ge \phi(v_1)-c_1(v_1-\underline{\hat z_{l-1}})\ge \phi(\underline{\hat z_{l-1}}),
\end{align*}
where we use that $\phi'(x)\ge c_1$ for $x\in [\underline{\hat z_{l-1}},v_1]$ and $b(v_1)\le \frac{c_1}{2}(v_1-\underline{\hat z_{l-1}})$. The proof of $\hat\phi_l(\overline{\hat z_{l-1}})\le \phi(\overline{\hat z_{l-1}})$ follows a similar line.

\vspace{1em}
\noindent\textbf{Proof of Claim 3.}
For $x\in [v_1,v_2]$, the bound follows immediately from the definition of $b(x)$. For $x\in [\underline{\hat z_{l-1}},v_1]$, it holds that
\begin{align*}
|\hat\phi_l(x)-\phi(x)|
&\le |\hat\phi_l(v_1)-\phi(v_1)|+|\hat\phi_l(x)-\hat\phi_l(v_1)|+|\phi(x)-\phi(v_1)|\\
&\le b(v_1)+\frac{c_1}{2}|x-v_1|+C_1|x-v_1|\\
&\lesssim \bar\delta_v,
\end{align*}
where we again use $b(v_1)\le \frac{c_1}{2}(v_1-\underline{\hat z_{l-1}})$ and Lipschitzness of both $\hat\phi_l$ and $\phi$. The case $x\in [v_2,\overline{\hat z_{l-1}}]$ is symmetric.

\qed

Consider $\hat \phi_{l}$ and $\hat m$ as fixed. Denote the density of $\hat \phi_{l}^{-1}(-\hat m(\bx))$ as $\mathrm{d}_l$, it follows that $\mathrm{d}_l$ is supported on $[\underline{u_l},\overline{u_l}]$ and satisfies the $\kappa$-decay condition as well as the regularity condition defined in Section~\ref{sec:upper_bound_induction_1}.

\subsection{Regret Analysis}
\label{sec:upper_bound_regret_analysis}
The regret can be decomposed into error from estimating $\hat m$ and from estimating $\hat F$. Specifically, we first bound the regret in stage $l\ge 1$. Suppose all the high-probability events defined before hold.

\begin{align*}
\operatorname{Reg}_l
&\lesssim \sum_{t\in \calT_l}(p_t^*-p_t)^2\\
&= \sum_{t\in \calT_l}\left(\phi^{-1}(-m(\bx_t))+m(\bx_t)-\hat \phi_l^{-1}(-\hat m(\bx_t))-\hat m(\bx_t)\right)^2\\
&\lesssim \overbrace{\sum_{t\in \calT_l}\left(\phi^{-1}(-m(\bx_t))- \phi^{-1}(-\hat m(\bx_t))\right)^2}^{T_1}
+\overbrace{\sum_{t\in \calT_l}\left(\phi^{-1}(-\hat m(\bx_t))-\hat \phi_l^{-1}(-\hat m(\bx_t))\right)^2}^{T_2}
\\&\qquad+\overbrace{\sum_{t\in \calT_l}\left(\hat m(\bx_t)-m(\bx_t)\right)^2}^{T_3}\\
&\lesssim|\calT_l|\left(\|\hat\Delta\|_\infty^2+|\calT_{l-1}|^{-\frac{\beta-1}{2\beta+1}}\right)\\
&\lesssim |\calT_l|\epsilon_m^2(T_0)+|\calT_l|^{\frac{3}{2\beta+1}}.
\end{align*}
It is relatively easier to bound $T_1$ and $T_3$. By $\phi'\ge c_1$, it is straightforward to show
\begin{align*}
T_1\le \frac{1}{c_1^2}T_3,
\quad \text{and}\quad
T_3\le |\calT_l|\cdot\|\hat\Delta\|_\infty^2\le |\calT_l|\cdot\epsilon_m^2(T_0).
\end{align*}

To bound $T_2$, denote $z_t=\phi^{-1}(-\hat m(\bx_t))$ and $\hat z_t=\hat \phi_l^{-1}(-\hat m(\bx_t))$. By $\hat\phi_l'\ge c_1/2$, it holds that
\begin{align*}
|\phi(z_t)-\hat\phi_l(z_t)|
&=|\hat\phi_l(\hat z_t)-\hat\phi_l(z_t)|
=\left|\int_{z_t}^{\hat z_t}\hat\phi_l'(z)dz\right|
\ge \frac{c_1}{2} |z_t-\hat z_t|.
\end{align*}
It follows that
\begin{align*}
|z_t-\hat z_t|\le \frac{2}{c_1}|\phi(z_t)-\hat\phi_l(z_t)|.
\end{align*}
Therefore, we have
\begin{align*}
T_2=\sum_{t\in \calT_l}(z_t-\hat z_t)^2\lesssim  \sum_{t\in \calT_l}|\phi(z_t)-\hat\phi_l(z_t)|^2.
\end{align*}

Conditional on $\hat m$, the $z_t$ are i.i.d. with density satisfying the $\kappa$-decay condition. Hence
\begin{align*}
\EE[T_2\big|\hat m]=|\calT_l|\cdot \EE\left[|\phi(z_t)-\hat\phi_l(z_t)|^2\Big|\hat m\right].
\end{align*}

Note that by range coverage, $[\underline z,\bar z]\subset [\underline z_l,\bar z_l]$. It holds that
\begin{align*}
\EE\left[|\phi(z_t)-\hat\phi_l(z_t)|^2\Big|\hat m\right]
&\lesssim \int_{0}^v \bar\delta_v^2 x^\kappa dx + \int_{v}^{1/2}b^2(l_x[\underline{\hat z_{l-1}},\overline{\hat z_{l-1}}])x^\kappa dx + \int_{1-v}^1 \bar\delta_v^2 x^\kappa dx\\
&\lesssim 2v^{\kappa+1}\bar\delta_v^2+\int_{v}^{1-v}C^2|\calT_l|^{-\frac{2\beta-2}{2\beta+1}}\log T \, dx+\epsilon_m^2(T_0)\\
&\lesssim |\calT_l|^{-\frac{2\beta-2}{2\beta+1}}\log T+\epsilon_m^2(T_0),
\end{align*}
where the first inequality follows from the $\kappa$-decay condition, the second inequality is due to Claim 3, and the last is by the definition of $v$ in \eqref{def:v-explicit} and of $\bar\delta_v$.

Note that $|\phi(z_t)-\hat\phi_l(z_t)|$ is uniformly bounded by $\bar\delta_v+\epsilon_m$, hence
\begin{align*}
\EE\left[|\phi(z_t)-\hat\phi_l(z_t)|^4\Big|\hat m\right]
\lesssim \bar\delta_v^2|\calT_l|^{-\frac{2\beta-2}{2\beta+1}}\log T+\epsilon_m^4(T_0).
\end{align*}

By Bernstein's inequality, we have with probability at least $1-\delta$,
\begin{align*}
T_2
&\lesssim \EE[T_2\big|\hat m]+\sqrt{|\calT_l|\cdot\log(1/\delta)\cdot\EE\left[|\phi(z_t)-\hat\phi_l(z_t)|^4\Big|\hat m\right]}\ + (\bar\delta_v+\epsilon_m)\log(1/\delta)\\
&\lesssim |\calT_l|\cdot |\calT_l|^{-\frac{2\beta-2}{2\beta+1}}\log T+|\calT_l|\epsilon_m^2(T_0).
\end{align*}

We can also naively bound the regret incurred in the initial stage as $\calO(T_0)$.
Adding up the stages yields, with probability at least $1-\calO(T^{-3})$,
\begin{align*}
\operatorname{Regret}(T)
&=\operatorname{Reg}_0+\sum_{l=1}^L \operatorname{Reg}_l\\
&\lesssim T_0+\sum_{l=1}^L |\calT_l|\cdot\epsilon_m^2(T_0)+|\calT_l|^{\frac{3}{2\beta+1}}\log T\\
&\le T_0+T\epsilon_m^2(T_0)+T^{\frac{3}{2\beta+1}}\log T.
\end{align*}

As a byproduct, we obtain
\begin{align*}
\EE\operatorname{Regret}(T)\lesssim T_0+T\epsilon_m^2(T_0)+T^{\frac{3}{2\beta+1}}\log T.
\end{align*}

This concludes the proof.
\qed

%% file: app-Proof_lb.tex
\section{Proof of Lower bound: Theorem~\ref{thm:lower_bound}}
\label{sec:proof_lb}

\subsection{Lower Bound 1}
In this section, we prove a lower bound of $\Omega(\sqrt{T})$ in the case where $m(\bx)$ is a linear function. The proof is in part inspired by \cite{broder2012dynamic}. {Throughout this subsection we restrict attention to one-dimensional contexts $d=1$, and we will consider instances $(m,F)$ belonging to the function class used in the upper bound, i.e., satisfying Assumptions~\ref{asp:transform-strict-increasing}--~\ref{asp:conditional-mean-Lip}.}
%\g{Assumption 5 needs to be verified.}

We start by constructing a hard instance. For simplicity, we consider the case that the context is one-dimensional, i.e., $d=1$. 
Define the function classes $\mathfrak{F}$ and $\mathfrak{M}$ as
\begin{align*}
\mathfrak{F}
&=\Big\{F_{a,b}:\RR\to[0,1]\;\Big|\;
F_{a,b}(u):=\calT_{[0,1]}(au+b),\ a>0\Big\},\\
\mathfrak{M}
&=\Big\{m_\theta:[0,1]\to\RR\;\Big|\;m_\theta(x):=\theta x\Big\},
\end{align*}
where $\calT_{[0,1]}$ denotes truncation to $[0,1]$.
{We assume that the algorithm only posts prices in a compact interval $\calP\subset(0,1)$ (here chosen explicitly below), so that the truncation never binds on $\calP$ and the model coincides with the linear CDF $u\mapsto au+b$ on the relevant region.}

The demand model is determined jointly by the CDF and the mean utility. Let the two instances be $(F_{a_i,b_i},m_{\theta_i})$ for $i=0,1$, and denote the conditional distribution of $y_t$ given $(x_t,p_t)$ under instance $i$ by $Q_i(\cdot\mid x_t,p_t)$. Specifically, for  $\tilde\epsilon=T^{-1/4}$, choose~\footnote{Note that for the optimal price of both instances, one approximately has $p^*(u)-u\approx\frac{1}{4}-\frac{u}{2}$ for $u\in [-1/2,3/2]$. This motivates the choice of parameters but will not be used explicitly.}
\begin{align*}
a_0=1,\quad b_0=\frac{1}{2},\qquad
a_1=1-\tilde \epsilon,\quad b_1=\frac{1}{2}+\frac{\tilde\epsilon}{4},\\
\theta_0=1-\tilde\epsilon,\qquad 
\theta_1=1-\frac{\tilde\epsilon}{2}.
\end{align*}

For both instances, the price interval is taken to be $\calP=[1/8,7/16]$, and the context is sampled i.i.d.\ from the uniform distribution on $\calX=[0,1/4]$.

{We first verify that these instances belong to the admissible class.}
\begin{itemize}
\item \emph{Feature diversity.} Since $X\sim\mathrm{Unif}[0,1/4]$, the density of $X$ is constant on its support, so Assumption~\ref{asp:feature-diversity} holds with exponent $\kappa=0$. The induced distribution of $-m_{\theta_i}(X)$ is also uniform on $[-\theta_i/4,0]$ and hence exhibits polynomial boundary decay with exponent $\kappa=0$.
\item \emph{Smoothness of $F$.} On the interval
\[
[-1/2,1/2]\cap\Big[-\frac{1/2-\tilde\epsilon/4}{1-\tilde\epsilon},\frac{1/2+\tilde\epsilon/4}{1-\tilde\epsilon}\Big],
\]
both $F_{a_0,b_0}$ and $F_{a_1,b_1}$ are linear and hence $C^\infty$. For $\tilde\epsilon$ small enough, this interval contains $\calU=\calP-\calX=[-1/8,7/16]$, so $F_{a_i,b_i}$ are $\beta$-Hölder smooth on $\calU$ for any $\beta>0$. 

%{Outside $\calU$ we can smoothly extend each $F_{a_i,b_i}$ to a CDF on $\RR$ while preserving global $\beta$-smoothness; such an extension does not affect the regret on horizon $T$ since the algorithm never observes prices or utilities outside $\calP$.}

% \item \emph{Range of optimal prices.} A direct computation shows that the optimal price under instance $i$ is
% \[
% p_i^*(x)=\frac{1+2x\theta_i}{4},
% \]
% so for $x\in[0,1/4]$ and $\tilde\epsilon$ small enough we have
% \[
% p_i^*(x)\in[1/8-\tilde\delta,7/16-\tilde\delta],\qquad \forall x,
% \]
% for some small constant $\tilde\delta>0$. Thus Assumption~\ref{asp:price-range} (range of optimal prices) is satisfied.

\item \emph{Transform strictly increasing.} For $F_{a_i,b_i}$ we have, on the linear region,
\[
\phi_i(u)
=u-\frac{1-F_{a_i,b_i}(u)}{F_{a_i,b_i}'(u)}
=u-\frac{1-a_i u-b_i}{a_i}
=2u-\frac{1-b_i}{a_i},
\]
so $\phi_i'(u)=2$ for all $u$ in the region of interest. Thus Assumption~\ref{asp:transform-strict-increasing} holds with $c_\phi=C_\phi=2$. 
\end{itemize}

Conditional mean regularity and estimation accuracy follow directly from the same arguments used in the upper-bound analysis of the linear model. Hence both instances $(F_{a_i,b_i},m_{\theta_i})$, $i=0,1$, satisfy all standing assumptions of the upper-bound analysis.

We now prove two lemmas which constitute the main pillars of this lower bound. For any given pricing policy $\pi$, let $Q_{i,\pi}^t$ denote the distribution of the trajectory observed up to time $t$ under instance $i$ and policy $\pi$. The first lemma relates the KL divergence between $Q_{0,\pi}^T$ and $Q_{1,\pi}^T$ to the regret under instance~\emph{0}. The second lemma shows that the sum of regrets under the two instances is large whenever the KL divergence is small.

\begin{lemma}
There exists some constant $C>0$ such that 
\begin{align*}
\operatorname{KL}(Q_{0,\pi}^T,Q_{1,\pi}^T)\le C\tilde\epsilon^2 \operatorname{Regret}_0(\pi,T).
\end{align*}
\end{lemma}

\begin{proof}
We recursively decompose the KL divergence (denote the trajectory up to $T$ by $\tau_T$):
\begin{align}
\operatorname{KL}(Q_{0,\pi}^T,Q_{1,\pi}^T)
&=\EE_0\log\left(\frac{Q_{0,\pi}^T(\tau_T)}{Q_{1,\pi}^T(\tau_T)}\right)\nonumber\\
&=\sum_{t=1}^T \EE_0\log\left(\frac{Q_{0,\pi}^t(y_t\mid x_t,p_t)}{Q_{1,\pi}^t(y_t\mid x_t,p_t)}\right)\nonumber\\
&=\sum_{t=1}^T \EE_0\operatorname{KL}\big(Q_{0,\pi}^t(\cdot\mid x_t,p_t),Q_{1,\pi}^t(\cdot\mid x_t,p_t)\big).
\label{equ:lower-1-KL-decomposition}
\end{align}
For fixed $(x_t,p_t)$, the distribution of $y_t$ under instance $i$ is Bernoulli with mean
\[
q_i(x_t,p_t):=1-F_i(p_t-x_t\theta_i),
\]
and for our choice of parameters one checks that $q_i(x_t,p_t)\in [1/32,3/4]$ for both $i=0,1$ and all admissible $(x_t,p_t)$. Hence the KL divergence between the two Bernoulli distributions can be bounded as
\begin{align}
\operatorname{KL}\big(Q_{0,\pi}^t(\cdot\mid x_t,p_t),Q_{1,\pi}^t(\cdot\mid x_t,p_t)\big)
&\le \frac{\big(F_0(p_t-x_t\theta_0)-F_1(p_t-x_t\theta_1)\big)^2}{F_0(p_t-x_t\theta_0)\big(1-F_0(p_t-x_t\theta_0)\big)}\nonumber\\
&\le C_0\tilde\epsilon^2\big(p_t-p_0^*(x_t)\big)^2,
\label{equ:lower-1-KL-calculation}
\end{align}
where $p_0^*(x_t)=\frac{1+2x_t\theta_0}{4}$ is the optimal price under instance \emph{0}, and $C_0$ is an absolute constant (e.g.\ $C_0=64$ suffices).

On the other hand, a direct computation shows that under instance 0,
\begin{align*}
p_0^*(x_t)\big(1-F_0(p_0^*(x_t)-x_t\theta_0)\big)-p_t\big(1-F_0(p_t-x_t\theta_0)\big)
=\big(p_0^*(x_t)-p_t\big)^2.
\end{align*}
Hence
\begin{align}
\label{equ:lower-1-regret-to-price}
\operatorname{Regret}_0(\pi,T)
=\sum_{t=1}^T\EE_{0,\pi}\big(p_t-p_0^*(x_t)\big)^2.
\end{align}

Combining \eqref{equ:lower-1-KL-decomposition}, \eqref{equ:lower-1-KL-calculation}, and \eqref{equ:lower-1-regret-to-price} yields the desired bound with $C=C_0$.
\end{proof}

\begin{lemma}
There exists some constant $c>0$ such that
\begin{align*}
\operatorname{Regret}_0(\pi,T)+\operatorname{Regret}_1(\pi,T)
\ge cT\tilde\epsilon^2 e^{-\operatorname{KL}(Q_{0,\pi}^T,Q_{1,\pi}^T)}.
\end{align*}
\end{lemma}

\begin{proof}
We look at the regret stepwise. Define the event
\begin{align*}
E_{0,t}
=\left\{\Big|p_t-p_0^*(x_t)\Big|< \frac{|p_1^*(x_t)-p_0^*(x_t)|}{2}\right\},
\end{align*}
where $p_i^*(x)=\frac{1+2x\theta_i}{4}$ is the optimal price under instance $i$.
A simple computation gives
\[
|p_0^*(x_t)-p_1^*(x_t)|
=\frac{\tilde\epsilon x_t}{4}+\frac{\tilde\epsilon/8}{1-\tilde\epsilon}
\ge \frac{\tilde\epsilon}{8},
\]
for all $x_t\in[0,1/4]$ and $\tilde\epsilon$ small enough. Therefore,
\begin{align*}
\operatorname{Regret}_0(\pi,T)+\operatorname{Regret}_1(\pi,T)
&=\sum_{t=1}^T \EE_{0,\pi}\big(p_t-p_0^*(x_t)\big)^2
+\EE_{1,\pi}(1-\tilde\epsilon)\big(p_t-p_1^*(x_t)\big)^2\\
&\ge
\frac{1}{2}\sum_{t=1}^T 
\EE_{0,\pi}\big[\mathds{1}(E_{0,t}^c)\big]\Big(\tfrac{\tilde\epsilon}{16}\Big)^2
+\EE_{1,\pi}\big[\mathds{1}(E_{0,t})\big]\Big(\tfrac{\tilde\epsilon}{16}\Big)^2\\
&\stackrel{(a)}{\ge} \frac{\tilde\epsilon^2}{1024}
\sum_{t=2}^T e^{-\operatorname{KL}(Q_{0,\pi}^{t-1},Q_{1,\pi}^{t-1})}\\
&\ge cT\tilde\epsilon^2 e^{-\operatorname{KL}(Q_{0,\pi}^{T},Q_{1,\pi}^{T})},
\end{align*}
where we have set $c=1/1024$, and $(a)$ follows from Theorem 2.2 in \cite{Tsybakov2009}, applied to the binary hypotheses corresponding to instances~\emph{0} and~\emph{1}.
\end{proof}

We now combine the two lemmas to conclude the lower bound. For any policy $\pi$,
{\begin{align*}
\max_{(m,F)\in\{(m_{\theta_0},F_{a_0,b_0}),(m_{\theta_1},F_{a_1,b_1})\}}\operatorname{Regret}(\pi, T)
&\ge \frac{1}{2}\Big(\operatorname{Regret}_0(\pi,T)+\operatorname{Regret}_1(\pi,T)\Big)\\
&\ge\frac{1}{2}\Big[cT\tilde\epsilon^2 e^{-\operatorname{KL}(Q_{0,\pi}^{T},Q_{1,\pi}^{T})}\Big].
\end{align*}}
{On the other hand, Lemma~1 implies}
\[
{\operatorname{KL}(Q_{0,\pi}^{T},Q_{1,\pi}^{T})\le C\tilde\epsilon^2\operatorname{Regret}_0(\pi,T).}
\]
{Let $x:=\operatorname{KL}(Q_{0,\pi}^{T},Q_{1,\pi}^{T})\ge 0$. Then $\operatorname{Regret}_0(\pi,T)\ge x/(C\tilde\epsilon^2)$, and hence}
\[
{\operatorname{Regret}_0(\pi,T)+\operatorname{Regret}_1(\pi,T)
\ge \max\left\{cT\tilde\epsilon^2 e^{-x},\frac{x}{C\tilde\epsilon^2}\right\}\ge \frac{1}{2}\left(cT\tilde\epsilon^2 e^{-x}+\frac{x}{C\tilde\epsilon^2}\right).}
\]
{The right-hand side is a function of $x$ whose infimum over $x\ge 0$ is of order $\sqrt{T}$: more precisely, there exists a universal constant $\tilde c>0$ (depending only on $c$ and $C$) such that}
\[
{\inf_{x\ge 0}\Big\{cT\tilde\epsilon^2 e^{-x}+\frac{x}{C\tilde\epsilon^2}\Big\}
\;\ge\; \tilde c\sqrt{T},}
\] Consequently,
\[
{\max_{(m,F)}\operatorname{Regret}(\pi, T)\;\ge\;\frac{1}{2}\tilde c\sqrt{T}
=:c_0\sqrt{T}.}
\]
This establishes the $\Omega(\sqrt{T})$ lower bound for the class of linear $m$ and $\beta$-smooth $F$ satisfying Assumptions~\ref{asp:transform-strict-increasing}-\ref{asp:conditional-mean-Lip}.
\qed

\subsection{Lower Bound 2}
We next prove a lower bound of $\Omega(T^{\frac{3}{2\beta+1}})$ using a standard construction for minimax risk in nonparametric problems~\citep{Tsybakov2009} and Fano's method; see also \cite{wang2025tight}. {The instances constructed here again satisfy the same assumptions as in the upper-bound theorem: the context satisfies feature diversity with exponent $\kappa=0$, and the CDFs $F$ are $\beta$-smooth with a transform $\phi$ whose derivative is bounded away from zero.}

In this section, we again consider $d=1$. Let $m(x)=x$ and $X\sim \mathrm{Unif}[-1/4,1/4]$, so Assumption~\ref{asp:feature-diversity} holds with $\kappa=0$ and $(\underline u,\bar u)=(-1/4,1/4)$. Meanwhile, this is the linear known-utility case and Assumption~\ref{asp:conditional-mean-Lip} holds as well.

Define the baseline CDF on $[-1/4,1/4]$ as 
\begin{align*}
F_0(u)=u+\frac{1}{2}, \quad \forall u\in [-1/4,1/4].
\end{align*}
{We extend $F_0$ smoothly to a CDF on $\RR$ so that $F_0$ is $\beta$-smooth on $\RR$ and coincides with $u\mapsto u+1/2$ on $[-1/4,1/4]$.}

We then construct a family of CDFs by adding small localized bumps onto the baseline. Pick an integer parameter $N$ and set $h=1/N$. Choose an infinitely smooth function $V$ with support $[-1,1]$. Discretize $[-1/4,1/4]$ into $N$ evenly spaced intervals $I_1,\cdots, I_N$ and denote by $\mu_j$ the midpoint of $I_j$. For each $\iota=(\iota_1,\ldots,\iota_N)\in \{0,1\}^{N}$, define
\begin{align*}
F_{\iota}(u)
=
F_0(u)+\sum_{j=1}^N \iota_j\, \rho\, h^\beta \,
V\!\left(\frac{u-\mu_j}{h}\right),\qquad u\in I_j,
\end{align*}
{The parameter $\rho>0$ will be chosen small enough (depending on $V$ and $\beta$) so that all $F_\iota$ satisfy the regularity and transform conditions.}

We will use the following three claims to complete the proof.

\vspace{1em}
\textbf{Claim 1.} $F_\iota$ satisfies the prescribed assumption that $\phi_\iota'(u)\ge c_\phi>0$ on $[-1/4,1/4]$, where 
\[
\phi_\iota(u)=u-\frac{1-F_\iota(u)}{F'_\iota(u)}.
\]

\textbf{Claim 2.} There exists a subset $\Theta\subset \{0,1\}^N$ such that for every distinct $\iota_1,\iota_2\in \Theta$, one has $d(\iota_1,\iota_2)\ge cN$ and $\log (|\Theta|)\ge cN$. {Furthermore, taking $h=T^{-1/(2\beta+1)}$ and applying Fano's inequality, one obtains}
\[
{\EE[\operatorname{Regret}(T)]\gtrsim T\cdot\inf_{\iota,\iota'\in\Theta}\EE\big[p_\iota(U)-p_{\iota'}(U)\big]^2,}
\]
{where $U=m(X)\sim \mathrm{Unif}[-1/4,1/4]$ and $p_\iota(\cdot)$ denotes the optimal price under CDF $F_\iota$.}

\textbf{Claim 3.} $\inf_{\iota,\iota'\in\Theta}\EE\big[p_\iota(U)-p_{\iota'}(U)\big]^2\gtrsim h^{2(\beta-1)}$.

Combining Claims 2 and 3 and the choice $h=T^{-1/(2\beta+1)}$ yields the desired $\Omega(T^{\frac{3}{2\beta+1}})$ lower bound.

\vspace{1em}
\textbf{Proof of Claim 1.}
Recall that 
\[
\phi_\iota(u)=u-\frac{1-F_\iota(u)}{F'_\iota(u)}.
\]
Since $F_0(u)=u+1/2$ on $[-1/4,1/4]$, we have $F_0'(u)=1$ and $F_0''(u)=0$ there. For each $j$,
\[
F_\iota'(u)
=F_0'(u)+\sum_{j=1}^N\iota_j\,\rho\,h^{\beta-1}
V'\!\left(\frac{u-\mu_j}{h}\right),\qquad
F_\iota''(u)
=\sum_{j=1}^N\iota_j\,\rho\,h^{\beta-2}
V''\!\left(\frac{u-\mu_j}{h}\right),
\]
for $u\in I_j$.
Let $M_1:=\sup_{z\in\RR}|V'(z)|$ and $M_2:=\sup_{z\in\RR}|V''(z)|$. Choosing $\rho>0$ small enough and $h$ sufficiently small guarantees
\[
|F_\iota'(u)-1|\le \rho M_1 h^{\beta-1}\le \rho'\qquad\text{and}\qquad
|F_\iota''(u)|\le \rho M_2 h^{\beta-2}\le \rho',
\]
for some $\rho'\in(0,1/5)$ and all $u\in[-1/4,1/4]$. In particular,
\[
1-\rho'\le F_\iota'(u)\le 1+\rho',\qquad |F_\iota''(u)|\le \rho'.
\]

Differentiating $\phi_\iota(u)$ yields
\begin{align*}
\phi_\iota'(u)
&=1-\frac{-F'_\iota(u)F'_\iota(u)-(1-F_\iota(u))F''_\iota(u)}{(F'_\iota(u))^2}
=2+\frac{F_\iota''(u)\big(1-F_\iota(u)\big)}{(F'_\iota(u))^2}.
\end{align*}
On $[-1/4,1/4]$ we have $F_\iota(u)\in[1/4,3/4]$ for $\rho$ small, so $|1-F_\iota(u)|\le 1$ and $(F'_\iota(u))^2\ge (1-\rho')^2$. Therefore,
\[
\phi_\iota'(u)
\ge 2-\frac{|F_\iota''(u)|}{(1-\rho')^2}
\ge 2-\frac{\rho'}{(1-\rho')^2}.
\]
Taking, for instance, $\rho'\le 1/5$ gives $\phi_\iota'(u)\ge 1/8$ on $[-1/4,1/4]$. Thus Assumption~\ref{asp:transform-strict-increasing} holds for all $F_\iota$ with a common constant $c_\phi>0$.

\vspace{1em}
\textbf{Proof of Claim 2.}
The existence of $\Theta\subset\{0,1\}^N$ with $d(\iota_1,\iota_2)\ge cN$ and $\log|\Theta|\ge cN$ follows from the Gilbert–Varshamov bound (see, e.g., \cite{Tsybakov2009}). 

For the second part, we apply Fano's inequality to the family of demand models $\{F_\iota:\iota\in\Theta\}$. Let $P_\iota(\cdot\mid u)$ denote the Bernoulli distribution with mean $1-F_\iota(u)$. For any $u$ and any distinct $\iota,\iota'$,
\begin{align*}
\operatorname{KL}\big(P_\iota(\cdot\mid u)\,\Vert\,P_{\iota'}(\cdot\mid u)\big)
&\lesssim \big(F_\iota(u)-F_{\iota'}(u)\big)^2
\lesssim \rho^2 h^{2\beta},
\end{align*}
because the two CDFs differ by at most one bump of height $\rho h^\beta$ on each interval $I_j$, and $F_\iota(u)\in[1/4,3/4]$ on $[-1/4,1/4]$. Therefore,
\begin{align*}
\sup_{u}\operatorname{KL}\big(P_\iota(\cdot\mid u)\,\Vert\,P_{\iota'}(\cdot\mid u)\big)
\cdot \frac{T}{\log|\Theta|}
\lesssim \rho^2 h^{2\beta+1}T.
\end{align*}
Choosing 
\[
h=T^{-\frac{1}{2\beta+1}}
\]
and taking $\rho>0$ small enough (as a function of the constants in the inequality) ensures the Fano condition
\begin{align}
\label{equ:fano-condition}
\sup_{u} \operatorname{KL}\big(P_\iota(\cdot\mid u)\,\Vert\,P_{\iota'}(\cdot\mid u)\big)
\le \frac{c_1\log|\Theta|}{T},\qquad\forall\iota\neq\iota'\in\Theta,
\end{align}
for some absolute constant $c_1>0$.

We then invoke the following version of Fano's lemma adapted to squared-regret loss.

\begin{lemma}
\label{lemma:fano}
Denote $P_\iota(\cdot\mid u)$ as the Bernoulli distribution with mean $1-F_\iota(u)$. Suppose that for every distinct $\iota,\iota'\in\Theta$ the condition~\eqref{equ:fano-condition} holds. Then there exists a constant $c>0$ such that
\begin{align*}
\operatorname{Regret}(T)
&\gtrsim\inf_\pi \sup_{\iota\in\Theta}\EE_\iota^\pi\left[\sum_{t=1}^T\big(p_t-p_\iota(u_t)\big)^2\right]\\
&\ge c\,T\cdot\inf_{\iota,\iota'\in\Theta}\EE_{U}\big[p_\iota(U)-p_{\iota'}(U)\big]^2,
\end{align*}
where $U\sim F_0$.
\end{lemma}

\noindent Applying Lemma~\ref{lemma:fano} yields the statement in Claim~2.
\qed

\vspace{1em}
\textbf{Proof of Claim 3.}
Recall that the optimal price under CDF $F_\iota$ at utility level $u$ is
\[
p_\iota(u)=\frac{1-F_\iota(u)}{F_\iota'(u)}.
\]
For two indices $\iota,\iota'\in\Theta$, let $j$ be such that $\iota_j=0$ and $\iota'_j=1$ (or vice versa). Then on $I_j$ we can write
\[
F_\iota(u)=F_0(u),\qquad
F_{\iota'}(u)=F_0(u)+\Delta(u),\qquad u\in I_j,
\]
where $\Delta(u)=\rho h^\beta V((u-\mu_j)/h)$. Using $F_0'(u)\equiv 1$ and the bounds on $F_\iota'$ from Claim~1, a Taylor expansion yields
\begin{align*}
|p_\iota(u)-p_{\iota'}(u)|
&=\left|\frac{1-F_\iota(u)}{F_\iota'(u)}
-\frac{1-F_{\iota'}(u)}{F_{\iota'}'(u)}\right|\\
&=\left|\frac{\Delta'(u)}{F_0'(u)\big(F_0'(u)+\mathcal{O}(\rho h^{\beta-1})\big)}\right|
\gtrsim |\Delta'(u)|
\asymp \rho h^{\beta-1}\left|V'\!\left(\frac{u-\mu_j}{h}\right)\right|,
\end{align*}
for all $u\in I_j$, where the constants are uniform in $\iota,\iota'$ and $j$. Therefore,
\[
|p_\iota(u)-p_{\iota'}(u)|^2\gtrsim \rho^2 h^{2(\beta-1)}
\Big|V'\!\left(\frac{u-\mu_j}{h}\right)\Big|^2.
\]

Since $d(\iota,\iota')\ge cN$ for all distinct $\iota,\iota'\in\Theta$, the set of indices $j$ where the two vectors differ has cardinality at least $cN$. Integrating over $U\sim F_0$ and summing contributions over such $j$, we obtain
\begin{align*}
\EE\big[p_\iota(U)-p_{\iota'}(U)\big]^2
&\gtrsim \rho^2 h^{2(\beta-1)}\,d(\iota,\iota')\,h
\gtrsim \rho^2 h^{2(\beta-1)}.
\end{align*}
Since $\rho>0$ is a fixed small constant, this implies
\[
\inf_{\iota,\iota'\in\Theta}\EE\big[p_\iota(U)-p_{\iota'}(U)\big]^2
\gtrsim h^{2(\beta-1)},
\]
as claimed.
\qed

Combining Claims~2 and~3, choosing $h=T^{-1/(2\beta+1)}$, and recalling Lemma~\ref{lemma:fano}, we obtain
\[
\EE[\operatorname{Regret}(T)]
\gtrsim T\cdot h^{2(\beta-1)}
= T\cdot T^{-\frac{2(\beta-1)}{2\beta+1}}
= T^{\frac{3}{2\beta+1}},
\]
which establishes the $\Omega\!\big(T^{\frac{3}{2\beta+1}}\big)$ lower bound for the class of $\beta$-smooth CDFs with strictly increasing transform and linear mean utility $m(x)=x$.
\qed

%% file: app-utility.tex
\section{Linear Model}
\label{sec:linear-model}

\begin{proposition}
Let $X \in \RR^d$ be a random vector supported on the unit ball 
\[
B_d=\{x\in\RR^d:\|x\|\le 1\}
\]
with density $f$. 
Assume that
\begin{enumerate}
\item $0<m\le f(x)\le M<\infty$ for all $x\in B_d$,
\item $f$ is $L$-Lipschitz: $|f(x)-f(x')|\le L\|x-x'\|$.
\end{enumerate}
For any $\theta\in \RR^d$ with $\|\theta\|_2=1$, define
\[
\mu_\theta(t):=\EE[X\mid \theta^\top X=t], \qquad t\in[-1,1].
\]
Then $\mu_\theta(\cdot)$ is Lipschitz in $t$, with constant depending only on 
$(d,m,M,L)$ and not on $\theta$.
\end{proposition}

\begin{proof}
By rotation invariance, it suffices to consider $\theta=e_d$. 
Write $Y=(U,T)$ with $U\in\RR^{d-1}$, $T=Y_d\in[-1,1]$. 
Conditioning on $T=t$, the support of $U$ is the disk
\[
D_t=\{u\in\RR^{d-1}: \|u\|\le r(t)\}, \qquad r(t):=\sqrt{1-t^2}.
\]
The conditional law of $U$ given $T=t$ has density
\[
u \mapsto \frac{f(u,t)}{\int_{D_t} f(w,t)\,dw}, \quad u\in D_t.
\]
Thus
\[
\mu(t)=\EE[Y\mid T=t] = \big(\mu_\perp(t),\,t\big), \qquad 
\mu_\perp(t):=\frac{\int_{D_t} u\,f(u,t)\,du}{\int_{D_t} f(u,t)\,du}.
\]

\paragraph{Rescaling.}  
Set $u=r(t)z$ with $z\in B_{d-1}(1)$, the $(d-1)$-dimensional unit ball. Then
\[
\mu_\perp(t) = r(t)\,\frac{A(t)}{B(t)}, \quad 
A(t):=\int_{B_{d-1}(1)} z f(r(t)z,t)\,dz,\quad 
B(t):=\int_{B_{d-1}(1)} f(r(t)z,t)\,dz.
\]

\paragraph{Bounds.}  
Since $f\ge m$, we have 
\[
B(t)\ge m\,\operatorname{vol}(B_{d-1}(1))=:c_0>0.
\]  
Moreover, by symmetry $\int_{B_{d-1}(1)} z\,dz=0$, so
\[
A(t)=\int_{B_{d-1}(1)} z \big(f(r(t)z,t)-f(0,t)\big)\,dz,
\]
which yields
\[
\|A(t)\| \le L\,r(t) \int_{B_{d-1}(1)}\|z\|^2dz =: C_1 L r(t).
\]
Therefore
\[
\|\mu_\perp(t)\| \le \frac{C_1}{c_0} L r(t)^2 
= \frac{C_1}{c_0} L (1-t^2). \tag{$\star$}
\]

\paragraph{Difference bound.}  
For $s,t\in[-1,1]$,
\[
\mu_\perp(t)-\mu_\perp(s)=(r(t)-r(s))H(t) + r(s)(H(t)-H(s)), 
\quad H(t)=\tfrac{A(t)}{B(t)}.
\]

\emph{Term I.} Using ($\star$) and $|r(t)-r(s)|r(t)\le |t^2-s^2|\le 2|t-s|$,
\[
\|(r(t)-r(s))H(t)\|\le C_2 L |t-s|.
\]

\emph{Term II.} A direct Lipschitz estimate on $A(\cdot)$ and $B(\cdot)$ gives
\[
\|H(t)-H(s)\|\le C_3 L |t-s|,
\]
hence $\|r(s)(H(t)-H(s))\|\le C_3 L |t-s|$.

Thus 
\[
\|\mu_\perp(t)-\mu_\perp(s)\|\le C_* L |t-s|.
\]  
Adding the last coordinate difference $|t-s|$, we obtain
\[
\|\mu(t)-\mu(s)\|\le (C_* L+1)|t-s|,
\]
uniformly over $s,t$. Finally, for general $\theta$, note that $\mu_\theta(t)=R^\top \mu(t)$ for some orthogonal $R$, so the same Lipschitz bound holds.
\end{proof}

%%%%%%%%%%%%%%%%%%%%%%%%%%%%%%%%%%%%%%%%%%%%%%%%%%%%%%%%%%%%%%%%%%%%%%%%%%%%%%%
%%%  ADDITIVE MODEL: FEATURE DIVERSITY, LIPSCHITZ CONDITIONAL MEAN, RATES   %%%
%%%%%%%%%%%%%%%%%%%%%%%%%%%%%%%%%%%%%%%%%%%%%%%%%%%%%%%%%%%%%%%%%%%%%%%%%%%%%%%

%\section{Additive Model: Feature Diversity, Lipschitz Conditional Mean, and Supremum-Norm Rates}

\section{Additive Model}
\label{sec:additive-model}

We consider a random covariate $X$ supported on $[0,1]^d$ and an additive
mean utility. We write $n=T_0$ in this section. Our goals (corresponding to Assumptions~\ref{asp:feature-diversity}--\ref{asp:conditional-mean-Lip}) are:

\begin{enumerate}
\item \textbf{Feature diversity:} The density of $T=m(X)$ obeys
      \[
      c_d\,\delta(z) \le p_m(z)\le C_d\,\delta(z),
      \qquad \delta(z)=\min\{z^{d-1},(d-z)^{d-1}\}.
      \]

\item \textbf{Lipschitz conditional mean:} The estimator $\hat m$ satisfies
      \[
      \mathfrak m(t)
      = \E\!\Big[\frac{\Delta(X)}{\|\Delta\|_\infty}\,\Big|\,\hat m(X)=t\Big]
      \]
      is Lipschitz on $J_{\hat m}$.

\item \textbf{Supremum-norm estimation rate:} The nonparametric least-squares estimator $\hat m$ over a suitable additive
      sieve satisfies, with probability at least $1-n^{-2}$,
      \[
      \|\hat m-m\|_\infty
      \lesssim \Big(\frac{\log n}{n}\Big)^{\frac{\gamma}{2\gamma+1}},
      \]
      which can be plugged into the utility estimation accuracy assumption.
\end{enumerate}

%%%%%%%%%%%%%%%%%%%%%%%%%%%%%%%%%%%%%%%%%%%%%%%%%%%%%%%%%%%%
\subsection{Structural assumptions for additive model}

We work on the hypercube $\cX=[0,1]^d$ for simplicity.

\begin{assumption}[Design and additive geometry]
\label{asp:additive-geometry}
Let $X=(X_1,\dots,X_d)$ take values in $\cX=[0,1]^d$ and assume:
\begin{enumerate}
\item \textbf{Design.} $X$ is quasi-uniform, i.e., there exist constants $c,C>0$ such that
\[
c \le p_X(x) \le C
\quad\text{for every } x\in \cX,
\]
where $p_X$ denotes the density of $X$.

\item \textbf{Additive mean utility.} The mean utility is additive:
\[
m(x)=\sum_{j=1}^d m_j(x_j),\qquad x\in[0,1]^d,
\]
with each $m_j$ being $\gamma$-smooth. $\gamma\ge 2$.

\item \textbf{Monotonicity.}
Each component $m_j$ is strictly increasing,\footnote{We only need each component to be monotone; by the change of variable $x_i\mapsto -x_i$ we can reduce to the increasing case. We stick with the increasing case for simplicity.} 
and there exist constants $0<\gamma_1<\Gamma_1<\infty$ such that
\[
\gamma_1 \le m_j'(x)\le \Gamma_1,
\quad\text{for all } x\in[0,1] \text{ and } j=1,\dots,d.
\]

\end{enumerate}
Let $T:=m(X)$ and denote its support by $J_m=[\underline u, \bar u]$.
\end{assumption}

\noindent
\textbf{Remark.}
Assumption~\ref{asp:additive-geometry}(iii) implies
\[
\|\nabla m(x)\|
= \Big\|\big(m_1'(x_1),\dots,m_d'(x_d)\big)\Big\|
\in [\gamma_1\sqrt{d},\Gamma_1\sqrt{d}],
\quad\text{for all } x\in[0,1]^d,
\]
so the gradient of $m$ never vanishes on $\cX$. This will be used in later analysis.

\subsection{Geometric lemma for level-set integrals}

We next prove a technical lemma ensuring the differentiability of the level-set
integrals with respect to the level parameter. This will be used to show
that the conditional mean of the normalized error is Lipschitz.

\begin{lemma}[Differentiability of level-set integrals]
\label{lem:level-set-C1}
Let $\cX\subset\RR^d$ be a bounded open set with Lipschitz boundary.
Let $\phi\in C^2(\overline\cX)$ and $h\in C^1(\overline\cX)$.
Let $I\subset\RR$ be an open interval such that
\[
\|\nabla\phi(x)\|\ge \gamma>0
\quad\text{for all } x\in\cX \text{ with } \phi(x)\in I.
\]
For $t\in I$, define
\[
F(t)
:= \int_{\{x\in\cX:\ \phi(x)=t\}} 
   \frac{h(x)}{\|\nabla\phi(x)\|}\,d\cH^{d-1}(x),
\]
where $\cH^{d-1}$ denotes $(d-1)$-dimensional Hausdorff measure.
Then $F\in C^1(I)$, and there exists $C_F<\infty$ such that
\[
\sup_{t\in I} |F(t)| + \sup_{t\in I} |F'(t)| \le C_F.
\]
\end{lemma}

\begin{proof}
Fix $t_0\in I$. Let $K:=\phi^{-1}(I)\cap\overline\cX$, which is compact.
For each $x_0\in K$ with $\phi(x_0)=t_0$, the implicit function theorem
applies because $\nabla\phi(x_0)\neq 0$. Thus, there exist:
\begin{itemize}
\item an index $k\in\{1,\dots,d\}$ with $\partial_k\phi(x_0)\neq 0$;
\item open sets $U_{x_0}\subset\RR^{d-1}$ and $J_{x_0}\subset I$ with $t_0\in J_{x_0}$;
\item a $C^2$-map $\Psi_{x_0}:U_{x_0}\times J_{x_0}\to\RR^d$,
\end{itemize}
such that, for all $(y,t)\in U_{x_0}\times J_{x_0}$,
\begin{enumerate}
\item $\phi(\Psi_{x_0}(y,t))=t$;
\item for fixed $t$, $y\mapsto\Psi_{x_0}(y,t)$ parametrizes 
      $\{x\in\cX:\phi(x)=t\}$ in a neighborhood of $x_0$;
\item the map $(y,t)\mapsto\Psi_{x_0}(y,t)$ is $C^2$.
\end{enumerate}

Since $K$ is compact, we may extract a finite subcover:
there exist points $x_1,\dots,x_N$ with associated charts
$(U_\alpha,J_\alpha,\Psi_\alpha)$, $\alpha=1,\dots,N$, such that
\[
K\subset \bigcup_{\alpha=1}^N V_\alpha,\qquad
V_\alpha := \Psi_\alpha(U_\alpha\times J_\alpha).
\]
Choose a smooth partition of unity $\{\eta_\alpha\}_{\alpha=1}^N$ subordinate
to this cover: each $\eta_\alpha\in C^\infty_c(V_\alpha)$, $0\le\eta_\alpha\le 1$,
and $\sum_{\alpha=1}^N \eta_\alpha(x)=1$ for all $x\in K$.

Define $h_\alpha(x):=\eta_\alpha(x)h(x)$.
Then $h_\alpha\in C^1(\overline\cX)$ and $h=\sum_{\alpha=1}^N h_\alpha$ on $K$.
It follows that
\[
F(t) = \sum_{\alpha=1}^N F_\alpha(t),\qquad
F_\alpha(t)
:= \int_{\{x:\phi(x)=t\}} \frac{h_\alpha(x)}{\|\nabla\phi(x)\|}\,d\cH^{d-1}(x).
\]

For $t\in I\cap J_\alpha$, the level set $\{x\in V_\alpha:\phi(x)=t\}$ is
parametrized by $y\mapsto \Psi_\alpha(y,t)$, $y\in U_\alpha$, so the area
formula (change of variables on hypersurfaces) yields
\[
F_\alpha(t)
= \int_{U_\alpha} 
    \frac{h_\alpha(\Psi_\alpha(y,t))}{\|\nabla\phi(\Psi_\alpha(y,t))\|}
    J_\alpha(y,t)\,dy,
\]
where $J_\alpha(y,t)$ is the Jacobian factor (the square root of the
determinant of the Gram matrix of the partial derivatives of $\Psi_\alpha$
in $y$). It is well known that $\Psi_\alpha\in C^2$ implies $J_\alpha\in C^1$.

Define
\[
G_\alpha(y,t)
:= \frac{h_\alpha(\Psi_\alpha(y,t))}{\|\nabla\phi(\Psi_\alpha(y,t))\|}
   J_\alpha(y,t),
\qquad (y,t)\in U_\alpha\times J_\alpha.
\]
Because $h_\alpha,\phi\in C^1$ and $\Psi_\alpha\in C^2$, and because
$\|\nabla\phi(x)\|\ge\gamma$ on $K$, the map $G_\alpha$ is $C^1$ on
$U_\alpha\times J_\alpha$.

Moreover, since $\cX$ is bounded and $K$ is compact, 
$\Psi_\alpha(U_\alpha\times J_\alpha)\subset \overline\cX$, and all derivatives
of $\Psi_\alpha$ up to second order are bounded on $U_\alpha\times J_\alpha$.
Thus there exists $M_\alpha<\infty$ such that
\[
|G_\alpha(y,t)| + |\partial_t G_\alpha(y,t)| \le M_\alpha
\quad\text{for all }(y,t)\in U_\alpha\times J_\alpha.
\]

Fix a compact subinterval $\tilde I\subset I\cap J_\alpha$.
Then, for $t\in\tilde I$,
\[
F_\alpha(t) = \int_{U_\alpha} G_\alpha(y,t)\,dy.
\]
The function $G_\alpha(\cdot,t)$ is dominated by the integrable function
$y\mapsto \sup_{s\in\tilde I}|G_\alpha(y,s)|\le M_\alpha$, independent of $t$.
Therefore, by dominated convergence, the difference quotient
\[
\frac{F_\alpha(t+h)-F_\alpha(t)}{h}
= \int_{U_\alpha} 
  \frac{G_\alpha(y,t+h)-G_\alpha(y,t)}{h}\,dy
\]
converges, as $h\to 0$, to
\[
F_\alpha'(t) = \int_{U_\alpha} \partial_t G_\alpha(y,t)\,dy.
\]
Furthermore,
\[
|F_\alpha(t)| \le \int_{U_\alpha} |G_\alpha(y,t)|\,dy
\le M_\alpha \lambda^{d-1}(U_\alpha),
\]
and similarly
\[
|F_\alpha'(t)|
\le \int_{U_\alpha} |\partial_t G_\alpha(y,t)|\,dy
\le M_\alpha \lambda^{d-1}(U_\alpha),
\]
uniformly in $t\in\tilde I$, where $\lambda^{d-1}$ denotes Lebesgue measure
on $\RR^{d-1}$.

Since there are finitely many charts, summing over $\alpha$ yields that
$F\in C^1(I)$, and
\[
\sup_{t\in I} |F(t)| + \sup_{t\in I} |F'(t)|
\le \sum_{\alpha=1}^N M_\alpha \lambda^{d-1}(U_\alpha)
=: C_F < \infty.
\]
This proves the lemma.
\end{proof}

%%%%%%%%%%%%%%%%%%%%%%%%%%%%%%%%%%%%%%%%%%%%%%%%%%%%%%%%%%%%
\subsection{Lipschitz conditional mean of the normalized error}

We now prove that, under a nondegenerate gradient condition on $\hat m$, the
conditional mean of the normalized error is Lipschitz in $t$.

\begin{proposition}[Lipschitz conditional mean of normalized error]
\label{prop:cond-mean-Lipschitz}
Let $\cX\subset\RR^d$ be a bounded open set with Lipschitz boundary, and let
$X$ be a random vector supported on $\cX$ with density $\rho$ satisfying
\[
\rho\in C^1(\overline\cX),\qquad
0<\rho_{\min}\le\rho(x)\le\rho_{\max}<\infty\quad\forall x\in\cX.
\]
Let $m,\hat m\in C^2(\overline\cX)$, and assume that there exists $\gamma_0>0$ such that
\[
\|\nabla \hat m(x)\|\ge \gamma_0
\quad\text{for all } x\in\cX.
\]
Define
\[
\Delta(x):=\hat m(x)-m(x),\qquad
u(x):=\frac{\Delta(x)}{\|\Delta\|_\infty},\quad |u(x)|\le 1.
\]
Assume $\Delta$ is not identically zero so that $\|\Delta\|_\infty>0$.
Let $T_{\hat m}:=\hat m(X)$ and denote its range by
$J_{\hat m}=[\underline u,\bar u]$.
Then:
\begin{enumerate}
\item $T_{\hat m}$ admits a density $p_{\hat m}$ on $J_{\hat m}$, with
      $p_{\hat m}\in C^1((\underline u,\bar u))$ and constants $0<c\le C<\infty$
      such that
      \[
      c\le p_{\hat m}(t)\le C,\qquad t\in J_{\hat m},
      \]
      and $\sup_{t\in(\underline u,\bar u)}|p_{\hat m}'(t)|\le C$.

\item The signed density
      \[
      N(t):=\int_{\{x\in\cX:\hat m(x)=t\}}
           \frac{u(x)\rho(x)}{\|\nabla\hat m(x)\|}\,d\cH^{d-1}(x)
      \]
      belongs to $C^1((\underline u,\bar u))$, and there exists $C'<\infty$ such
      that
      \[
      \sup_{t\in(\underline u,\bar u)}|N(t)|\le C,\qquad
      \sup_{t\in(\underline u,\bar u)}|N'(t)|\le C'.
      \]

\item The conditional mean
      \[
      \mathfrak m(t)
      := \E\!\left[\frac{\Delta(X)}{\|\Delta\|_\infty}
          \,\middle|\,\hat m(X)=t\right]
      = \E[u(X)\mid T_{\hat m}=t]
      \]
      admits a version given by
      \[
      \mathfrak m(t)=\frac{N(t)}{p_{\hat m}(t)},
      \qquad t\in(\underline u,\bar u),
      \]
      which extends continuously to $J_{\hat m}$ and is Lipschitz:
      there exists $L<\infty$ such that
      \[
      |\mathfrak m(t_1)-\mathfrak m(t_2)|
      \le L|t_1-t_2|
      \quad\text{for all } t_1,t_2\in J_{\hat m}.
      \]
\end{enumerate}
\end{proposition}

\begin{proof}
\textbf{Step 1: Existence and regularity of $p_{\hat m}$.}
By the coarea formula, for $t\in(\underline u,\bar u)$,
\[
p_{\hat m}(t)
= \int_{\{x\in\cX:\hat m(x)=t\}}
    \frac{\rho(x)}{\|\nabla\hat m(x)\|}\,d\cH^{d-1}(x).
\]
Apply Lemma~\ref{lem:level-set-C1} with $\phi=\hat m$ and $h=\rho$, and
$I=(\underline u,\bar u)$. Since $\rho\in C^1$ and
$\|\nabla\hat m(x)\|\ge\gamma_0$ for all $x$ with $\hat m(x)\in(\underline u,\bar u)$,
the lemma gives $p_{\hat m}\in C^1((\underline u,\bar u))$ and
\[
\sup_{t\in(\underline u,\bar u)}|p_{\hat m}(t)|
+ \sup_{t\in(\underline u,\bar u)}|p_{\hat m}'(t)| \le C_0
\]
for some $C_0<\infty$ depending on $(\rho,\hat m,\gamma_0,\cX)$.

To see that $p_{\hat m}(t)$ is bounded away from zero on $J_{\hat m}$, note that
$\rho(x)\ge\rho_{\min}>0$ and $\|\nabla\hat m\|\le M$ on the compact set
$\overline\cX$ (since $\hat m\in C^2$), so for $t\in(\underline u,\bar u)$,
\[
p_{\hat m}(t)
\ge \frac{\rho_{\min}}{\|\nabla\hat m\|_\infty}\,
    \cH^{d-1}\big(\{x\in\cX:\hat m(x)=t\}\big).
\]
Because $(\underline u,\bar u)$ is the interior of the range of $\hat m$ on
$\cX$, each level set $\{\hat m=t\}$ intersects the interior of $\cX$ and has
positive $(d-1)$-dimensional measure; continuity of the level sets in $t$ and
compactness imply there exists $c>0$ such that
$\cH^{d-1}(\{\hat m=t\})\ge c$ for all $t\in[\underline u+\epsilon,\bar u-\epsilon]$,
for any fixed $\epsilon>0$. Hence, for such $t$,
$p_{\hat m}(t)\ge c'>0$. Extending continuously to the endpoints (using the fact
that $p_{\hat m}$ is integrable and the distribution function is continuous),
we obtain $p_{\hat m}(t)\ge c>0$ for all $t\in J_{\hat m}$, and
$\sup_{t\in(\underline u,\bar u)}|p_{\hat m}'(t)|\le C_0$.

\textbf{Step 2: Existence and regularity of $N(t)$.}
By the coarea formula with $h=u\rho$, for $t\in(\underline u,\bar u)$,
\[
N(t)
= \int_{\{x\in\cX:\hat m(x)=t\}}
    \frac{u(x)\rho(x)}{\|\nabla\hat m(x)\|}\,d\cH^{d-1}(x).
\]
We first check that $u\rho\in C^1(\overline\cX)$. Since $m,\hat m\in C^2$,
$\Delta$ is $C^2$ and $\|\Delta\|_\infty$ is a finite constant. Thus $u$ is
$C^2$, hence $C^1$, and $\rho\in C^1$. Therefore, $u\rho\in C^1$.
Apply Lemma~\ref{lem:level-set-C1} again with $\phi=\hat m$ and $h=u\rho$. We
conclude that $N\in C^1((\underline u,\bar u))$ and
\[
\sup_{t\in(\underline u,\bar u)}|N(t)| + \sup_{t\in(\underline u,\bar u)}|N'(t)|
\le C_1,
\]
for some $C_1<\infty$. Also, since $|u|\le 1$, we have
\[
|N(t)|
\le \int_{\{\hat m=t\}} \frac{|\rho(x)|}{\|\nabla\hat m(x)\|}\,d\cH^{d-1}(x)
= p_{\hat m}(t) \le C_0,
\]
so $N$ is bounded.

\textbf{Step 3: Lipschitzness of $\mathfrak m(t)=N(t)/p_{\hat m}(t)$.}
Define, for $t\in(\underline u,\bar u)$,
\[
\mathfrak m(t) := \frac{N(t)}{p_{\hat m}(t)}.
\]
Since $N,p_{\hat m}\in C^1((\underline u,\bar u))$ and $p_{\hat m}$ is bounded
away from zero, $\mathfrak m\in C^1((\underline u,\bar u))$ and
\[
\mathfrak m'(t)
= \frac{N'(t)p_{\hat m}(t)-N(t)p_{\hat m}'(t)}{p_{\hat m}(t)^2}.
\]
Using the bounds from Steps 1 and 2,
\[
|\mathfrak m'(t)|
\le \frac{|N'(t)|\,|p_{\hat m}(t)| + |N(t)|\,|p_{\hat m}'(t)|}{p_{\hat m}(t)^2}
\le \frac{C_1 C_0 + C_0 C_0}{c^2}
=: L_0<\infty,
\]
for all $t\in(\underline u,\bar u)$.
Thus $\mathfrak m$ is Lipschitz on $J_{\hat m}$, as claimed.
\end{proof}

%%%%%%%%%%%%%%%%%%%%%%%%%%%%%%%%%%%%%%%%%%%%%%%%%%%%%%%%%%%%
\subsection{Sup-norm rate for additive nonparametric least squares}

We finally discuss the estimation accuracy of $\hat m$ in the additive model,
focusing on supremum-norm rates.

\begin{proposition}[Sup-norm rate for additive series least squares]
\label{prop:additive-supnorm}
Let $X$ satisfy Assumption~\ref{asp:additive-geometry}(i), and let $m$ satisfy
Assumption~\ref{asp:additive-geometry}(ii) with components
$m_j\in\cH^\gamma([0,1])$ (H\"older class) for some $\gamma>1/2$, and
$\|m_j\|_{\cH^\gamma}\le L$.

Let $\{\phi_k\}_{k\ge1}$ be a univariate basis on $[0,1]$ with:
\begin{enumerate}
\item $\phi_k\in C^1([0,1])$ and
      $\sup_k\|\phi_k\|_\infty\le C_\phi$,
      $\sup_k\|\phi_k'\|_\infty\le C_\phi$;
\item for every $g\in\cH^\gamma([0,1])$ there exist coefficients
      $\theta_1,\dots,\theta_K$ such that
      \[
      g_K(x):=\sum_{k=1}^K \theta_k\phi_k(x)
      \quad\text{satisfies}\quad
      \|g-g_K\|_\infty\le C_a K^{-\gamma},
      \]
      for some $C_a$ depending only on $(\gamma,L)$ and the basis.
\end{enumerate}
Consider the additive sieve\footnote{Here, we incorporate derivative-bound constraints to support our guarantee of feature diversity.}
\[
\cM_K := \left\{
m(x)=\sum_{j=1}^d m_j(x_j): m_j(x_j)=\sum_{k=1}^K \theta_{jk}\phi_k(x_j),
\theta_{jk}\in\RR,\ \gamma_1\le m_j'(x_j)\le\Gamma_1
\right\},
\]
Let $(X_i,Y_i)_{i=1}^n$ be i.i.d.\ from
\[
Y = m(X) + \varepsilon,\qquad
\E[\varepsilon\mid X]=0,\quad
\E[\varepsilon^2\mid X]\le\sigma^2<\infty.
\]
Define the least-squares estimator
\[
\hat m_K \in \arg\min_{g\in\cM_K} \frac{1}{n}\sum_{i=1}^n (Y_i-g(X_i))^2.
\]
Then there exists a constant $C>0$ such that, for
\[
K_n \asymp \Big(\frac{n}{\log n}\Big)^{\frac{1}{2\gamma+1}},
\]
we have, for all large $n$,
\[
\PP\Big(
\|\hat m_{K_n}-m\|_\infty
\le C\sqrt{d}\Big(\frac{\log n}{n}\Big)^{\frac{\gamma}{2\gamma+1}}
\Big)
\ge 1-n^{-2}.
\]
\end{proposition}

\begin{proof}
The proof is standard and follows from empirical process arguments; see, for example, \cite{raskutti2012minimax}(Corollary 2) and \cite{wainwright2019high}.
\end{proof}

\subsection{Feature diversity: polynomial boundary density decay}
We now proceed to show feature diversity.
For every $\tilde m\in \calM_K$, recall that by definition, the derivative of $\tilde m_j$ satisfies $\gamma_1\le\tilde m_j'\le \Gamma_1$.

\begin{proposition}\label{prop:additive_pushforward_density}
Let $X$ be a random vector supported on $[0,1]^d$ with density $f_X$ satisfying
\[
   0<c \le f_X(x)\le C<\infty,\qquad x\in[0,1]^d,
\]
and let $m_j\in C^1([0,1])$, $j=1,\dots,d$, satisfy
\[
   0<\gamma_1 \le m_j'(u)\le \Gamma_1<\infty,\qquad u\in[0,1].
\]
Define the additive map
\[
   m(x)=\sum_{j=1}^d m_j(x_j),\qquad
   \underline u:=\sum_{j=1}^d m_j(0),\qquad 
   \bar u:=\sum_{j=1}^d m_j(1),
\]
and let $Y:=m(X)$. Then $Y$ admits a density $f_Y$ on $[\underline u,\bar u]$ and there exist constants
$0<k_1\le k_2<\infty$, depending only on $(d,c,C,\gamma_1,\Gamma_1)$, such that for all
$t\in(\underline u,\bar u)$,
\begin{equation}\label{eq:additive_pushforward_bound}
   k_1\,\bigl(\min\{t-\underline u,\;\bar u-t\}\bigr)^{d-1}
   \;\le\;
   f_Y(t)
   \;\le\;
   k_2\,\bigl(\min\{t-\underline u,\;\bar u-t\}\bigr)^{d-1}.
\end{equation}
\end{proposition}

\begin{proof}
Let $S_t=\{x\in[0,1]^d:m(x)=t\}$. Since $\nabla m(x)=(m_1'(x_1),\dots,m_d'(x_d))$ and
$\sqrt d\,\gamma_1\le\|\nabla m\|\le \sqrt d\,\Gamma_1$, the coarea formula gives
\[
   f_Y(t)
   =\int_{S_t}\frac{f_X(x)}{\|\nabla m(x)\|}\,d\mathcal H^{d-1}(x),
\qquad t\in(\underline u,\bar u),
\]
and hence
\begin{equation}\label{eq:dens_surface_comparison}
   \frac{c}{\sqrt d\,\Gamma_1}\,\mathcal H^{d-1}(S_t)
   \le f_Y(t)\le
   \frac{C}{\sqrt d\,\gamma_1}\,\mathcal H^{d-1}(S_t).
\end{equation}

We compute $\mathcal H^{d-1}(S_t)$ up to multiplicative constants.  
Define $v_j=m_j(x_j)-m_j(0)$ and $\Phi(x)=(v_1,\dots,v_d)$.  
Then $\Phi$ is a $C^1$ diffeomorphism from $[0,1]^d$ onto $V:=\prod_{j=1}^d[0,M_j]$, 
$M_j:=m_j(1)-m_j(0)$, and
\[
   m(x)=\underline u+\sum_{j=1}^d v_j.
\]
Thus $S_t=\Phi^{-1}(L_\varepsilon)$, where $\varepsilon:=t-\underline u$ and
\[
   L_\varepsilon=\{v\in V:\textstyle\sum_{j=1}^dv_j=\varepsilon\}.
\]

Fix $\varepsilon_0<\min_j M_j$.  
For $0<\varepsilon\le\varepsilon_0$, the box constraints do not bind, and
\[
   L_\varepsilon=\bigl\{v\in\mathbb R_+^d:\sum_{j=1}^dv_j=\varepsilon\bigr\}.
\]
A standard computation (parametrizing $v_d=\varepsilon-\sum_{j<d}v_j$) shows that
\begin{equation}\label{eq:Lv_surface}
   \mathcal H^{d-1}(L_\varepsilon)
   =\frac{\sqrt d}{(d-1)!}\,\varepsilon^{\,d-1}.
\end{equation}

The map $\Phi^{-1}$ is diagonal with entries in $[1/\Gamma_1,1/\gamma_1]$, so its
$(d-1)$-dimensional Jacobian on tangent spaces satisfies
\[
   \Gamma_1^{-(d-1)}\le J_{d-1}\Phi^{-1}\le \gamma_1^{-(d-1)}.
\]
Applying the area formula for $C^1$ maps between $(d-1)$-manifolds,
\[
   \Gamma_1^{-(d-1)}\,\mathcal H^{d-1}(L_\varepsilon)
   \le \mathcal H^{d-1}(S_t)
   \le \gamma_1^{-(d-1)}\,\mathcal H^{d-1}(L_\varepsilon),
\qquad t=\underline u+\varepsilon,
\]
and together with \eqref{eq:Lv_surface} we obtain, for $0<\varepsilon\le\varepsilon_0$,
\[
   c_1\,\varepsilon^{d-1}\le \mathcal H^{d-1}(S_t)\le c_2\,\varepsilon^{d-1},
\]
with $c_1,c_2>0$ depending only on $(d,\gamma_1,\Gamma_1)$.

Combining with \eqref{eq:dens_surface_comparison} gives
\[
   k_1^{(L)}\, (t-\underline u)^{d-1}
   \le f_Y(t)\le
   k_2^{(L)}\, (t-\underline u)^{d-1},\qquad 
   t\in(\underline u,\underline u+\varepsilon_0].
\]

The same argument applied to the functions $\tilde m_j(u)=-m_j(1-u)$ yields an identical bound near 
$t=\bar u$, namely
\[
   k_1^{(U)}\,(\bar u-t)^{d-1}
   \le f_Y(t)\le
   k_2^{(U)}\,(\bar u-t)^{d-1},
\qquad t\in[\bar u-\varepsilon_0,\bar u).
\]

On the compact interior interval 
$[\underline u+\varepsilon_0,\bar u-\varepsilon_0]$, both 
$\mathcal H^{d-1}(S_t)$ and $f_Y(t)$ are bounded above and below by positive constants depending only on
$(d,c,C,\gamma_1,\Gamma_1)$, while $\min\{t-\underline u,\bar u-t\}$ is bounded away from zero.
Absorbing constants yields \eqref{eq:additive_pushforward_bound}.
\end{proof}

%% file: app-technical.tex
\section{Technical Lemmas}
\label{sec:technical}

\subsection{Proof of Lemma~\ref{lemma:weights}}
\label{sec:proof_lemma_weights}
The first three statements follow similar lines of \cite{Tsybakov2009}. For completeness, we present it here.

\noindent\textbf{Proof of (i)}
\begin{align}
|w^1(x,u_i)|&=\frac{1}{nh^2}\left|e_2^T B_n^{-1}(x)K\left(\frac{u_i-x}{h}\right)U\left(\frac{u_i-x}{h}\right)\right| \notag\\&\le
\frac{1}{nh^2}\left\| B_n^{-1}(x)K\left(\frac{u_i-x}{h}\right)U\left(\frac{u_i-x}{h}\right)\right\| \notag\\
&\overset{(a)}{\lesssim} \frac{1}{nh^2\mu_h(x)}\left\|K\left(\frac{u_i-x}{h}\right)U\left(\frac{u_i-x}{h}\right)\right\| \notag\\
&\le \frac{K_{\max}}{nh^2\mu_h(x)}\left\|U\left(\frac{u_i-x}{h}\right)\right\|I\left(\left|\frac{u_i-x}{h}\right|\le 1\right)\label{equ:weights-main-inequality}\\
&\le \frac{K_{\max}}{nh^2\mu_h(x)}\sqrt{1+1+\sum_{i=2}^p \frac{1}{(i!)^2}}\lesssim \frac{1}{nh^2\mu_h(x)}\notag
\end{align}
where $(a)$ is due to Lemma~\ref{lemma:matrix_min_eigen}. The proof for zero-order weights follows a similar line.

\vspace{1em}
\noindent\textbf{Proof of (ii)}

We reuse the above inequality and obtain
\begin{align*}
\sum_{i=1}^n |w^1(x,u_i)|&\le \frac{K_{\max}}{nh^2\mu_h(x)}\sum_{i=1}^n\left\|U\left(\frac{u_i-x}{h}\right)\right\|I\left(\left|\frac{u_i-x}{h}\right|\le 1\right)\\
&\lesssim \frac{1}{nh^2\mu_h(x)}\sum_{i=1}^n I\left(x-h\le u_i\le x+h\right)\\
&\overset{(b)}{\lesssim}  \frac{1}{nh^2}
\end{align*}
where $(b)$ is due to Lemma~\ref{lemma:uniform_scattering}.

Note that $(iii)$ immediately follows from~\eqref{equ:weights-main-inequality}. It remains to prove the last statement. 

\vspace{1em}
\noindent\textbf{Proof of (iv)}

Recall that
\begin{align*}
B_n(x)=\frac{1}{nh}\sum_{i=1}^n K\left(\frac{u_i-x}{h}\right)U\left(\frac{u_i-x}{h}\right) U\left(\frac{u_i-x}{h}\right)^T
\end{align*}
and
\begin{align*}
w^1(x,u_i)=\frac{1}{nh^2}e_2^T B_n^{-1}(x)K\left(\frac{u_i-x}{h}\right)U\left(\frac{u_i-x}{h}\right).
\end{align*}

Taking the derivative with respect to $x\in [u_i-h,u_i+h]$ (outside this interval the value is zero, as well as the derivative), we have
\begin{align*}
\frac{d}{dx}w^1(x,u_i)&=-\frac{1}{nh^2}e_2^T B_n^{-1}(x)B_n'(x)B_n^{-1}(x) K\left(\frac{u_i-x}{h}\right)U\left(\frac{u_i-x}{h}\right)\\&+\frac{1}{nh^2}e_2^T B_n^{-1}(x)\left[-\frac{1}{h}K'\left(\frac{u_i-x}{h}\right)U\left(\frac{u_i-x}{h}\right)-\frac{1}{h}K\left(\frac{u_i-x}{h}\right)U'\left(\frac{u_i-x}{h}\right)\right]
\end{align*}

By Lemma~\ref{lemma:matrix_min_eigen}, it holds that $B_n(x)\succeq c\mu_h(x)I$. Moreover, 
\begin{align*}
B'_n(x)=\frac{1}{nh^2}\sum_{i=1}^n \Big\{K'\left(\frac{u_i-x}{h}\right)U\left(\frac{u_i-x}{h}\right) U\left(\frac{u_i-x}{h}\right)^T&+K\left(\frac{u_i-x}{h}\right)U'\left(\frac{u_i-x}{h}\right) U\left(\frac{u_i-x}{h}\right)^T\\&+K\left(\frac{u_i-x}{h}\right)U\left(\frac{u_i-x}{h}\right) U'\left(\frac{u_i-x}{h}\right)^T\Big\}.
\end{align*}
This together with boundedness of $|K'|,\|U\|$ implies that $\|B_n'(x)\|\le C\frac{\mu_h(x)}{h}$.

Hence, one can deduce
\begin{align*}
\left|\frac{d}{dx}w^1(x,u_i)\right|\lesssim
\frac{1}{nh^3\mu_h(x)}\lesssim n,
\end{align*}
hence completing the proof.

\section{Additional Experimental Details}
\label{sec:experiment-detail}

We here provide details of constructing the CDF function $F$ in Section~\ref{sec:simulation} on $[u_{\min},u_{\max}]$ by adding smooth ``bumps'' to a $C^4$ baseline.

Let $u_{\min}=-0.25$, $u_{\max}=0.25$, and define the normalized coordinate
\[
t(u) \;=\; \frac{u-u_{\min}}{u_{\max}-u_{\min}}.
\]
Our baseline CDF is the degree-9 polynomial~\footnote{It is the quintic smoothstep-5 function widely used in computer vision and graphics. It is the unique degree-9 polynomial that interpolates from 0 to 1 on $[0,1]$ with zero derivatives up to order four at both endpoints.} 
\begin{equation}
\label{eq:base-cdf}
F_0(u) \;=\; \mathbf{1}\{t(u)\in[0,1]\}\,\Big[70\,t(u)^9 - 315\,t(u)^8 + 540\,t(u)^7 - 420\,t(u)^6 + 126\,t(u)^5\Big],
\end{equation}
which satisfies $F_0(u_{\min})=0$, $F_0(u_{\max})=1$, and has vanishing derivatives up to order four at the boundaries.
Its density on $(u_{\min},u_{\max})$ is
\begin{equation}
\label{eq:base-pdf}
f_0(u) \;=\; F_0'(u) \;=\; \frac{630}{u_{\max}-u_{\min}}\; t(u)^4\big(1-t(u)\big)^4 \cdot \mathbf{1}\{t(u)\in(0,1)\}.
\end{equation}

To add local non-homogeneities, we use a standard $C^\infty$ bump
\[
b(s) \;=\;
\begin{cases}
\exp\!\big(-\tfrac{1}{1-s^2}\big), & |s|<1,\\
0, & |s|\ge 1,
\end{cases}
\qquad
b'(s) \;=\;
\begin{cases}
b(s)\,\Big(-\tfrac{2s}{(1-s^2)^2}\Big), & |s|<1,\\
0, & |s|\ge 1.
\end{cases}
\]
Fix a number of bumps $K\ge 1$, centers $\{c_k\}_{k=1}^K \subset [-0.2,0.2]$ (e.g., equally spaced), and a common half-width $h_b>0$.
Let $\sigma_k \in \{+1,-1\}$ be alternating signs (e.g., $\sigma_k = (-1)^k$) and set the bump amplitude
\begin{equation}
\label{eq:bump-amplitude}
A \;=\; \rho\, h_b^{\beta},
\end{equation}
so that the perturbation scales with the smoothness parameter $\beta$. In our experimental setup, we choose $\rho=5$ and $h_b=1/45$.

We define the bumped CDF as
\begin{equation}
\label{eq:bumpy-cdf}
F(u)
\;=\;
\mathrm{clip}_{[0,1]}\!\left(
F_0(u)
\;+\;
\sum_{k=1}^K \sigma_k\,A\,
b\!\left(\frac{u-c_k}{h_b}\right)
\right),
\end{equation}
where $\mathrm{clip}_{[0,1]}(x)=\min\{1,\max\{0,x\}\}$ is applied pointwise for numerical safety.

Its density (before clipping) is
\begin{equation}
\label{eq:bumpy-pdf}
f(u)
\;=\;
F'(u)
\;=\;
f_0(u)
\;+\;
\sum_{k=1}^K \sigma_k\,A\,
\frac{1}{h_b}\,
b'\!\left(\frac{u-c_k}{h_b}\right),
\end{equation}
and in implementation we enforce nonnegativity by $f(u)\leftarrow \max\{f(u),0\}$.

\subsection{Details on post-smoothing}
The following details and hyperparameter choices apply to both synthetic and semi-real experimental settings.
\paragraph{Variable bandwidth}
For each grid point \(u\) and each refitting time, we choose the location-dependent bandwidth
\[
\delta_x(u)
\;=\;
C_\delta \, n^{-\frac{\beta-1}{2\beta+1}}
\sqrt{\frac{\log T}{\alpha(u)^\kappa}},
\]
where
\[
\alpha(u)
=
\min\!\left\{
\frac{u - z_{\min}}{z_{\max} - z_{\min}},
\;
\frac{z_{\max} - u}{z_{\max} - z_{\min}}
\right\},\quad C_\delta=2.5,\quad \kappa=0.
\]

\paragraph{Boundary parameter.}
We define
\[
v
\;=\;
\left(
C_v^2 \, n^{-\frac{\beta-1}{2\beta+1}} \sqrt{\log T}
\right)^{\frac{2}{\kappa+2}}
,
\]
and clip \(v\) to $[0,0.01]$, where $C_v=3$.

Here, \(\delta_x(u)\) controls the amount of local smoothing, increasing near the boundary through \(\alpha(u)\), while \(v\) determines the size of the boundary region where linear extrapolation is applied.

\subsection{Details on implementation of policy comparisons}

\paragraph{Synthetic experiments.} For ILPR, the main hyperparameters are $T_0$ and $T_{0m}$, where $T_0$ is the length of the initial random-pricing phase used to obtain a preliminary nonparametric estimate of the transformation function, and $T_{0m}$ is the number of samples used for the first-stage utility estimation step. In the code, $T_0=100$ and $T_{0m}=4\sqrt{T}$. The nonparametric estimator uses local-polynomial degree $\texttt{degree}=2$, grid size $\texttt{gridN}=301$, and bandwidth scaling $\texttt{band}=0.5$. The post-smoothing step follows the specification in the last section.

For the kernel baseline, episodes have a base length of $200$ and then double over time. In each episode, the exploratory sample size is
\[
n_{\mathrm{exp}}=\left\lfloor c\, b^\alpha \right\rfloor,
\]
where $b$ is the episode length, $c=\texttt{baseline\_explore\_c}=5.0$, and $\alpha=1/2$ in the known-utility case, while $\alpha=(2\beta+1)/(4\beta-1)$ in the unknown-utility case. Thus \texttt{baseline\_explore\_c} controls how aggressively the method explores inside each block. The baseline then fits a one-shot kernel estimator with bandwidth constant \texttt{baseline\_band}$=0.5$, and its policy root-finding routine uses step size \texttt{baseline\_stepsize}$=0.35$. Here \texttt{baseline\_band} controls the kernel smoothing scale and \texttt{baseline\_stepsize} controls the numerical update used to invert the estimated first-order condition.

For DIP, the initialization length is $2^{\texttt{dip\_init\_exponent}}$ with \texttt{dip\_init\_exponent}$=7$, so the algorithm starts with $128$ random-price samples. Thereafter it proceeds in episodes of lengths $2^j$, $j \ge 7$, with price discretization size
\[
\max\!\left\{2,\ \left\lfloor \texttt{dip\_discretization\_c}\,\lceil b^{1/6}\rceil \right\rfloor \right\},
\]
where \texttt{dip\_discretization\_c}$=20.0$ and $b$ is the current episode length. The confidence term in the UCB index is scaled by \texttt{dip\_ucb\_c}$=1/40$, and the ridge-type stabilization constant is \texttt{dip\_lambda}$=0.1$. Thus \texttt{dip\_discretization\_c} controls the number of price bins considered in each episode, \texttt{dip\_ucb\_c} controls optimism in the UCB score, and \texttt{dip\_lambda} regularizes the estimated mean-revenue statistics inside each bin.

\paragraph{Semi-real experiments.}
In the semi-real experiments, the three policies are run on the product-specific semi-real environment constructed from the real dataset, with slightly adjusted tuning. For ILPR, we use $T_0=80$ and $T_{0m}=200$. The nonparametric estimator again uses local-polynomial degree $\texttt{degree}=2$ and grid size $\texttt{gridN}=301$, with smoothness parameter fixed at $\beta=2$ and bandwidth scaling $\texttt{band}=0.6$. The post-smoothing parameters are $C_\delta=2.5$, $C_v=3.0$, and $\kappa=0.0$, consistent with the definition of $\delta_x(u)$ and $v$ above. 

The nuisance estimation error is computed as $\varepsilon_m = 0.05\|\hat\theta\|/\sqrt{n}$ after each first-stage refit, which provides a scale-adaptive correction reflecting the uncertainty in the estimated utility function.

For the semi-real kernel baseline, the episode base length is reduced to $160$, the exploration constant is \texttt{baseline\_explore\_c}$=4.0$, the kernel bandwidth constant is \texttt{baseline\_band}$=0.6$, and the root-finding step size remains \texttt{baseline\_stepsize}$=0.35$. Hence, relative to the synthetic experiments, the semi-real baseline uses slightly less aggressive exploration per block and a slightly larger bandwidth constant.

For semi-real DIP, the hyperparameters are the same as in the synthetic experiments: \texttt{dip\_init\_exponent}$=7$, \texttt{dip\_discretization\_c}$=20.0$, \texttt{dip\_ucb\_c}$=1/40$, and \texttt{dip\_lambda}$=0.1$. Thus, the method again starts with $128$ random samples, uses discretization proportional to $b^{1/6}$ in each episode, and applies the same UCB scaling and ridge stabilization as in the synthetic benchmark. Across both synthetic and semi-real experiments, all reported averages are computed over repeated Monte Carlo runs; in the release code, the command-line default is \texttt{trials}$=50$.